\documentclass[11pt,letterpaper]{article}

\usepackage{authblk} % for authors

\usepackage[margin=1in]{geometry} % for layout

\usepackage{graphicx} % Required for inserting images

\usepackage{amsthm}
\usepackage{amsfonts}
\usepackage{amsmath}
\usepackage{amssymb}

\usepackage{mathtools} % tools for math
\usepackage{physics}
\usepackage{dsfont} % fonts
\usepackage{bbold} %  what is this?

\usepackage[dvipsnames]{xcolor} % for colors

\usepackage{comment} % for comments

\allowdisplaybreaks

\usepackage{xspace}
\usepackage{url}

\usepackage{hyperref} % To add hyperlinks and choose their colors
\hypersetup{
  hidelinks,
  colorlinks = true,
  linkcolor = Red,
  citecolor = Red
}
\usepackage[nameinlink,capitalize]{cleveref} % To improve cross references

\usepackage{enumitem}

\usepackage{microtype} %to improve change line

\theoremstyle{plain}
\newtheorem{theorem}{Theorem}
\newtheorem{lemma}[theorem]{Lemma}
\newtheorem{claim}[theorem]{Claim}
\newtheorem{proposition}[theorem]{Proposition}
\newtheorem{corollary}[theorem]{Corollary}

\newtheorem*{lemma*}{Lemma}
\newtheorem*{theorem*}{Theorem}

\theoremstyle{definition}
\newtheorem{definition}[theorem]{Definition}

\AddToHook{env/lemma/begin}{\crefalias{theorem}{lemma}}
\AddToHook{env/claim/begin}{\crefalias{theorem}{claim}}
\AddToHook{env/proposition/begin}{\crefalias{theorem}{proposition}}
\AddToHook{env/corollary/begin}{\crefalias{theorem}{corollary}}
\AddToHook{env/definition/begin}{\crefalias{theorem}{definition}}

\providecommand{\citation}[1]{}

\newcommand{\MM}
{\mathcal{M}}

\newcommand{\conf}{\mathbf{c}}

\newcommand{\dejavu}{\textnormal{\textsc{DéjàVu}}\xspace}
\newcommand{\dejavuh}{\textnormal{\textsc{DéjàVu}}(\(h\))\xspace}

\newcommand{\hmaj}{\(h\)-\textnormal{\textsc{Majority}}\xspace}
\newcommand{\threemaj}{\(3\)-\textnormal{\textsc{Majority}}\xspace}
\newcommand{\twochoices}{\(2\)-\textnormal{\textsc{Choices}}\xspace}

\global\long\def\pull{\mathcal{PULL}}
\newcommand{\pullh}{\(\pull\)(\(h\))\xspace}
\newcommand{\pullstar}{\(\mathcal{PULL}^*\)\xspace}

\DeclareMathOperator{\bbN}{\mathbb{N}}

\DeclareMathOperator{\bbR}{\mathbb{R}}
\DeclareMathOperator{\bbE}{\mathbb{E}}

\newcommand{\cpa}[1]{\left\{ #1 \right\}} % curly parenthesis
\newcommand{\pa}[1]{\left( #1 \right)} % round brackets

\newcommand{\pr}[1]{\Pr\left(#1\right)}
\newcommand{\Var}{\mathrm{Var}} % Fra: I would add the argument here

\newcommand{\argmin}{\textrm{argmin}}

\definecolor{myOrange}{rgb}{0.8,0.4,0}

\author[1]{Francesco d'Amore}
\author[2]{Niccolò D'Archivio}
\author[3]{George Giakkoupis}
\author[2]{Frédéric Giroire}
\author[2]{Emanuele Natale} 

\affil[1]{Gran Sasso Science Institute, Italy}
\affil[2]{INRIA, COATI, Université Côte d'Azur, France}
\affil[3]{INRIA Rennes}

\title{DéjàVu: A Minimalistic Mechanism for Distributed Plurality Consensus}
\date{}

\begin{document}

\maketitle
\begin{abstract}
  We study the plurality consensus problem in distributed systems where a population of extremely simple agents, each initially holding one of \(k\) opinions, aims to agree on the initially most frequent one.
  In this setting, \hmaj is arguably the simplest and most studied protocol, in which each agent samples the opinion of \(h\) neighbors uniformly at random and updates its opinion to the most frequent value in the sample.

  We propose a new, extremely simple mechanism called \dejavu{}: an agent queries neighbors until it encounters an opinion for the second time, at which point it updates its own opinion to the duplicate value. This rule does not require agents to maintain counters or estimate frequencies, nor to choose any parameter (such as a sample size \(h\)); it relies solely on the primitive ability to detect repetition.
  We provide a rigorous analysis of \dejavu that relies on several technical ideas of independent interest and demonstrates that it is competitive with \hmaj and, in some regimes, substantially more communication-efficient, thus yielding a powerful primitive for plurality consensus. 
\end{abstract}
%\tableofcontents
\section{Introduction}

Plurality consensus is a fundamental problem in distributed computing and multi-agent systems, where a collection of agents, each initially holding one of \(k\) possible opinions, seeks to agree on the initially most frequent opinion~\cite{becchettiConsensusDynamicsOverview2020}. This problem serves as a building block for various distributed tasks, with fundamental applications that range from coordination in swarm robotics to modeling collective behavior in biology~\cite{FraigniaudN19}. The core challenge in these scenarios lies in achieving consensus rapidly and reliably using agents with limited memory, computational power, and local information, often in the absence of a central coordinator or global identifiers~\cite{shahGossipAlgorithms2007}.

A standard approach to plurality consensus is the \hmaj dynamics, in which an agent queries \(h\) random neighbors and updates its opinion to the majority among them (if one exists)~\cite{becchettiSimpleDynamicsPlurality2017}.
The underlying communication model is the \pullh model where, at each discrete-time round, each agent can observe the opinion of a sample of agents of size \(h\), sampled independently and uniformly at random~\cite{demersEpidemicAlgorithmsReplicated1987}.
While \hmaj is effective and has been extensively analyzed~\cite{damoreHMajorityDynamicsMany2025,CooperMRSS25}, it requires agents to query a fixed number of neighbors and perform a count-based comparison. This implies a need for explicit counting capabilities and knowledge of the parameter \(h\). Motivated by the quest for minimal computational assumptions, relevant for molecular computing, nanorobotics, or biological modeling~\cite{feinermanBreatheSpeakingEfficient2017,FuggerNR24}, we ask: is it possible to achieve efficient plurality consensus without explicit counting?

We answer this question affirmatively by proposing and analyzing \dejavu, a minimal protocol based solely on repetition detection. In the \dejavu protocol, an agent does not count or collect a fixed number of samples. Instead, it sequentially queries neighbors and stops as soon as it sees the same opinion for the second time. It then adopts this ``duplicate'' opinion as its new state. This rule relies exclusively on the ability to recognize a previously seen value within a short sampling window, a primitive operation that is simpler than arithmetic counting.

Our main contribution is a rigorous analysis of \dejavu that demonstrates its efficiency and robustness. We prove that, starting from a configuration with sufficient bias toward a plurality opinion, the system converges to consensus on that opinion with high probability. Our analysis shows that \dejavu acts as a powerful amplifier of plurality bias. Moreover, \dejavu is not only competitive with the \(h\)-majority rule, but in some regimes can be more communication-efficient. Specifically, we upper bound with high probability the total number of samples required by \dejavu until consensus, showing that this quantity adapts naturally to the distribution of opinions and can be smaller than that of a fixed-\(h\) rule when the plurality opinion has large enough support.
We present a more detailed overview of our results in \Cref{ssec:contribution}.

Our analysis relies on several technical contributions of independent interest. First, we establish an exact equivalence between the \dejavu dynamics and Poisson clock races, providing a robust framework for analyzing sampling-based stopping times. Second, we prove the bias-amplification property of the protocol by leveraging Newton's inequalities on elementary symmetric sums of the opinion frequencies. This allows us to show that the probability ratio of adopting the plurality opinion is monotone in the sample size. Finally, we characterize the protocol's communication complexity as a function of the \(\ell_2\)-norm of the opinion distribution through a generalized birthday paradox analysis, demonstrating its inherent adaptivity to the system's state.
A more detailed overview of our technical contributions is given in \cref{ssec:overview}.

The protocol's connection to Poisson races, a framework often used to model decision-making in neural systems~\cite{townsendStochasticModelingElementary1983}, suggests that \dejavu may also serve as a plausible model for biological consensus, where exact counting is cognitively expensive~\cite{dasguptaNeuralTheoryCounting2022}. However, our primary focus is its effectiveness as a distributed algorithm. By replacing fixed sample sizes with a dynamic stopping condition, \dejavu offers a novel design principle for distributed consensus that prioritizes agent simplicity and communication efficiency.

\subsection{Our Contribution}
\label{ssec:contribution}

Our main result is a high-probability bound on the convergence time of \dejavu in the \pullh model. Here and throughout the paper, w.h.p.\ means with probability at least \(1-1/n^c\) for some constant \(c>0\), where \(n\) is the number of nodes.
In the following, given any round \(t \ge 0 \), we denote by \(C^{(t)} = (C_1^{(t)}, \ldots, C_k^{(t)})\) the \textit{configuration of the system} at round \(t\), that is, \(C_i^{(t)}\) denotes the number of nodes supporting opinion \(i\) at time \(t\).
We omit the dependence on \(t\) when it is clear from the context.
%(\dejavuh, for short).

\begin{theorem}
  \label{thm:convergence_time_dejavu}
  Let \(h\geq 2\) and \(C=(C_1,\ldots,C_k)\) be an initial system configuration where each agent supports an opinion in \(\{1,\ldots,k\}\), with \(C_1\geq C_2 \geq \cdots \geq C_k\).
  Assume that \(C_1 = \omega\pa{\log n}\) and that, for a large enough constant \(\lambda > 0\),
  \[
    C_1 - C_2 \ge \lambda \sqrt{ \max\cpa{\frac{n}{h^2}, C_1} \log n } .
  \]
  \dejavu converges to consensus on the first opinion w.h.p.\ in \(O\pa{ \qty(\frac{n}{(h^2 C_1)}+1) \log n}\) rounds.
\end{theorem}

We emphasize that the hypotheses of the previous theorem are very general compared with the state of the art for \hmaj. The condition \(C_1 = \Omega\pa{\log n}\) is necessary for any high-probability guarantee. Moreover, when \(C_1 \geq n/h^2\), the required bias is essentially optimal, since it matches the scale of the standard deviation. This is the case, for instance, as soon as \(h \geq \sqrt{k}\). Furthermore, the convergence time is essentially optimal, as we discuss below.

We remark that a general upper bound on the convergence time of \hmaj matching the \(\Omega(k/h^2)\) lower bound shown in~\cite{becchettiSimpleDynamicsPlurality2017} %(originally appeared at SPAA'14)
is still an open problem, with ongoing recent progress~\cite{damoreHMajorityDynamicsMany2025}.
Our next result is a generalization of the previous lower bound, which allows for a more general comparison of \dejavu with \hmaj.

  \begin{theorem}[Generalization of Theorem 4.12 in~\cite{becchettiSimpleDynamicsPlurality2017}]
  \label{thm:app:lb:h-majority}
  Let \(\varepsilon > 0\) be any arbitrarily small constant and \(C=(C_1,\ldots,C_k)\) be the starting system configuration, with \(C_1 \ge \ldots \ge C_k\) and \(C_1 \le n/100\).
  For \(h =\Omega( {n^{3/4+\varepsilon}}/{C_1})\), w.h.p., \hmaj requires at least \(\Omega\pa{ {n}/{(h^2 C_1)}+1}\) rounds to reach consensus.
\end{theorem}

Thus, the convergence times of \dejavu and \hmaj match over a wide range of configurations.

Our next theorem compares the number of samples required by \dejavu and \hmaj until consensus, and shows that \dejavu is more sample-efficient over a large range of configurations. In fact, we conjecture that \dejavu is always more sample-efficient than \hmaj.

\begin{theorem}
  \label{thm:sample-efficiency}
  Let \(C=(C_1,\ldots,C_k)\) be a system configuration such that \(C_1\geq \cdots \geq C_k\), \(C_1 = \omega\pa{\log^2 n}\), and that, for a large enough constant \(\lambda > 0\),
  \[
    C_1 - C_2 \ge \lambda \sqrt{ \max\cpa{\frac{n}{h^2}, C_1} \log n } .
  \]
  Let \(S_{d}\) and \(S_{m}\) be the numbers of samples until consensus of, respectively, \dejavu and \hmaj.
  Fix any arbitrarily small constant \(\varepsilon > 0\).
  For \(h = \Omega({\min\{n^{3/4+\varepsilon}/{C_1}, \sqrt{n/C_1}\}})\), w.h.p. we have
  \begin{align*}
      \begin{cases}
          & S_{d} \cdot \frac{O\left(\max\{1,h\frac{\norm{C}_2}{n}\right)\}}{\log^3 n}\leq S_{m} \text{ if } \norm{C}_2 = O(\sqrt{n} \log n) \text{ and } h \norm{C}_2 \ge \frac{n}{\log n}, \\
          & S_{d} \cdot \frac{O\left(\max\{1,h\frac{\norm{C}_2}{n}\right)\}}{\log n}\leq S_{m} \text{ otherwise}.
      \end{cases}
  \end{align*}
\end{theorem}
As soon as \(h\norm{C}_2 \gg n\log^3 n\) and \(C_1 - C_2 = \omega(\sqrt{C_1 \log n})\), it is guaranteed that \(S_{d} < S_{m}\).
In particular, if the lower bound on \hmaj does not apply, a node running \hmaj still samples at least \(h\) opinions in the first round. In this regime, \cref{thm:sample-efficiency} shows that the average number of samples per node before convergence in \dejavu is either competitive with or strictly smaller than that of \hmaj.

\paragraph{Roadmap.}
The rest of the paper is organized as follows. In \cref{ssec:overview}, we provide an overview of the main technical ideas behind the proofs of our results, and in \cref{sec:related} we discuss related work. In \cref{sec:preliminaries}, we introduce the model and the notation used throughout the paper.

\Cref{sec:amplification_bias_PULL*,sec:star_to_h,sec:amplification} contain the analysis of the bias amplification mechanism of \dejavu, leading to the proof of \cref{thm:convergence_time_dejavu}. The lower bound for \hmaj (\cref{thm:app:lb:h-majority}) is proved in \cref{sec:lower_bound_h_majority}, and the sample-efficiency result (\cref{thm:sample-efficiency}) is proved in \cref{sec:samplesdejavu}. We conclude with open questions in \cref{sec:conclusions}.

\subsection{Main Technical Ideas}
\label{ssec:overview}

In this section we highlight the original technical ideas used in the proof of our main theorems.

In the following, let \(C = (C_1, \dots, C_k)\) be the configuration of the system at a given time, where \(C_i\) is the number of nodes supporting opinion \(i\), and let \(p_i = C_i/n\) be the corresponding density.
Let \(p = (p_1, \ldots, p_k)\) denote the vector of opinion densities.
We assume without loss of generality that \(C_1 \ge \ldots \ge C_k\).
For every \(i \in [k]\), let \(C_i'\) be the random variable counting the number of nodes supporting opinion \(i\) at the next round.
Let \(\MM_i\) be the event that an agent updates to opinion \(i\) when the number of samples is unbounded.

A key idea in the proof of \Cref{thm:convergence_time_dejavu} is to study the ratio
\[
  \frac{\pr{\MM_1 \mid C}}{\pr{\MM_2\mid C}},
\]
namely, the ratio between the probability of updating to the plurality opinion and the probability of updating to the second most frequent opinion.

\paragraph{Poisson race in the \pullstar model for bias amplification.}
To bound the aforementioned ratio, we couple the \dejavu process to a continuous-time process, inspired by the technique of Poisson approximation for Balls-into-Bins processes~\cite{mitzenmacher2005}.
This equivalence is given in \Cref{sec:amplification_bias_PULL*}, and the resulting continuous-time process turns out to be an instance of a so-called Poisson race problem~\cite{ruan2007poisson}, where we are required to estimate the probability that a certain Poisson clock is the first to ring for the second time.
Our result is also new in that context and of independent interest.
Such estimation, combined with a way to decompose the probability ratio given in \Cref{sec:amplification_bias_PULL*}, allows us to prove that
\[
  \pa{\frac{p_i}{p_j}}^2 \ge \frac{\pr{\MM_i\mid C}}{\pr{\MM_j\mid C}} \geq \pa{\frac{p_i}{p_j}}^2 \frac{p_i+3p_j}{3p_i+p_j},
\]
for each \(i \le j\) (see \cref{lemma:lower_upper_ratio_probability}).

The previous inequality is derived in a model in which an agent can collect arbitrarily many samples, which we denote by \pullstar. In the \pullh model, where \dejavuh can collect at most \(h\) samples, many agents do not see a repeated opinion and therefore keep their current opinion.

\paragraph{From \pullstar to \pullh.}
In order to relate the two models, we thus need to estimate the probability ratio when we condition on the event that an agent sees an opinion twice within its \(h\) samples, and to estimate how many agents will actually update.
The first part is given in \cref{sec:star_to_h}, where we leverage Newton's inequalities for symmetric polynomials (\Cref{lem:newton_ineq}) to prove the following key result (formally stated in \cref{lemma:monotonicity_samples_dejavu}).
Let \(H\) be the number of samples until an agent samples an opinion a second time. Then, for all opinions \(i,j\in[k]\) such that \(p_i\ge p_j\), and for all \(h = 2,\ldots, k+1\), we have
  \[
    \pa{\frac{p_i}{p_j}}^2 \ge \frac{\pr{\MM_i, H\leq h\mid C}}{\pr{\MM_j, H\leq h\mid C}} \geq \frac{\pr{\MM_i\mid C}}{\pr{\MM_j\mid C}}~.
  \]
In other words, truncating the sample size at \(h\) can only improve the ratio of the winning probabilities.

The second part is to estimate the number of agents that will actually update by sampling twice an opinion. We do so by proving upper and lower bounds on the probability that a repeated opinion appears within the first \(h\) samples, namely a generalized birthday paradox for a non-uniform distribution. The lower bound relies on a Chen-Stein estimate due to Arratia et al.~\cite{arratia1989}.
We remark that the aforementioned question can be viewed as a generalized birthday paradox over a non-uniform distribution, a fundamental problem that is of independent interest~\cite{Galbraith2012ANB}.
We obtain the following lemma, which we prove in \cref{sec:sample_size_dejavu} (\cref{lemma:tight_bound_probability_collision}).
Let \(D\) be the number of agents that see an opinion twice within the first \(h\) samples, for any given \(h \ge 2\).
Then,
\[
    \bbE[D \mid C] =
    \begin{cases}
      n \cdot \Theta\pa{ \min\{h^2 \| p \|_2^2, 1\} } & \text{if }  h \| p \|_2 = o(1),
      \\
      \Theta(n)                                       & \text{otherwise.}
    \end{cases}
\]

\paragraph{Expected amplification of the bias in \pullh.}

In \cref{sec:expected-amplification}, we combine previous results and algebraic manipulations to get a lower bound on the amplification of the multiplicative bias in expectation:
\[
  \frac{\bbE[C'_1 \mid C] }{\bbE[C'_i \mid C] } \geq \frac{C_1}{C_i} + \Omega\pa{ \min\cpa{\frac{C_1}{n} h^2, 1}}\pa{ \frac{C_1}{C_i} - 1 }~.
\]
We turn the previous inequality into an expected additive amplification for the bias in \cref{sec:expected-additive-amplification}, as follows.
Let the current and next bias be \(\Delta_i = C_1 - C_i\) and \(\Delta_i' = C_1' - C_i'\), respectively.
Then, it holds that
\[
    \bbE[\Delta_j' \mid C] \ge \Delta_j\left(1 + \Omega\pa{ \min\cpa{\frac{C_1}{n} h^2, 1}}\right).
\]

\paragraph{Amplification of the bias in concentration.}
We use Bernstein's inequality to obtain concentration around the preceding estimate and to show that the bias-growth condition is preserved from one round to the next.
In \cref{lemma:bias_concentration_bernstein} we prove that
\[
    \Pr\left(
      \Delta_i' \ge \Delta_i \left(1 + \Omega\left(\min\left\{\frac{C_1}{n}h^2, 1\right\}\right)\right)
      \mid C
    \right) \ge 1 - n^{-\Theta(1)},
\]
whenever \(C_1 = \Omega(\log n)\), \(C_1 \le 3n/4\), and
\[
    \Delta_i \ge \lambda \sqrt{\max\left\{\frac{n}{h^2}, C_1 \right\} \log n},
\]
for a sufficiently large constant \(\lambda > 0\).
The proof splits into two regimes, depending on whether \(nC_1\) is at most or at least a constant multiple of \(\|C\|_2^2\). In the genuinely unbalanced regime, we prove stronger expectation bounds both for the bias and for the plurality opinion, which compensate for the larger concentration error.
Combining these ingredients, we show that the bias condition can be iterated round by round until the plurality opinion exceeds \(3n/4\). More precisely, we prove that \(C_1\) exceeds \(3n/4\) within
\[
  O\pa{\left(\frac{n}{h^2 C_1} + 1\right)\log n}
\]
rounds w.h.p.
Once this threshold is reached, we merge all non-plurality opinions into a single competing opinion and prove, via a coupling argument, that this binary process stochastically dominates the remaining consensus time of the original process. In the binary setting, \dejavu coincides with \twochoices when \(h=2\), and with \threemaj when \(h>2\), so the remaining time to consensus is \(O(\log n)\) w.h.p.

\paragraph{Lower bound for \hmaj.}
There is a known lower bound of \(\Omega(k/h^2 + 1)\) rounds for the \hmaj dynamics when \(h = \Omega(k/ n^{1/4 - \varepsilon})\) by~\cite{becchettiSimpleDynamicsPlurality2017}, which holds w.h.p.
In \cref{sec:lower_bound_h_majority} we take inspiration from the proof technique of~\cite{becchettiConsensusDynamicsOverview2020} and generalize the lower bound to \(\Omega(n/(C_1h^2) + 1)\) rounds for the \hmaj dynamics when \(h = \Omega(n^{3/4 + \varepsilon}/C_1)\).
The argument essentially consists in showing that any opinion cannot grow faster than a multiplicative factor \(1 + h^2 C_1/n\) every round w.h.p.\ with the right initial conditions.

\paragraph{Number of samples.}
In \cref{sec:samplesdejavu} we study the number of samples required by \dejavu to converge in \pullh. In this part, we assume \(C_1 = \omega(\log^2 n)\), as required by our sample-complexity theorem.
First, in \cref{sec:norm-submartingale} we prove that, in all majority-boosting opinion dynamics, the 2-norm of the configuration is a submartingale.
We say that a dynamics is majority boosting whenever \(C_i\geq C_j>0\) implies \(\frac{\bbE\pa{C'_i \mid C}}{\bbE\pa{C'_j \mid C}} \geq \frac{C_i}{C_j}\).
Inspired by~\cite{Shimizu2025}, which studies the evolution of the configuration norm for both the \threemaj and the \twochoices dynamics, we then use a one-sided Bernstein-type inequality (Freedman's inequality) to show that, with high probability, the 2-norm of the configuration does not decrease significantly throughout the execution of \dejavu (see \cref{sec:structure_from_shimizu}).
More specifically, if \(C^{(t)}\) denotes the configuration of the system at time \(t \ge 0\),
we show that \(\norm{C^{(t)}}_2^2 = \Omega({ \norm{C^{(0)}}_2^2 })\) for all rounds \(t\) up to consensus time, from which our stated sample-complexity bound follows.
Finally, in \cref{sec:comparison-samples} we leverage the lower bound on the consensus time for the \hmaj dynamics to compare the number of samples needed for consensus by \dejavu with that needed by the \hmaj dynamics.

\section{Related Work}
\label{sec:related}

The study of \dejavu falls within the broader area of opinion dynamics, and more generally simple computational dynamics~\cite{mosselOpinionExchangeDynamics2017,becchettiConsensusDynamicsOverview2020}.
Informally, these are synchronous consensus protocols based on simple update rules that do not change over time.
Well-studied examples include the \hmaj dynamics, the {\normalfont\scshape Undecided-State} dynamics, and the \twochoices dynamics.
In this section, we summarize the results most closely related to our contribution.

As throughout the paper, whenever we refer to a configuration \(C=(C_1,\ldots,C_k)\), we assume without loss of generality that \(C_1 \ge \cdots \ge C_k\). In particular, \(C_1\) denotes the plurality opinion, and the additive bias of the configuration is at least \(C_1-C_2\). Unless otherwise specified, all statements in this section hold w.h.p.
We stress that our focus is on \emph{plurality consensus}, i.e., convergence to the initially most supported opinion, rather than consensus to an arbitrary opinion. Accordingly, when summarizing prior work, we distinguish between results that guarantee plurality consensus and those that only guarantee consensus to some opinion.

% \subsection{\texorpdfstring{The \hmaj and the \twochoices dynamics}{The h-majority and the 2-choices dynamics}}

The closest works to our contribution are those analyzing the \hmaj dynamics, which has been widely investigated in the distributed computing community~\cite{becchettiSimpleDynamicsPlurality2017,BerenbrinkCEKMN17,BerenbrinkCGHKR22,GhaffariL18,CooperMRSS25,Shimizu2025,damoreHMajorityDynamicsMany2025}.
Most previous works analyzed the \threemaj dynamics, that is, the \hmaj when \(h = 3\).
The first work providing bounds on the \threemaj dynamics was~\cite{becchettiSimpleDynamicsPlurality2017} (presented at SPAA'14), which established an \(O(\lambda \log n)\) upper bound on the convergence time, w.h.p., provided that \(C_1\geq n/\lambda\) and that the bias is \(\Omega(\sqrt{\lambda n \log n})\). They further showed that \hmaj cannot converge in less than \(\Omega(n/h^2)\) rounds from certain configurations.
Subsequently,~\cite{BecchettiCNPT16} established an upper bound of \(O\pa{(k^2 \sqrt{\log n} + k \log n)(k + \log n)}\) rounds to reach consensus that holds w.h.p.,
with the hypothesis that the number of opinions \(k\) initially present satisfies \(k \le n^{\alpha}\) for a suitable positive constant \(\alpha < 1\).
Later, the bound was improved in~\cite{becchettiSimpleDynamicsPlurality2017}, which showed an upper bound of \(O(\min\{k, (n/\log n)^{1/3}\} \log n)\) rounds that holds w.h.p., provided that the bias of the initial configuration is at least \(c \sqrt{\min\{2k, (n/\log n)^{1/3}\} n \log n}\) for some constant \(c > 0\).
In the same work, the authors also proved a lower bound of \(\Omega(k \log n)\) rounds to reach consensus w.h.p.\ when the initial configuration is almost balanced, namely when \(C_1 \le n/k + (n/k)^{1-\varepsilon}\) for some \(\varepsilon > 0\) and \(k \le (n/\log n)^{1/4}\).

The \threemaj dynamics is closely related to another popular process, the \twochoices dynamics, defined as follows: each agent samples two neighbors u.a.r.\ with repetition, observes their opinions, and adopts that opinion if the two samples agree; otherwise, it keeps its current opinion.
It can be viewed as a variant of \threemaj in which one of the three opinions is the agent's current opinion, so that ties are broken in favor of the current state.
Despite being very similar, it has been shown that the two dynamics exhibit different behaviors.
\cite{BerenbrinkCEKMN17} proved a generic lower bound of \(\Omega(\min\{k, n/\log n\})\) rounds to reach consensus starting from the initial perfectly balanced configuration that holds w.h.p.\ for the \twochoices dynamics.
Furthermore, they proved that the \threemaj  dynamics works better in symmetric configurations (i.e., with no initial bias) when, e.g., \(\max_{i \in [k]}\{\conf_0(i)\} = O(\log n)\).
In particular, the \threemaj reaches consensus in time at most \(O(n^{3/4} \log^{7/8} n)\) w.h.p., regardless of further assumptions on the initial configuration, whereas the \twochoices requires time \(\Omega(n / \log n)\) whenever \(C_1 = O(\log n)\).
The authors of \cite{BerenbrinkCEKMN17} were the first to notice that, when the number of opinions \(k\) is large, the \threemaj dynamics is polynomially (in \(k\)) faster than the \twochoices dynamics.

The work~\cite{GhaffariL18} improved upon~\cite{becchettiSimpleDynamicsPlurality2017} and showed that the convergence time to consensus is \(O(k \log n)\), with high probability, for both the \twochoices dynamics with \(k = O(\sqrt{n / \log n})\) and the \threemaj dynamics with \(k = O(n^{1/3}/\sqrt{\log n})\) opinions.
This upper bound is tight because it matches the lower bound by \cite{becchettiSimpleDynamicsPlurality2017}, at least as long as \(k \le (n/\log n)^{1/4}\).
Furthermore, the authors showed that the unconditional convergence time of the \threemaj dynamics is \(O(n^{2/3} \log^{3/2} n)\) w.h.p., without any further hypothesis.

A more recent work~\cite{Shimizu2025} provided the tightest analysis of both the \threemaj and the \twochoices dynamics.
The authors proved that, w.h.p., the \threemaj dynamics reaches consensus in \(O(k\log n)\) rounds if \(k = o(\sqrt{n}/\log n)\), while it takes time \(O(\sqrt{n} \log^2 n )\) for other values of \(k\).
Furthermore, they showed that plurality consensus is ensured w.h.p.\ as long as the initial bias is \(\omega(\sqrt{n \log n})\).
As for the \twochoices dynamics, they proved that, w.h.p., it reaches consensus in \(O(k\log n)\) rounds if \(k = o(n/\log^2 n)\), while it takes time \(O(n \log^3 n)\) otherwise.
In this case, plurality consensus is ensured w.h.p.\ as long as the initial bias is \( \omega(\sqrt{C_1 \log n})\),
which matched the lower bound given by \cite{BerenbrinkCEKMN17} up to logarithmic factors.

As for the asynchronous setting, \cite{BerenbrinkCGHKR22} showed that the dynamics converges in \(O(n \log n)\) rounds w.h.p., when \(k = 2\).
A more general result was given in \cite{CooperMRSS25}, which proved that the convergence time is \(O(\min\{kn\log^2 n, n^{3/2} \log^{3/2} n\})\), w.h.p., for any number of initial opinions.
\cite{CooperMRSS25} also provided a generic lower bound of \(\Omega(\min\{kn, n^{3/2}/\sqrt{\log n}\})\) rounds to reach consensus that holds w.h.p. when the initial configuration is almost-balanced.

Other works analyzed the \threemaj dynamics in settings in which communication can be corrupted by some form of noise, which tries to capture the instability of real-world environments \cite{DAmoreZ22,dAmoreZ25}, while others analyzed the process when one opinion is preferred, in the sense that there is some probability that an agent spontaneously adopts the preferred opinion \cite{lesfariBiasedMajorityOpinion2022,CrucianiMQR23}.

% \paragraph{General case with \texorpdfstring{\(h \gg 1\).}{General case with h >> 1.}}
The regime \(h \gg 1\) is much less understood.
The authors of \cite{BerenbrinkCGHKR22} showed that, when \(k = 2\), the \hmaj exhibits a probabilistic hierarchy:
for any given \(t\), the probability that the \hmaj converges to consensus within time \(t\) is smaller than that of the \((h+1)\)-majority dynamics.
Whether the hierarchy holds for the general case with \(k > 2\) is still open.
For large $h$, the work \cite{becchettiSimpleDynamicsPlurality2017} provided a lower bound of \(\Omega(k / h^2 + 1)\) rounds to reach consensus that holds w.h.p.
The only matching upper bound in the literature was recently provided by \cite{damoreHMajorityDynamicsMany2025}, which showed that \hmaj converges in \(O(\log n)\) rounds whenever \(h = \omega(n \log n / C_1)\), \(C_1 = \omega(\log n)\), and the initial bias is \(\omega(\sqrt{C_1})\).
This result showed that the lower bound of \(\Omega(k / h^2 + 1)\) rounds to reach consensus that holds w.h.p.\ by \cite{becchettiSimpleDynamicsPlurality2017} cannot be pushed further than \(\Omega(k \log^2 n  /h + 1)\) in the worst case.
The general case with arbitrary \(h\) and no initial bias is a major open question in the area.

Before providing a direct comparison of \dejavu with the \hmaj and the \twochoices dynamics, we briefly summarize the results on the {\normalfont\scshape Undecided-State} dynamics.
In the {\normalfont\scshape Undecided-State} dynamics, at each round, each agent pulls a single neighboring opinion \(x\) uniformly at random.
If the agent’s former opinion \(y\) differs from \(x\), the agent becomes undecided.
Once undecided, the agent adopts the next opinion it encounters.
It was first introduced by \cite{AngluinAE08}, and then multiple papers analyzed its behavior \cite{perron09,ClementiGGNPS18,AmirABBHKL23,BankhamerBBEHKK22,hayek2025,BecchettiCNPS15}, even in the presence of noisy communication \cite{DAmoreCN20,DAmoreCN22} or stubborn agents \cite{berenbrink2024}.
We do not provide a full overview of the literature on the {\normalfont\scshape Undecided-State} dynamics, but we mention that in the synchronous setting~\cite{BecchettiCNPS15} proved convergence in time \(O(k \log n)\) whenever \(k = O((n \log n)^{1/3})\), w.h.p.
In \cite{AmirABBHKL23}, the authors investigate the asynchronous setting and prove that the protocol converges to consensus in \(O(k n \log n)\) rounds, w.h.p., whenever \(k \le \sqrt{n} \log^2 n\).
These results are tight as \cite{hayek2025} proved a lower bound for the asynchronous setting: they showed that the protocol takes at least \(\Omega(k n \log n)\) rounds, w.h.p., even with large bias allowed, when \(k = o(\sqrt{n}/\log n)\).
Several regimes remain open, especially in the synchronous setting. At the current state of the art, the {\normalfont\scshape Undecided-State} dynamics performs similarly to \twochoices, except that the bias required for convergence is always at least \(\Omega(\sqrt{n \log n})\). For this reason, a separate comparison with \dejavu would add little here.

\subsection{\texorpdfstring{Comparison of \dejavu with \hmaj}{Comparison of DéjàVu with h-majority}}

% \paragraph{Convergence time and bias.}
The \hmaj dynamics is known to converge to plurality consensus in time \(O(\log n)\) whenever the initial additive bias is at least \(\omega(\sqrt{C_1})\) and \(h = \omega(n \log n / C_1)\), which becomes \(h = \omega( k \log n)\) in almost-balanced configurations with \(k\) opinions \cite{damoreHMajorityDynamicsMany2025}.
For arbitrary \(k\) and \(4 \le h = O(n \log n / C_1)\), we do not have upper bounds yet.
However, \cite{becchettiSimpleDynamicsPlurality2017} provided a lower bound of \(\Omega(k/h^2 + 1)\) rounds, provided that \(h = \Omega( k/ n^{1/4 - \varepsilon})\) for any arbitrarily small constant \(\varepsilon > 0\).
In this work, we generalize this lower bound (\cref{thm:app:lb:h-majority}).
For \(h = 3\), the convergence time of \hmaj is \(O(k \log n)\) when \(k = o(\sqrt{n} / \log n)\), and \(O(\sqrt{n} \log^2 n)\) otherwise.
Moreover, plurality consensus is ensured when the initial additive bias is at least \(\omega(\sqrt{n\log n})\) \cite{Shimizu2025}.

Our upper bound on the convergence time of \dejavu (\cref{thm:convergence_time_dejavu}) matches that of the \hmaj at least in the studied regime \(h \gg 1\), and almost matches our lower bound of \(\Omega(n / (h^2 C_1) + 1)\) rounds needed by the \hmaj to converge when \(h = \Omega(n^{3/4 + \varepsilon}/C_1)\).
When \(h = 3\), the convergence time becomes \(O((n / C_1) \log n)\), matching that of the \threemaj when \(C_1 = \Omega(\sqrt{n})\).
For smaller values of \(C_1\), the \threemaj converges faster than our upper bound, but does not guarantee plurality consensus.
We emphasize that the scope of this work is plurality consensus; general consensus is left for future work (see also \cref{sec:conclusions}).
Note that the bias we require for plurality consensus is always competitive with the state of the art required by the \hmaj dynamics: we lose at most a \(\sqrt{\log n}\) multiplicative factor.

\subsubsection{\texorpdfstring{Comparison of \dejavu with \twochoices}{Comparison of DéjàVu with 2-Choices}}
\label{ssec:comparison_dejavu_twochoices}

% \paragraph{Convergence time and bias.}
As for the \twochoices dynamics, note that \dejavu for \(h = 2\) is exactly equivalent to it, so all results on the \twochoices dynamics apply to \dejavu.
When restricted to $h=2$, our analysis is worse than the state of the art analysis for \twochoices in terms of the minimum bias required to reach plurality consensus, which is \(\omega(\sqrt{C_1 \log n})\) \cite{Shimizu2025}.
Our requirement on the bias for constant values of \(h\) is comparable to that required by the \threemaj, namely \(\omega(\sqrt{n \log n})\), and remains comparable to the state of the art for \hmaj as \(h\) grows, reaching \(\omega(\sqrt{C_1 \log n})\) when \(h \gg \sqrt{n / C_1}\), which is the same bias required by \twochoices.

% \newpage
% \subsection*{\centering\huge Supplementary Information}
% This file contains the full details of the mathematical and empirical analysis discussed in the main paper.

\section{Preliminaries}
\label{sec:preliminaries}

Consider a complete graph of \(n\) nodes/agents with self-loops.
At time \(t = 0\), each node supports one out of \(k\) opinions.
Time is synchronous and dictated by some global clock.
In the \pullstar model, the protocol \dejavu works as follows:
At each round, agents start sampling opinion u.a.r.\ with repetition from the network.
The moment an agent samples for the second time some opinion \(x\), it adopts \(x\).
Trivially, since there are \(k\) opinions, update takes place in at most \(k+1\) samples.
After all nodes have updated, time increases by 1 and the nodes repeat the same protocol.
In \pullh, the number of samples is capped at \(h\), for any given \(h \ge 2\).
If a node does not sample twice any opinion after \(h\) samples, it will simply update by keeping its own opinion.
After each update, time increases by 1 and the nodes repeat the same protocol.
Note that if \(h = 2\), \dejavu is exactly the \twochoices protocol: hence, in some sense, \dejavu is a generalization of \twochoices.

Given any round \(t \ge 0 \), we denote by \(C^{(t)} = (C_1^{(t)}, \ldots, C_k^{(t)})\) the \textit{configuration of the system} at round \(t\), that is, by \(C_i^{(t)}\) we denote the number of nodes supporting opinion \(i\) at time \(t\).
We omit the dependence on \(t\) when it is clear from the context.
Conditioned on a configuration \(C\) at some round \(t \ge 0\), we denote by \(C' = (C_1', \ldots, C_k')\) the random configuration at time \(t+1\).
Given any configuration \(C\) at time \(t \ge 0\), let \(p_i = C_i/n\) be the density of opinion \(i\) at time \(t\), and \(p = (p_1, \ldots, p_k)\) the \textit{configuration of densities} at time \(t\).
Similarly to \(C'\), we also denote by \(p' = (p_1', \ldots, p_k')\) the random densities at time \(t+1\).

Given a configuration \(C\) at time \(t \ge 0\), we also define \(q_i\) as the expected amount of agents that update to opinion \(i\) in the next round by sampling opinion \(i\) twice in \pullh.
Let \(Q = \sum_{i \in [k]} q_i\).

Furthermore, we define \(\MM_i\) as the event that an agent updates to opinion \(i\) when running \dejavu in \pullstar.

We are interested in analyzing the growth of the \textit{bias} of the configuration, that is, given a configuration \(C = (C_1, \ldots, C_k)\) at time \(t \ge 0\), the \textit{bias} of \(C\) is 
\(\Delta = \max_{i \in [k]}\{\min\{j \neq i\} \{\Delta_{i,j}\}\}\), where \(\Delta_{i,j} = C_i - C_j\).
Usually, given a configuration \(C\), we will assume that it is \textit{ordered}, that is, \(C_1 \ge \ldots \ge C_k\) and, hence, \(\Delta = C_1 - C_2\).
In such a case, we also set \(\Delta_i = C_1 - C_i\).
Similarly to before, given an ordered configuration \(C\), we define \(\Delta_i' = C_1' - C_i'\). 
Also, when we start from an ordered configuration \(C = (C_1, \ldots, C_k\)) at time \(0\), we define \(\Delta^{(t)}_j = C_1^{(t)} - C_j^{(t)}\) for every \(j \in [k]\).

In the following, given a vector \(x = (x_1, \ldots, x_m)\), we write \(\norm{x} = \norm{x}_2 = \sqrt{\sum_{i \in [m]} x_i^2}\).

\section{\texorpdfstring{Amplification of the multiplicative bias of \dejavu in $\pull^*$}{Amplification bias Dejavu in PULL*}}
\label{sec:amplification_bias_PULL*}
In this section, we focus on the \pullstar model. 
This alternative perspective allows us to derive upper and lower bounds on the expected amplification of the multiplicative bias at the next round, presented in the next lemma.

We denote by $\MM_{j}$ the event that an agent adopts opinion $j$ after one round of \dejavu in $\pull^*$.

\begin{lemma}
    \label{lemma:lower_upper_ratio_probability}
    Consider an ordered configuration \(C\) at any given time. 
    For any $i\leq j\in[k]$, we have
    \[
        \frac{p_i^2}{p_j^2}\cdot\frac{p_i+3p_j}{3p_i+p_j} \leq \frac{\pr{\MM_{i} \mid C}}{\pr{\MM_{j} \mid C}} \leq  \frac{p_i^2}{p_j^2}
    \]
\end{lemma}

Its proof relies on an equivalent continuous-time interpretation of the sampling process based on a Poisson race argument. This viewpoint allows us to directly compare
\[
    \frac{\pr{\MM_i \mid C}}{\pr{\MM_j \mid C}}
\]
without explicitly accounting for the remaining opinions.

For completeness, we first establish the equivalence between the $\pull^*$ dynamics and the corresponding Poisson race formulation. The proof of \Cref{lemma:lower_upper_ratio_probability} will follow.

In our model $\mathcal{PULL}^*$, an agent following \dejavu repeatedly samples opinions u.a.r.\ from a system with configuration \(\pa{C_1,\ldots,C_k}\) until some opinion is observed for the second time.
Formally, let \( \mathbf{e}_i \in \{0,1\}^k \) denote the standard basis vector with a 1 in the \( i \)-th position and 0 elsewhere, i.e.,
\[
\left( \mathbf{e}_i \right)_j =
\begin{cases}
    1 & \text{if } j = i, \\
    0 & \text{otherwise},
\end{cases}
\quad \text{for all } j \in [k].
\]
Let $\cpa{S_j}_{j\in\bbN}$ be i.i.d. random variables such that $\pr{ S_j = \mathbf{e}_i }=C_i/n$. These random variables describe the outcomes of the samples of opinions.
Let
\[
    \pa{X_1^{(\ell)},\ldots,X_k^{(\ell)}} = \sum_{j=1}^\ell S_j \sim \text{Multinomial}\pa{ \ell, \pa{C_1/n,\ldots, C_k/n}}~,
\] describing the outcome of the first $\ell$ samples. Let $L_i\coloneq  \min\cpa{ \ell\in\bbN : X_i^{(\ell)} = 2 }$,
the first time opinion $i$ gets $2$ samples, and $L\coloneq  \min_{i\in[k]}\cpa{ L_i}$, so that only after the $L$-th sample, the whole sample contains $2$ occurrences of an opinion. According to \dejavu's update rule, the agent adopts opinion $i$ if and only if
\[
    i= \argmin_{j\in[k]} \cpa{ L_j }.
\]
An equivalent way to model this process is the so-called Poisson race.
Let $\cpa{Y_i(t),t\geq 0}_{i\in[k]}$ be independent homogeneous Poisson processes with support on $\cpa{0,1,2,\ldots}$, meaning that
\begin{enumerate}
    \item For all $t\geq 0$, $Y_i(t) \sim \text{Poisson}(C_i\cdot t)$.
    \item The process \( \cpa{Y_i(t),t\geq 0} \) has independent increments, i.e. for any times \( 0 \leq t_0 < t_1 < \cdots < t_n \), the random variables $\left\{ Y_j(t_{i+1}) - Y_j(t_i) \right\}_{i=0}^{n-1}$ are mutually independent.
\end{enumerate}
where $Y_i(t)$ counts the number of samples of opinion $i$ up to time $t$. Note that this continuous time is distinct from the discrete time of the samples. Let $T_i= \inf\cpa{t>0 : Y_i(t) = 2} $ be the time opinion $i$ is sampled for the second time. We say that the index $i$ wins the Poisson race if
\[
    i= \argmin_{j\in[k]} \cpa{ T_j }.
\]
\begin{theorem}
    \label{thm:equivalence_poisson_multinomial}
    The protocol \dejavu in $\pull^*$ can be modeled equivalently by a Poisson race; in other words, we have
    \[
        \pr{\MM_i \mid C} = \pr{ \text{``$i$ wins the Poisson race"} }
    \]
\end{theorem}

\begin{proof}
By the fact that, for all $t\geq 0$ and $i\in[k]$, $Y_i(t)\sim \text{Poisson}(C_i\cdot t)$ and $\pa{Y_i(t)}_{i\in[k]}$ are mutually independent, we can apply \cref{thm:poisson_approximation} and we obtain that
\begin{equation}
    \label{eq:poisson_multinomial_equivalence}
    \pa{ \pa{Y_i(t)}_{i\in[k]} \mid  \sum_{i\in[k]} Y_i(t) = 1  } \sim \text{Multinomial}\pa{ 1, \pa{C_1/n,\ldots, C_k/n} } \sim S_1.
\end{equation}
Let $Y(t)= \sum_{i\in[k]} Y_i(t)$. Since independence is preserved under summation and the sum of Poisson r.v. is a Poisson r.v. we have that $Y(t)$ is a homogeneous Poisson process with support on $\cpa{0,1,2,\ldots}$ s.t.
\begin{enumerate}
    \item For all $t\geq 0$, $Y(t) \sim \text{Poisson}(t)$.
    \item The process \( \cpa{Y(t),t\geq 0} \) has independent increments.
\end{enumerate}
Let the stopping times $\tau^{(s)}\coloneq  \inf\cpa{ t \geq 0: Y(t) \geq  s }$ for all $s\in\bbN$. Since $Y_i(t)$ is an independent homogeneous Poisson process, we have that the r.v.s ~$Y_i(t + \tau^{(s)}) - Y_i( \tau^{(s)})$ and  $Y_i(t)$ are independent and identically distributed, for all $t\geq0, s\in\bbN, i\in[k]$.
This fact, together with \cref{eq:poisson_multinomial_equivalence}, implies that if the number of total samples in the interval $[\tau^{(s)}, t+\tau^{(s)}]$ is equal to one, then the distribution of any opinion $i$ in that interval is $\text{Multinomial}\pa{ 1, \pa{C_1/n,\ldots, C_k/n} }$;
in formulas, we have that for all $t> 0$ and all $s\in\bbN$,
\begin{equation}
    \label{eq:equivalence_poisson_mult_with_conditioning}
    \pa{ \pa{Y_i(t+ \tau^{(s)}) - Y_i(\tau^{(s)})  }_{i\in[k]} \mid  Y(t+ \tau^{(s)}) - Y(\tau^{(s)}) = 1  } \sim S_1.
\end{equation}
Since for all ~$t,\varepsilon > 0$,
\[
    \pr{Y(t+\varepsilon)-Y(t)\geq 1} = \pr{Y(\varepsilon)\geq 1} = 1 - \exp{-\varepsilon}\xrightarrow{\varepsilon\to 0}0,
\]
we have that $\tau^{(s)}= \inf\cpa{ t \geq 0: Y(t) =  s }$.
Therefore, we obtain that $$\pr{Y(\tau^{(s+1)})=s+1, Y(\tau^{(s)})=s}=1,$$ and, by setting $t= \tau^{(s+1)} - \tau^{(s)}$, from \cref{eq:equivalence_poisson_mult_with_conditioning} we have that
\begin{align*}
    &\pr{\pa{Y_i(\tau^{(s+1)}) - Y_i(\tau^{(s)})  }_{i\in[k]} = x}
    \\
    = \ & \pr{Y(\tau^{(s+1)}) - Y(\tau^{(s)}) = 1} \cdot \pr{ \pa{Y_i(\tau^{(s+1)}) - Y_i(\tau^{(s)})  }_{i\in[k]} = x \mid  Y(\tau^{(s+1)}) - Y(\tau^{(s)}) = 1  }
    \\
    = \ & \pr {S_j = x},
\end{align*}
for all $x$ in the probability space.
Moreover, by the fact the increments of the process $Y(t)$ are independent, as the samples $\pa{S_j}_{j\in[\bbN]}$, we obtain, for all $\ell\in\bbN$, that
\[
    \pa{X_1^{(\ell)},\ldots,X_k^{(\ell)}} = \sum_{j\in[\ell]} S_j \sim \sum_{j\in[\ell]} \pa{Y_i(\tau^{(j)}) - Y_i(\tau^{(j-1)})  }_{i\in[k]} = \pa{Y_i(\tau^{(\ell)})  }_{i\in[k]}.
\]
In particular, this implies that
\begin{align*}
    T_i &= \inf\cpa{t>0 : Y_i(t)=2} 
    \\
    &= \tau\pa{\min\cpa{s\in\bbN : Y_i(\tau^{(s)})=2}} 
    \\
    &\sim \tau\pa{\min\cpa{s\in\bbN : X_i^{(s)}=2}} 
    \\
    &= \tau\pa{L_i}.
\end{align*}
Since, by definition, $\tau^{(s)}$ is non-decreasing in $s$, the index that minimizes \(\tau\!\pa{L_i}\) also minimizes \(L_i\), concluding the proof of \cref{thm:equivalence_poisson_multinomial}.
\end{proof}

We have established that an equivalent way to model this protocol is the following.
Let $\cpa{X_i(t),t\geq 0}_{i\in[k]}$ independent Poisson processes with support on $\cpa{0,1,2,\ldots}$ s.t.
\[
    X_i(t) \sim \text{Poisson}(c_i\cdot t),
\]
where $c_i$ counts the number of agents with opinion $i$ at the current configuration. In this way, $X_i(t)$ counts the number of samples of opinion $i$ up to time $t$.
Let $T_i= \inf\cpa{t\geq 0 : X_i(t) = 2} $ be the time opinion $i$ is sampled twice. An agent adopts opinion $i$ if
\[
    i= \argmin_{j\in[k]} \cpa{ T_j }.
\]

\begin{proof}[Proof of \Cref{lemma:lower_upper_ratio_probability}]
    Suppose the sequence of samples is \textbf{infinite} (it extends beyond the point when some opinion is sampled twice).
    Let $T$ be the time an opinion $\ell\notin\{i,j\}$ is sampled twice. Formally,
    \[
        T\coloneq \min_{\ell\in[k]\setminus\cpa{i,j}} \cpa{ T_\ell},
    \]
    where \(T_\ell\) is the first time opinion \(\ell\) is sampled twice. Let
    \[
        X_{i,j}\coloneq  X_i(T) + X_j(T).
    \]
We have
\begin{align*}
    \frac{\pr{\MM_{i} \mid C}}{\pr{\MM_{j} \mid C}}
    &=
    \frac
    {\pr{X_i(T)=2 \mid X_{i,j}=2}\cdot\pr{X_{i,j}=2}+\pr{T_i<T_j \mid X_{i,j}\geq3}\cdot\pr{X_{i,j}\geq3}}
    {\pr{X_j(T)=2 \mid X_{i,j}=2}\cdot\pr{X_{i,j}=2}+\pr{T_j<T_i \mid X_{i,j}\geq3}\cdot\pr{X_{i,j}\geq3}}
\end{align*}
By the standard inequality $\min\left\{ \frac{x_1}{y_1}, \frac{x_2}{y_2} \right\} \leq \frac{x_1  + x_2}{y_1 + y_2} \leq \max\left\{ \frac{x_1}{y_1}, \frac{x_2}{y_2}\right\}$, we obtain
\begin{align*}
    \frac{\pr{\MM_{i} \mid C}}{\pr{\MM_{j} \mid C}}
    &\geq
    \min\left\{
        \frac{\pr{X_{i}(T)=2 \mid X_{i,j}=2}}
        {\pr{X_{j}(T) = 2 \mid X_{i,j}=2}}
        ,\
        \frac{\pr{T_i<T_j \mid X_{i,j}\geq3}}
        {\pr{T_j<T_i \mid X_{i,j}\geq3}}
    \right\}.
    \\
    \frac{\pr{\MM_{i} \mid C}}{\pr{\MM_{j} \mid C}}
    &\leq
    \max\left\{
        \frac{\pr{X_{i}(T)=2 \mid X_{i,j}=2}}
        {\pr{X_{j}(T) = 2 \mid X_{i,j}=2}}
        ,\
        \frac{\pr{T_i<T_j \mid X_{i,j}\geq3}}
        {\pr{T_j<T_i \mid X_{i,j}\geq3}}
    \right\}.
\end{align*}
Since \(T\) is a deterministic function of \(\cpa{X_\ell(t),t\geq 0}_{\ell\in[k]\setminus\cpa{i,j}}\), it is independent of \(\cpa{X_i(t),t\geq 0}\) and \(\cpa{X_j(t),t\geq 0}\). Therefore, conditioning on \(\cpa{T=t}\), we have \(X_i(T)\sim\text{Poisson}(p_i\cdot t)\) and \(X_j(T)\sim\text{Poisson}(p_j\cdot t)\) for all \(t\geq 0\). Let \(B_x\sim \text{Binomial}(x,\frac{p_i}{p_i+p_j})\). By applying \cref{thm:poisson_approximation}, we obtain \(\pa{X_i(T)\mid X_{i,j}=x}\sim B_x\), \(\pa{X_j(T)\mid X_{i,j}=x} = x - B_x\), and therefore
\[
    \frac{\pr{X_{i}(T)=2 \mid X_{i,j}=2}}
    {\pr{X_{j}(T)=2\mid X_{i,j}=2}}
    =
    \frac{\pa{\frac{p_i}{p_i+p_j}}^2}
    {\pa{\frac{p_j}{p_i+p_j}}^2}
    =
    \frac{p_i^2}{p_j^2}.
\]
Moreover,
\[
    \frac{\pr{T_i<T_j \mid X_{i,j}\geq3}}
    {\pr{T_j<T_i \mid X_{i,j}\geq3}}
    =
    \frac{(\frac{p_i}{p_i+p_j})^2 + 2(\frac{p_i}{p_i+p_j})^2\frac{p_j}{p_i+p_j}}
    {(\frac{p_j}{p_i+p_j})^2 + 2(\frac{p_j}{p_i+p_j})^2\frac{p_i}{p_i+p_j}}
    =
    \frac{p_i^2}{p_j^2}\cdot\frac{p_i+3p_j}{3p_i+p_j}
    ,
\]
where we used that $\pr{T_i<T_j \mid X_{i,j}\geq3}$ is the probability that the extraction sequence contains no occurrence of opinion \(j\) before the second occurrence of \(i\), plus the probability that the sequence contains both \(i\) and \(j\) and the last relevant extraction is \(i\).

We conclude that
\[
        \frac{p_i^2}{p_j^2}\cdot\frac{p_i+3p_j}{3p_i+p_j} \leq \frac{\pr{\MM_{i} \mid C}}{\pr{\MM_{j} \mid C}} \leq \pa{\frac{p_i}{p_j}}^2.
    \]
This concludes the proof of \cref{lemma:lower_upper_ratio_probability}.
\end{proof}

\section{From PULL* to PULL(h)}
\label{sec:star_to_h}
The Poisson clocks equivalence used in \Cref{lemma:lower_upper_ratio_probability} crucially relies on the fact that the sampling sequence is infinite. Consequently, this equivalence no longer holds when analyzing \dejavu under the $\pull(h)$ model, where the number of samples is capped at $h$. Interestingly, conditioning on the event that an agent observes the same opinion twice within the first $h$ samples further amplifies the multiplicative bias in favor of larger opinions. The goal of this section is to prove the following lemma, which allows us to transfer the bounds established for the $\pull^*$ model to the $\pull(h)$ model.

\begin{lemma}
    \label{lemma:monotonicity_samples_dejavu}
    Let $H$ be the number of samples until an agent samples an opinion a second time. Then, for all opinions $i,j \in [k]$ such that $p_i \geq p_j$, and for all $h = 2,\ldots, k+1$, we have
    \[
        \frac{\pr{\MM_i \mid C}}{\pr{\MM_j \mid C}} 
        \leq 
        \frac{\pr{\MM_i, H\leq h \mid C}}{\pr{\MM_j, H\leq h \mid C}} 
        \leq 
        \frac{p_i^2}{p_j^2}.
    \]
\end{lemma}
For \(h \ge k+1\), the event \(\cpa{H \le h}\) coincides with \(\cpa{H \le k+1}\), since after \(k+1\) samples some opinion must repeat.
The proof contains several algebraic manipulations, but the main idea is to manipulate the multinomial distribution and
apply Newton's inequalities to elementary symmetric sums of the probabilities.

For any $h,i\le k$, let 
\[
C_i^{h}\subseteq [k]\setminus\{i\}
\]
denote the collection of all subsets $S$ of $[k]\setminus\{i\}$ having 
cardinality $h$. 
Each such set $S\in C_i^h$ indexes one monomial of the $h$-th elementary 
symmetric polynomial in the variables $\{p_j : j\neq i\}$, namely
\[
\prod_{j\in S} p_j .
\]
Consequently,
\[
    e_h([k]\setminus\cpa{i}) \coloneq  \sum_{S\in C_i^h} \prod_{j\in S}p_j
\]
is exactly the $h$-th elementary symmetric sum in the variables 
\(\{p_j\}_{j\neq i}\). It satisfies the following monotonicity property in \(h\).

\begin{claim}
    \label{claim:elementary_symmetric_sum}
    We have that for $2 \leq h \leq k-2$ and for all opinion $i,j \in [k]$ s.t. $p_i \geq p_j$,
    \[
    \frac{ \sum_{S\in C_i^h} \prod_{\ell\in S}p_\ell }{ \sum_{S\in C_j^h} \prod_{\ell\in S}p_\ell  } \leq 
    \frac{ \sum_{S\in C_i^{h-1}} \prod_{\ell\in S}p_\ell }{ \sum_{S\in C_j^{h-1}} \prod_{\ell\in S}p_\ell  }
    \]
\end{claim}
\begin{proof}
Let
\[
R_h\coloneq \frac{e_h([k]\setminus\{i\})}{e_h([k]\setminus\{j\})}.
\]

Put $C=[k]\setminus\{i,j\}$.  We have
\[
e_h([k]\setminus\{i\})=e_h(C)+p_j\,e_{h-1}(C),\qquad
e_h([k]\setminus\{j\})=e_h(C)+p_i\,e_{h-1}(C).
\]
Hence, setting
\[
E\coloneq e_h(C),\qquad F\coloneq e_{h-1}(C),\qquad G\coloneq e_{h-2}(C),
\]
we may write
\[
R_h=\frac{E+p_jF}{E+p_iF},\qquad
R_{h-1}=\frac{F+p_jG}{F+p_iG}.
\]

To compare $R_h$ and $R_{h-1}$, we compute
\[
R_h\le R_{h-1}
\ \iff\
(E+p_jF)(F+p_iG)\le (E+p_iF)(F+p_jG).
\]
Expanding both sides and canceling common terms, this inequality becomes
\[
(p_i-p_j)\,(EG-F^2)\le0.
\]
Since $p_i\ge p_j$, the last inequality is equivalent to
\[
F^2\ge EG.
\]

Newton's inequalities state that
\[
e_{m}(C)^2\ge e_{m-1}(C)\,e_{m+1}(C)
\qquad\text{for }1\le m\le |C|-1.
\]
Taking $m=h-1$, gives that, for $2\leq h \leq |C| = k-2$, we have 
\[
e_{h-1}(C)^2\ge e_{h-2}(C)\,e_h(C),
\]
i.e.\ $F^2\ge EG$.  Therefore $R_h\le R_{h-1}$.
\medskip
This proves the claim for every $h$ such that $2\le h\le k-2$.
\end{proof}
We now study the multinomial distribution conditioned on the number of samples until the agent observes an opinion twice, and obtain the following.
\begin{claim}
    \label{claim:monotonicity_samples_dejavu_equality}
    Let $H$ be the number of samples until an agent samples an opinion a second time. We have, for all opinions $i,j \in [k]$ such that $p_i \ge p_j$ and all $h = 3,\ldots, k+1$, that
    \[
        \frac{\pr{\MM_i\mid H=h-1, C}}{\pr{\MM_j\mid H=h-1, C}} \geq \frac{\pr{\MM_i\mid H=h, C}}{\pr{\MM_j\mid H=h, C}}
    \]
\end{claim}
\begin{proof}
    For every \(i\in [k]\), the event \(\MM_i \cap \cpa{H=h}\) occurs exactly when opinion \(i\) appears once among the first \(h-1\) samples, all remaining opinions observed in those \(h-1\) samples are distinct, and the \(h\)-th sample is again \(i\). Therefore,
    \begin{align*}
        \pr{\MM_i, H=h \mid C}
        &= \pa{(h-1)! \ p_i \sum_{S\in C_i^{h-2}} \prod_{j\in S} p_j } p_i \\
        &= (h-1)! \ p_i^2 \sum_{S\in C_i^{h-2}} \prod_{j\in S}p_j.
    \end{align*}
    If $h=3$, then
    \[
        \frac{\pr{\MM_i\mid H=3, C}}{\pr{\MM_j\mid H=3, C}}
        =
        \frac{p_i^2 \sum_{\ell\neq i} p_\ell}{p_j^2 \sum_{\ell\neq j} p_\ell}
        =
        \frac{p_i^2(1-p_i)}{p_j^2(1-p_j)}
        \le
        \frac{p_i^2}{p_j^2}
        =
        \frac{\pr{\MM_i\mid H=2, C}}{\pr{\MM_j\mid H=2, C}},
    \]
    where the inequality follows from $p_i\ge p_j$.
    For $4\le h\le k$, by the previous identity and \cref{claim:elementary_symmetric_sum}, we obtain
    \[
        \frac{\pr{\MM_i\mid H=h, C}}{\pr{\MM_j\mid H=h, C}} 
        = \frac{p_i^2 \sum_{S\in C_i^{h-2}} \prod_{j\in S}p_j}{p_j^2 \sum_{S\in C_j^{h-2}} \prod_{j\in S}p_j} 
        \leq
        \frac{p_i^2 \sum_{S\in C_i^{h-3}} \prod_{j\in S}p_j}{p_j^2 \sum_{S\in C_j^{h-3}} \prod_{j\in S}p_j}
        =\frac{\pr{\MM_i\mid H=h-1, C}}{\pr{\MM_j\mid H=h-1, C}} .
    \]
    Finally, for $h=k+1$, since $C_i^{k-1}$ contains the single set $[k]\setminus\{i\}$, we have
    \[
        \frac{\pr{\MM_i\mid H=k+1, C}}{\pr{\MM_j\mid H=k+1, C}}
        =
        \frac{p_i^2 \prod_{\ell\neq i}p_\ell}{p_j^2 \prod_{\ell\neq j}p_\ell}
        =
        \frac{p_i}{p_j}.
    \]
    On the other hand,
    \[
        \frac{\pr{\MM_i\mid H=k, C}}{\pr{\MM_j\mid H=k, C}}
        =
        \frac{p_i^2 e_{k-2}([k]\setminus\{i\})}{p_j^2 e_{k-2}([k]\setminus\{j\})}
        =
        \frac{p_i}{p_j}
        \cdot
        \frac{\sum_{\ell\neq i} 1/p_\ell}{\sum_{\ell\neq j} 1/p_\ell}
        \ge
        \frac{p_i}{p_j},
    \]
    because \(p_i\ge p_j\) implies \(1/p_i\le 1/p_j\). Hence,
    \[
        \frac{\pr{\MM_i\mid H=k+1, C}}{\pr{\MM_j\mid H=k+1, C}}
        \le
        \frac{\pr{\MM_i\mid H=k, C}}{\pr{\MM_j\mid H=k, C}} .
    \]
\end{proof}
\begin{proof}[Proof of \cref{lemma:monotonicity_samples_dejavu}]
    The following inequalities are all equivalent:
    \begin{align*}
        & \frac{\pr{\MM_i, H\leq h \mid C}}{\pr{\MM_j, H\leq h \mid C}} \geq \frac{\pr{\MM_i \mid C}}{\pr{\MM_j \mid C}} 
        \\
        \iff & \ \pa{\sum_{\ell=2}^{k+1} \pr{\MM_j, H=\ell \mid C}} \pa{ \sum_{\ell=2}^h \pr{\MM_i, H = \ell \mid C}} \\
        \geq \ & \pa{\sum_{\ell=2}^{k+1} \pr{\MM_i, H=\ell \mid C}} \pa{ \sum_{\ell=2}^h \pr{\MM_j, H = \ell \mid C }}
        \\
        \iff & \ \pa{\sum_{\ell=h+1}^{k+1} \pr{\MM_j, H=\ell \mid C}} \pa{ \sum_{\ell=2}^h \pr{\MM_i, H = \ell \mid C}} \\
        \geq \ & \pa{\sum_{\ell=h+1}^{k+1} \pr{\MM_i, H=\ell \mid C}} \pa{ \sum_{\ell=2}^h \pr{\MM_j, H = \ell \mid C}}
        \\
        \iff & \ \frac{\pr{\MM_i, H\leq h \mid C}}{\pr{\MM_j, H\leq h \mid C}} \geq \frac{\pr{\MM_i, H > h\mid C}}{\pr{\MM_j, H > h\mid C} },
    \end{align*}
    where in the second equivalence we canceled in both sides the symmetrical part of the sum. We will prove the last equation holds.
    We have that
    \begin{align*}
        \frac{\pr{\MM_i, H\leq h \mid C}}{\pr{\MM_j, H\leq h \mid C}} 
        & = \frac{\sum_{\ell = 2}^{h }\pr{H=\ell\mid C} \pr{\MM_i\mid H=\ell, C}}{\sum_{\ell = 2}^{h} \pr{H=\ell\mid C} \pr{\MM_j\mid H=\ell,C}} \\
        & \geq \min_{\ell=2,\ldots,h} \frac{\pr{\MM_i \mid H =  \ell,C}}{\pr{\MM_j \mid H = \ell,C}} \\
        & = \frac{\pr{\MM_i \mid H =  h,C}}{\pr{\MM_j \mid H = h,C}},
    \end{align*}
    where the last equality holds by \cref{claim:monotonicity_samples_dejavu_equality}. Similarly, we have
    \begin{align*}
        \frac{\pr{\MM_i, H > h\mid C}}{\pr{\MM_j, H > h\mid C}} 
        & \leq \max_{\ell=h+1,\ldots,k+1} \frac{\pr{\MM_i \mid H =  \ell,C}}{\pr{\MM_j \mid H = \ell,C}} \\
        & = \frac{\pr{\MM_i \mid H =  h+1,C}}{\pr{\MM_j \mid H = h+1,C}} \\ 
        & \leq \frac{\pr{\MM_i \mid H = h,C}}{\pr{\MM_j \mid H = h,C}}.
    \end{align*}
    To conclude the proof of \cref{lemma:monotonicity_samples_dejavu}, we show that
    \begin{align*}
        \frac{\pr{\MM_i\mid C}}{\pr{\MM_j\mid C}} 
        & = \frac{\pr{\MM_i, H > 1\mid C}}{\pr{\MM_j, H > 1\mid C}} \\
        & \leq \max_{\ell=2,\ldots,k+1} \frac{\pr{\MM_i \mid H =  \ell,C}}{\pr{\MM_j \mid H = \ell,C}} \\
        & = \frac{\pr{\MM_i \mid H =  2,C}}{\pr{\MM_j \mid H = 2,C}} \\ 
        & = \frac{p_i^2}{p_j^2}.
    \end{align*}
\end{proof}

\subsection{Generalized birthday paradox}
\label{sec:sample_size_dejavu}
In this section we provide upper and lower bounds on the probability of seeing an opinion twice in \(h\) samples. Our goal is to prove the following lemma.
\begin{lemma}
    \label{lemma:tight_bound_probability_collision}
    Let \(C\) be an ordered configuration.
    Recall that \(\cpa{H \leq h}\) is the event that an agent samples an opinion twice within the first \(h\) samples. For every \(h \ge 2\), we have
    \[
        2^{-11}\min\cpa{ h^2 \norm{p}_2^2, 1}
        \le
        \pr{H \leq h \mid C}
        \le
        \min\cpa{ h^2 \norm{p}_2^2, 1}.
    \]
\end{lemma}

As a simple corollary of \cref{lemma:tight_bound_probability_collision}, we obtain the following.

\begin{corollary}
  \label{lemma:expectation-D}
  Let \(C\) be an ordered configuration.
  Given a sample size \(h \ge 2\) and an opinion probability density \(p\), 
  let \(D\) be the number of agents that update by seeing an opinion twice within the \(h\) samples.
  Writing
  \[
    Q\coloneq \Pr(H\le h\mid C),
  \]
  we have \(\bbE[D\mid C]=nQ\), and therefore
  \[
    2^{-11} n \min\cpa{h^2\| p \|_2^2,1}
    \le
    \bbE[D \mid C]
    \le
    n \min\cpa{h^2\| p \|_2^2,1}.
  \]
\end{corollary}

We start by proving an upper bound on the probability of seeing an opinion twice within the first $h$ samples.
\begin{lemma}
    \label{lemma:LB_poisson_approximation}
    Let $\pa{X_1,\ldots,X_k}$ be distributed as a multinomial random vector with parameters $h$ and $(p_1,\ldots,p_k)$. Then
    \[
        \pr{\cup_{i\in[k]} \cpa{X_i > 1 }} \leq \frac{h^2}{2} \sum_{j\in[k]} p_j^2
    \]
\end{lemma}
\begin{proof}
Let $Z_1,\dots,Z_h$ be the $h$ i.i.d. samples with distribution $(p_1,\dots,p_k)$.
Observe that
\[
\bigcup_{i\in[k]} \{X_i>1\}
\]
occurs if and only if there exist two distinct samples having the same opinion.

For $a<b$, define
\[
I_{a,b} = \mathbf{1}\{Z_a = Z_b\}.
\]
Then
\[
\Pr\!\left(\bigcup_{i\in[k]} \{X_i>1\}\right)
=
\Pr\!\left(\sum_{a<b} I_{a,b} \ge 1\right)
\le
\sum_{a<b} \Pr(Z_a = Z_b),
\]
where the last inequality follows from the union bound.

Since
\[
\Pr(Z_a = Z_b) = \sum_{j\in[k]} p_j^2,
\]
and there are $\binom{h}{2} \le \frac{h^2}{2}$ pairs, we obtain
\[
\Pr\!\left(\bigcup_{i\in[k]} \{X_i>1\}\right)
\le
\frac{h^2}{2} \sum_{j\in[k]} p_j^2.
\]
\end{proof}

\begin{lemma}[Concentration D]
    \label{lemma:concentration_D}
    Let \(C\) be an ordered configuration, and let \(D\) be the number of agents that update by seeing an opinion twice within the first \(h\) samples. Write
    \[
        Q\coloneq \Pr(H\le h\mid C),
        \qquad
        \bbE[D\mid C]=nQ.
    \]
    Then, for every constant \(a>0\), there exists a constant \(K_a>0\) such that, if \(nQ\ge K_a\log n\), then
    \[
        \Pr\!\left(
            D\in \left[\frac{nQ}{2},\frac{3nQ}{2}\right]
            \,\middle|\,
            C
        \right)
        \ge 1-\frac{1}{n^a}.
    \]
\end{lemma}
\begin{proof}
    Since each agent updates independently with probability \(Q\), conditional on \(C\) we have
    \[
        D\sim \mathrm{Bin}(n,Q).
    \]
    Applying the multiplicative Chernoff bound with \(\beta=1/2\), we obtain
    \[
        \Pr\!\left(
            D\notin \left[\frac{nQ}{2},\frac{3nQ}{2}\right]
            \,\middle|\,
            C
        \right)
        \le
        2\exp\!\left(-\frac{nQ}{12}\right).
    \]
    Therefore, if \(nQ\ge K_a\log n\) with \(K_a\ge 12(a+1)\), then
    \[
        2\exp\!\left(-\frac{nQ}{12}\right)\le \frac{1}{n^a},
    \]
    concluding the proof.
\end{proof}

In order to provide a lower bound we will need that \(h \norm{p}_2\) is smaller than some small constant. We will extend the result by monotonicity.

We will use a result by \cite{arratia1989}.
In order to state it, we need to introduce some notation.
Let \(I\) be a finite index set.
For \(\alpha \in I\), let \(X_\alpha\) be a Bernoulli r.v.\ with \(\Pr(X_\alpha = 1) = p_{\alpha} > 0\).
Let \(W = \sum_{\alpha \in I}X_{\alpha}\) and \(\lambda = \bbE[W]\).
For each \(\alpha \in I\), suppose we are given a \(B_\alpha \subseteq I\), with \(\alpha \in B_\alpha\): 
we will think of \(B_\alpha\) as the set such that \(X_\alpha\) is independent (or ``nearly'' independent) of \(X_\beta\) for each \(\beta \in I \setminus B_{\alpha}\).
For any two \(\alpha, \beta \in I\), let \(p_{\alpha \beta} = \bbE[X_{\alpha} X_{\beta}]\).
Furthermore, let 
\begin{align*}
    s_\alpha & = \bbE\Big[ \big\lvert \bbE\big[X_\alpha - p_{\alpha} \mid \sum_{\beta \in I \setminus B_{\alpha}} X_{\beta}\big] \big\rvert \Big].
\end{align*}
Now set
\begin{align*}
    b_1 & = \sum_{\alpha \in I} \sum_{\beta \in B_\alpha} p_\alpha p_\beta; \\
    b_2 & = \sum_{\alpha \in I} \sum_{\beta \in B_{\alpha} \setminus \{\alpha\}} p_{\alpha \beta}; \\
    b_3 & = \sum_{\alpha \in I} s_\alpha.
\end{align*}

The following holds by \cite[Theorem 1]{arratia1989}:

\begin{align}\label{eq:trick:arratia}
    \left \lvert \Pr(W = 0) - e^{-\lambda} \right \rvert \le \frac{1 - \frac{1}{e^{\lambda}}}{\lambda}\cdot (b_1 + b_2 + b_3).
\end{align}

We prove the following.

\begin{theorem}\label{thm:no-collision:small-norm}
    Let \((X_1, \ldots, X_k) \sim\mathrm{Multinomial}(h;p) \), where \(p = (p_1, \ldots, p_k)\) and \(h \ge 2\).
    Let \(C = h \norm{p}_2\).
    Then,
    \[
        \pr{\cap_{i \in [k]} \{X_i \le 1\}} \le \exp[- C^2 / 4] + (1 - \exp[-C^2 / 2])\left(\frac{2C^2}{h} + 2C \right). 
    \]
\end{theorem}
\begin{proof}
    Let $Y_1,\dots,Y_h$ be i.i.d.\ random variables taking values in $\{1,\dots,k\}$ with
    $\Pr(Y_t=i)=p_i$. 
    Define the multinomial counts
    \[
        X_i=\sum_{t=1}^h \mathbf 1\{Y_t=i\},\qquad i=1,\dots,k,
    \]
    so that $(X_1,\dots,X_k)\sim\mathrm{Multinomial}(h;p)$.
    For $1\le a<b\le h$ define the pair-collision events and indicators
    \[
        A_{ab}=\{Y_a=Y_b\},\qquad \xi_{ab}=\mathbf 1_{A_{ab}},
    \]
    and let
    \[
        W=\sum_{1\le a<b\le h} \xi_{ab}.
    \]
    Then $W$ counts the number of colliding pairs among the $h$ draws. Observe that
    \[
        \left\{\forall i,\ X_i\le 1\right\}\ \iff\ \{Y_1,\dots,Y_h \text{ are all distinct}\}
        \ \iff\ \{W=0\}.
    \]
    
    Let $I=\{(a,b):1\le a<b\le h\}$ index the indicators. For $\alpha=(a,b)\in I$ define
    \[
        B_{ab}=\{(c,d)\in I:\{c,d\}\cap\{a,b\}\neq\varnothing\}.
    \]
    If $(c,d)\notin B_{ab}$ then $\{c,d\}\cap\{a,b\}=\varnothing$, and hence \(\xi_{ab}\) is independent
    of \(\cpa{\xi_\beta}_{\beta \in I \setminus B_{ab}}\). Therefore the term \(b_3\) in \cref{eq:trick:arratia} equals \(0\).
    Define
    \[
        p_{ab}=\bbE[\xi_{ab}]=\Pr(Y_a=Y_b)=\sum_{i=1}^k p_i^2=\norm{p}_2^2,
    \]
    and note that $|I|=\binom{h}{2}$. Hence
    \[
        \lambda=\bbE[W]=\sum_{(a,b)\in I} p_{ab}=\binom{h}{2}\norm{p}_2^2.
    \]
    
    We bound the quantities $b_1$ and $b_2$ from \cref{eq:trick:arratia}, namely
    \[
        b_1=\sum_{\alpha\in I}\ \sum_{\beta\in B_\alpha} p_\alpha p_\beta,
        \qquad
        b_2=\sum_{\alpha\in I}\ \sum_{\substack{\beta\in B_\alpha\\ \beta\neq \alpha}} \bbE[\xi_\alpha\xi_\beta].
    \]
    Since $p_\alpha=\norm{p}_2^2$ for all $\alpha$ and $|B_{ab}|=2h-3$ for every $(a,b)$, we have
    \[
        b_1=\binom{h}{2}(2h-3)\norm{p}_2^4
        = \lambda(2h-3)\norm{p}_2^2.
    \]
    Next, if $\beta\in B_{ab}$ and $\beta\neq (a,b)$ then the two pairs share exactly one index, so for
    distinct $a,b,c$,
    \[
        \bbE[\xi_{ab}\xi_{ac}]=\Pr(Y_a=Y_b=Y_c)=\sum_{i=1}^k p_i^3.
    \]
    Each $(a,b)$ has exactly $2(h-2)$ such elements $\beta\neq (a,b)$ in \(B_{ab}\), hence
    \[
        b_2=\binom{h}{2}\,2(h-2)\sum_{i=1}^k p_i^3.
    \]
    Using $\sum_i p_i^3\le \norm{p}_\infty\sum_i p_i^2\le \norm{p}_2^3$, we get
    \[
        b_2\le \binom{h}{2}\,2(h-2)\norm{p}_2^3
        =2(h-2)\lambda\norm{p}_2.
    \]
    Then
    \[
        b_1=\lambda(2h-3)\norm{p}_2^2
        \le 2\lambda h\norm{p}_2^2
        =2\lambda \frac{(h\norm{p}_2)^2}{h}
        = 2C^2\,\frac{\lambda}{h},
    \]
    and
    \[
        b_2\le 2(h-2)\lambda\norm{p}_2 \le 2h\lambda\norm{p}_2 = 2C\,\lambda.
    \]
    Moreover,
    \[
        \lambda=\binom{h}{2}\norm{p}_2^2 \le \frac{h^2}{2}\norm{p}_2^2=\frac12 (h\norm{p}_2)^2= \frac{C^2}{2},
    \]
    and
    \[
        \lambda=\binom{h}{2}\norm{p}_2^2 \ge \frac{h^2}{4}\norm{p}_2^2=\frac14 (h\norm{p}_2)^2= \frac{C^2}{4}, 
    \]
    
    By \cref{eq:trick:arratia}, using $b_3=0$,
    \[
        \bigl|\Pr(W=0)-e^{-\lambda}\bigr|
        \le (b_1+b_2+b_3)\frac{1-e^{-\lambda}}{\lambda}
        = (b_1+b_2)\frac{1-e^{-\lambda}}{\lambda}.
    \]
    Hence,
    \[
        \Pr(W=0)\le \exp[- C^2 / 4] + (1 - \exp[-C^2 / 2])\left(\frac{2C^2}{h} + 2C \right).
    \]
\end{proof}

We are ready to put together upper and lower bounds and prove \Cref{lemma:tight_bound_probability_collision}.
\begin{proof}[Proof of \Cref{lemma:tight_bound_probability_collision}]
    Set
    \[
        c_1\coloneq \frac18,
        \qquad
        c_2\coloneq \frac12,
        \qquad
        c_3\coloneq \frac{1}{16}.
    \]
    Consider first the case $C = h \norm{p}_2 \leq c_1$. By \Cref{thm:no-collision:small-norm,lemma:LB_poisson_approximation}, we have that
    \[
        \pa{1-\exp\pa{- \frac{C^2}{2}}}\pa{1 - \frac{2C^2}{h} - 2C } \leq \pr{H \leq h \mid C} \leq 2\pa{ 1 - \exp\pa{-\frac{C^2}{2} }}.
    \]
    Since \(C\le c_1\), we have \(C^2/2\le 1/128\), and therefore
    \[
        e^{-C^2/2}\le 1-\frac{C^2}{4}.
    \]
    Thus
    \[
        \frac{C^2}{4}\pa{1 - \frac{2C^2}{h} - 2C } \leq \pr{H \leq h \mid C} \leq C^2.
    \]
    As \(C^2/h = h \norm{p}_2^2 \leq C \leq c_1\), we have
    \[
        1 - \frac{2C^2}{h} - 2C \ge 1-4C \ge \frac12,
    \]
    and therefore
    \[
        \frac{C^2}{8} \leq \pr{H \leq h \mid C} \leq C^2.
    \]
    Assume now that \(C=h\norm{p}_2>c_1\).
    If \(\norm{p}_2>c_1\), then
    \[
        \pr{H\le h\mid C}\ge \pr{Y_1=Y_2\mid C}=\norm{p}_2^2>c_1^2= \frac{1}{64}.
    \]
    Otherwise \(\norm{p}_2\le c_1\). If \(\norm{p}_2\ge c_1/2\), then
    \[
        \pr{H\le h\mid C}\ge \pr{H\le 2\mid C}=\norm{p}_2^2\ge \frac{c_1^2}{4}=\frac{1}{256}.
    \]
    If instead \(\norm{p}_2<c_1/2\), let \(h'\coloneq \lfloor c_1/\norm{p}_2\rfloor\). Then \(h'\ge 1\), \(h'<h\), and
    \[
        h'\norm{p}_2
        \ge
        c_1-\norm{p}_2
        >
        \frac{c_1}{2}
        =
        c_3,
        \qquad
        h'\norm{p}_2\le c_1.
    \]
    In the latter subcase, there exists \(h'<h\) such that \(c_3\le h'\norm{p}_2\le c_1\). Since \(C=h\norm{p}_2>c_1\), this also implies \(h\ge 3\), and therefore \(h'\ge 2\). Applying the first case to \(h'\), we obtain
    \[
        \pr{H \leq h \mid C}
        \geq
        \pr{H \leq h' \mid C}
        \geq
        \frac{h'^2 \norm{p}_2^2}{8}
        \geq
        \frac{c_3^2}{8}
        =
        \frac{1}{2048}.
    \]
    Combining the two cases, we conclude that
    \[
        \pr{H\le h\mid C}
        \ge
        \frac{1}{2048}\min\cpa{C^2,1}.
    \]
\end{proof}
\subsection{Some basic inequalities}
In the next section, we derive from \Cref{lemma:monotonicity_samples_dejavu}, several basic inequalities that will be useful throughout the rest of the paper.
\begin{claim}
    \label{claim:q_i/q_j_lb_ub_square_p_i/_pj}
    Let \(C\) be an ordered configuration.
    For any opinions $i,j\in[k]$, it holds
    \[
        \frac{1}{3}\frac{p_i^2}{p_j^2} \leq \frac{q_i}{q_j} \leq 3\frac{p_i^2}{p_j^2}.
    \]
\end{claim}
\begin{proof}
    By \Cref{lemma:monotonicity_samples_dejavu}, we know that if $p_i>p_j$, we have that 
    \[
        \frac{q_i}{q_j} \leq \frac{p_i^2}{p_j^2} \leq 3 \frac{p_i^2}{p_j^2}.
    \]
    Also, by the standard inequality $\min\left\{ \frac{x_1}{y_1}, \frac{x_2}{y_2} \right\} \leq \frac{x_1  + x_2}{y_1 + y_2} \leq \max\left\{ \frac{x_1}{y_1}, \frac{x_2}{y_2}\right\}$, we have
    \[
        \frac{q_i}{q_j} \geq \frac{p_i^2}{p_j^2} \frac{p_i + 3 p_j}{3 p_i + p_j} \geq \frac{p_i^2}{p_j^2} \min\cpa{\frac{1}{3}, 3} = \frac{1}{3}\frac{p_i^2}{p_j^2}.
    \]
    Taking the reciprocal in both equation we prove also the case $p_j<p_i$, concluding the proof of \Cref{claim:q_i/q_j_lb_ub_square_p_i/_pj}.
\end{proof}
\begin{claim}
    \label{claim:q_i_over_Q_ub_and_lb_norm2}
    Let \(C\) be an ordered configuration.
    For all opinions \(i\in[k]\), it holds
    \[
        \frac{p_i^2}{ 3 \norm{p}_2^2} \leq \frac{q_i}{Q} \leq \frac{3 p_i^2}{ \norm{p}_2^2}.
    \]
\end{claim}
\begin{proof}
    By \Cref{claim:q_i/q_j_lb_ub_square_p_i/_pj} for all $i,j\in[k]$, $\frac{q_i}{q_j} \geq \frac{1}{3} \frac{p_i^2}{p_j^2}$.
    Therefore, we obtain
    \[
        \frac{Q}{q_i} = \sum_{j\in[k]} \frac{q_j}{q_i} \geq \sum_{j\in[k]} \frac{1}{3} \frac{p_j^2}{p_i^2} = \frac{\norm{p}_2^2}{3 p_i^2}.
    \]
    Taking the reciprocal of both sides, we conclude the proof of the upper bound to $q_i/Q$.
    On the other hand, we know that for all $i,j\in[k]$, $\frac{q_i}{q_j} \leq 3 \frac{p_i^2}{p_j^2}$. Similarly, we have
    \[
        \frac{Q}{q_i} = \sum_{j\in[k]} \frac{q_j}{q_i} \leq \sum_{j\in[k]} 3 \frac{p_j^2}{p_i^2} = \frac{3 \norm{p}_2^2}{p_i^2}.
    \]
    Taking the reciprocal of both sides, we conclude the proof of \Cref{claim:q_i_over_Q_ub_and_lb_norm2}.
\end{proof}
\begin{claim}
    \label{claim:q_i_over_Q_lb_p_j}
    Let \(C\) be an ordered configuration.
    For any opinion $j$ s.t. $p_j > c_1 p_1$ for some constant $c_1>0$, it holds
    \[
        \frac{q_j}{Q} \geq \frac{c_1}{3} p_j.
    \]
\end{claim}
\begin{proof}
    By \Cref{claim:q_i_over_Q_ub_and_lb_norm2}, we have
    \[
        \frac{q_j}{Q} \geq \frac{p_j^2}{ 3 \sum_{\ell\in[k]} p_\ell^2} \geq  \frac{p_j^2}{ 3 p_1} \geq \frac{c_1}{3} p_j.
    \]
\end{proof}

\begin{claim}
    \label{eq:qdqlb2}
    For all opinions $i\in[k]$, it holds
    \[
        \frac{q_1}{q_i}
        \geq
        \frac{p_1^2}{p_i^2}\frac{p_1 + 3 p_i}{3 p_1 + p_i}
        \geq
        \frac{2p_1^2}{p_i(p_1+p_i)}
        .
    \]
\end{claim}
\begin{proof}
    By \Cref{lemma:monotonicity_samples_dejavu,lemma:lower_upper_ratio_probability}, we have that 
    \[
        \frac{q_1}{q_i} \geq  \frac{p_1^2}{p_i^2}\frac{p_1 + 3 p_i}{3 p_1 + p_i}.
    \]
    We will then prove that
    \[
        \frac{p_1^2}{p_i^2}\frac{p_1 + 3 p_i}{3 p_1 + p_i}
        \geq
        \frac{2p_1^2}{p_i(p_1+p_i)}.
    \]
    Since $p_1>0$, we divide both sides by $p_1^2$. The inequality is equivalent to
\[
\frac{1}{p_i^2} \frac{p_1 + 3 p_i}{3 p_1 + p_i}
\;\ge\;
\frac{2}{p_i(p_1+p_i)}.
\]

Multiplying both sides by the positive quantity 
$p_i^2(3p_1+p_i)(p_1+p_i)$, we obtain the equivalent inequality
\[
(p_1+3p_i)(p_1+p_i)
\;\ge\;
2p_i(3p_1+p_i).
\]

Expanding both sides,
\[
p_1^2 + 4p_1p_i + 3p_i^2
\;\ge\;
6p_1p_i + 2p_i^2.
\]

Rearranging terms gives
\[
p_1^2 - 2p_1p_i + p_i^2 \ge 0,
\]
which is
\[
(p_1 - p_i)^2 \ge 0.
\]
Since this holds for all $p_1\geq p_i$, the claim follows.
\end{proof}

Now we bound the variance of $ C^{(t)}_i$ conditioning on the previous round.
  \begin{lemma}
      \label{lemma:bound_variance}
      Let \(C\) be an ordered configuration.
      We have that 
      \[
        \Var\qty[C'_i \mid C]\leq C_i\min\cpa{c \,h^2 \norm{p}_2^2,1} \pa{ 1 + \frac{3 p_i}{\|p\|_2^2} }.
    \]
  \end{lemma}
  \begin{proof}
      Let \(D_i\) be the number of agents adopting opinion \(i\) after seeing it twice, and let \(D=\sum_{i\in[k]}D_i\). Then \(D\sim \mathrm{Bin}(n,\Pr(H\le h\mid C))\).
    Conditional on \(D=d\), we have
    \[
        C_i' = B_i + (C_i-H_i),
        \qquad
        B_i \sim \mathrm{Bin}\!\left(d,\frac{q_i}{Q}\right),
        \qquad
        H_i \sim \mathrm{Hypergeom}(n,C_i,d).
    \]
    We have that
    \begin{align*}
        \bbE[\Var(C_i'\mid D,C)\mid C] 
        &=
        \bbE\qty[D \frac{q_i}{Q} \pa{1 - \frac{q_i}{Q}} + D p_i (1-p_i)\frac{n-D}{n-1} \,\middle|\, C]
        \\
        &\leq 
        \bbE\qty[D \frac{q_i}{Q} + D p_i \,\middle|\, C ]
        \\
        &=
        n\Pr(H \leq h\mid C) \pa{\frac{q_i}{Q}  + p_i}
        \\
        &\leq
        n\Pr(H \leq h\mid C) p_i \pa{ \frac{3 p_i}{\norm{p}_2^2}  + 1}.
    \end{align*}
    Moreover, we have
    \begin{align*}
        \Var[\bbE(C_i'\mid D,C)\mid C]
        &=
        \Var[D \frac{q_i}{Q} + (n-D) p_i \mid C]
        \\
        &=
        \pa{\frac{q_i}{Q} - p_i}^2 \Var[D \mid C]
        \\
        &=\pa{\frac{q_i}{Q} - p_i}^2 n\Pr(H \leq h \mid C)\pa{1-\Pr(H \leq h \mid C)}
        \\
        &\leq
        n\Pr(H \leq h \mid C) \pa{\frac{q_i}{Q}  + p_i}
        \\
        &\leq
        n\Pr(H \leq h \mid C) p_i \pa{ \frac{3 p_i}{\norm{p}_2^2}  + 1}.
    \end{align*}
    Therefore, we obtain that
    \begin{align*}
        \Var[C_i' \mid C] 
        &=
        \bbE[\Var(C_i'\mid D, C)\mid C] + \Var[\bbE(C_i'\mid D,C) \mid C]
        \\
        &\leq
        2 n\Pr(H \leq h\mid C) p_i \pa{ \frac{3 p_i}{\norm{p}_2^2}  + 1}
        \\
        &\leq
        2 n \min\cpa{c h^2 \norm{p}_2^2,1} p_i \pa{ \frac{3 p_i}{\norm{p}_2^2}  + 1}
        &\pa{\text{by \Cref{lemma:tight_bound_probability_collision} }}
    \end{align*}
    Since \(np_i=C_i\), this proves the claim after adjusting the absolute constant.
  \end{proof}

\section{Amplification of the multiplicative and additive bias in \texorpdfstring{\pullh}{PULL(h)}}\label{sec:amplification}

\subsection{Expected amplification of the multiplicative bias}\label{sec:expected-amplification}
The following lemma characterizes the expected growth of the ratio $C_1/C_i$. Its proof relies on partitioning the agents into two groups: those who, after $h$ samples, observe two occurrences of the same opinion and update accordingly, and those who retain their current opinion. The first group is responsible for amplifying the bias in favor of $C_1$. The argument therefore builds on two ingredients: \Cref{lemma:expectation-D}, which provides the expected size of this group, and \Cref{lemma:lower_upper_ratio_probability}, which quantifies the bias of these agents toward adopting opinion~1.

\begin{lemma}
\label{lemma:growth_expected_ratio}
Let \(K\coloneq 2^{-11}\) and \(\eta\coloneq 2^{-16}\). Then, for any system configuration $C=\pa{C_1,C_2,\ldots,C_k}$ with $C_1 = \max\cpa{C_i}$ and for every \(i\in\cpa{2,\ldots,k}\) such that \(C_i>0\), after one round of \dejavu it holds
\[
    \frac{\bbE[C'_1 \mid C] }{\bbE[C'_i \mid C] }
    \geq
    \frac{C_1}{C_i} + \eta \min\cpa{\frac{C_1}{n} h^2, 1}\pa{ \frac{C_1}{C_i} - 1 }.
\]
Moreover, for every such \(i\) and for any integer $d\in [\frac{\bbE[D]}{2},\frac{3\bbE[D]}{2}]$, 
    \[
    \frac{\bbE[C'_1 \mid D = d, C] }{\bbE[C'_i \mid D = d , C] }
    \geq
    \frac{C_1}{C_i} + \eta \min\cpa{\frac{C_1}{n} h^2, 1}\pa{ \frac{C_1}{C_i} - 1 }.
\]
\end{lemma}
Since the proof of this lemma requires carefully handling a ratio, we divide it into several sublemmas covering different cases.

\begin{lemma}
    Let \(C\) be an ordered configuration.
    Let \(K\coloneq 2^{-11}\).
    If $h^2\norm{p}_2^2 <1$ and $K p_1 h^2 <6$, then for any integer $d\in [\frac{\bbE[D]}{2},\frac{3\bbE[D]}{2}]$
    \[
        \frac{\bbE[C'_1 \mid C] }{\bbE[C'_i \mid C] }, \frac{\bbE[C'_1 \mid D = d, C] }{\bbE[C'_i \mid D = d , C] } \geq \frac{C_1}{C_i} + \frac{K p_1 h^2}{24}\pa{ \frac{C_1}{C_i} - 1 }.
    \]
\end{lemma}
\begin{proof}
    By \Cref{lemma:tight_bound_probability_collision} and by hypothesis, we know that $Q \geq K h^2 \norm{p}_2^2$,
for some constant $K>0$. Therefore, conditioning on the event $D = d$, we have that $\frac{d}{n} \geq \frac{K}{2} h^2 \norm{p}_2^2$.
We have
\begin{align*}
        &\frac{\bbE[C'_1 \mid C] }{\bbE[C'_i \mid C] }, \frac{\bbE[C'_1 \mid D = d, C] }{\bbE[C'_i \mid D = d , C] }
        \\
        =&
        \frac{\frac{q_1}{Q}\cdot Q  + p_1(1-Q)}{\frac{q_i}{Q}\cdot Q  + p_i(1-Q)},
        \frac{\frac{q_1}{Q}\cdot \frac{d}{n}  + p_1(1-\frac{d}{n})}{\frac{q_i}{Q}\cdot \frac{d}{n}  + p_i(1-\frac{d}{n})}
        \\
        \geq &
        \frac{\frac{q_1}{Q}\cdot Q  + p_1}{\frac{q_i}{Q}\cdot Q  + p_i},
        \frac{\frac{q_1}{Q}\cdot \frac{d}{n}  + p_1}{\frac{q_i}{Q}\cdot \frac{d}{n}  + p_i}
         , \text{ b/c } \frac{q_1}{q_i}\geq \frac{p_1}{p_i}
         \\
        \geq &
        \frac{\frac{q_1}{Q}\cdot \frac{K}{2}\cdot h^2 \sum_j p_j^2 + p_1 }{\frac{q_i}{Q}\cdot \frac{K}{2} \cdot h^2 \sum_j p_j^2 + p_i}
        , \text{ b/c } Q, \frac{d}{n} \geq \frac{K}{2} h^2 \norm{p}_2^2 \text{ and } \frac{q_1}{q_i}\geq \frac{p_1}{p_i}
        \\
        = &
        \frac{\frac{1}{\sum_j q_j/q_1}\cdot \frac{K}{2} h^2 \sum_j p_j^2 + p_1}
        {\frac{q_i/q_1}{\sum_j q_j/q_1}\cdot \frac{K}{2} h^2 \sum_j p_j^2 + p_i }
        \\
        \geq &
        \frac{\frac{1}{\sum_j 3p_j^2/p_1^2}\cdot \frac{K}{2} h^2 \sum_j p_j^2 + p_1}
        {\frac{q_i/q_1}{\sum_j 3p_j^2/p_1^2}\cdot \frac{K}{2} h^2 \sum_j p_j^2 + p_i}
        , \text{ by } \cref{eq:qdqlb2} \text{ b/c } \frac{q_1}{q_j}\geq \frac{p_1}{p_j}
        \\
        = &
        \frac{p_1^2\cdot \frac{K}{2} h^2 + 3p_1}
        {p_1^2\cdot \frac{q_i}{q_1}\cdot  \frac{K}{2} h^2 + 3p_i}
        \\
        \geq &
        \frac{p_1^2\cdot \frac{K}{2} h^2 + 3p_1}
        {p_i\cdot \frac{p_1+p_i}{2}\cdot  \frac{K}{2} h^2 + 3p_i}
        , \text{ by \Cref{eq:qdqlb2}} 
        \\
        = &
        \frac{p_1}{p_i}\cdot
        \frac{p_1\cdot \frac{K}{2} h^2 + 3}
        {\frac{p_1+p_i}{2}\cdot \frac{K}{2} h^2 + 3}
        \\
        = &
        \frac{p_1}{p_i}\cdot
        \frac{p_1\cdot \frac{K}{2} h^2 + 3}
        {\frac{1}{2}\pa{1+\frac{p_i}{p_1}}p_1\cdot \frac{K}{2} h^2 + 3}
        \\
        \geq &
        \frac{p_1}{p_i}\cdot
        \pa{1+ \frac{1-\frac{1}{2}\pa{1+\frac{p_i}{p_1}}}{6} \frac{K}{2} p_1 h^2}
        , \text{ b/c } \frac{x+3}{cx+3}\ge 1+\frac{1-c}{6}x \text{ for } x \in [0,3], c \in [0,1].
        \\
        = &
        \frac{p_1}{p_i}+
        \frac{ K p_1 h^2}{24} \pa{\frac{p_1}{p_i}-1}
    \end{align*}
\end{proof}
\begin{lemma}
Let \(C\) be an ordered configuration.
    Let \(K\coloneq 2^{-11}\).
    If $h^2\norm{p}_2^2 <1$ and $K p_1 h^2 >6$, then for any integer $d\in [\frac{\bbE[D]}{2},\frac{3\bbE[D]}{2}]$
    \[
        \frac{\bbE[C'_1 \mid C] }{\bbE[C'_i \mid C] }, \frac{\bbE[C'_1 \mid D = d, C] }{\bbE[C'_i \mid D = d , C] } \geq
        \begin{cases}
             \frac{C_1}{C_i} + \frac{1}{7}\pa{ \frac{C_1}{C_i} - 1 } \quad & \text{if }\frac{p_1}{p_i}<2
             \\
             \frac{8}{7}\frac{C_1}{C_i} \quad & \text{if }\frac{p_1}{p_i}>2
        \end{cases}
    \]
\end{lemma}
\begin{proof}
    By \Cref{lemma:tight_bound_probability_collision} and by our hypothesis, we know that $Q \geq K h^2 \norm{p}_2^2$, for some constant $K>0$. Therefore, conditioning on the event $D = d$, we have that $\frac{d}{n} \geq \frac{K}{2} h^2 \norm{p}_2^2$.
We have
\begin{align*}
        \frac{\bbE[C'_1 \mid C] }{\bbE[C'_i \mid C] }, \frac{\bbE[C'_1 \mid D = d, C] }{\bbE[C'_i \mid D = d , C] }
        &=
        \frac{\frac{q_1}{Q}\cdot Q  + p_1(1-Q)}{\frac{q_i}{Q}\cdot Q  + p_i(1-Q)},
        \frac{\frac{q_1}{Q}\cdot \frac{d}{n}  + p_1(1-\frac{d}{n})}{\frac{q_i}{Q}\cdot \frac{d}{n}  + p_i(1-\frac{d}{n})}
        \\
        &\geq
        \frac{\frac{q_1}{Q}\cdot Q  + p_1}{\frac{q_i}{Q}\cdot Q  + p_i},
        \frac{\frac{q_1}{Q}\cdot \frac{d}{n}  + p_1}{\frac{q_i}{Q}\cdot \frac{d}{n}  + p_i}
         , \text{ b/c } \frac{q_1}{q_i}\geq \frac{p_1}{p_i}
         \\
        &\geq
        \frac{\frac{q_1}{Q}\cdot \frac{K}{2}\cdot h^2 \sum_j p_j^2 + p_1 }{\frac{q_i}{Q}\cdot \frac{K}{2} \cdot h^2 \sum_j p_j^2 + p_i}
        , \text{ b/c } Q, \frac{d}{n} \geq \frac{K}{2} h^2 \norm{p}_2^2 \text{ and } \frac{q_1}{q_i}\geq \frac{p_1}{p_i}
        \\
        &=
        \frac{\frac{1}{\sum_j q_j/q_1}\cdot \frac{K}{2} h^2 \sum_j p_j^2 + p_1}
        {\frac{q_i/q_1}{\sum_j q_j/q_1}\cdot \frac{K}{2} h^2 \sum_j p_j^2 + p_i }
        \\&
        \geq
        \frac{\frac{1}{\sum_j 3p_j^2/p_1^2}\cdot \frac{K}{2} h^2 \sum_j p_j^2 + p_1}
        {\frac{q_i/q_1}{\sum_j 3p_j^2/p_1^2}\cdot \frac{K}{2} h^2 \sum_j p_j^2 + p_i}
        , \text{ by } \cref{eq:qdqlb2} \text{ b/c } \frac{q_1}{q_j}\geq \frac{p_1}{p_j}
        \\&
        =
        \frac{p_1^2\cdot \frac{K}{2} h^2 + 3p_1}
        {p_1^2\cdot \frac{q_i}{q_1}\cdot  \frac{K}{2} h^2 + 3p_i}
        \\&
        \geq
        \frac{p_1^2\cdot \frac{K}{2} h^2 + 3p_1}
        {p_i\cdot \frac{p_1+p_i}{2}\cdot  \frac{K}{2} h^2 + 3p_i}
        , \text{ by } \cref{eq:qdqlb2}
        \\&
        =
        \frac{p_1}{p_i}\cdot
        \frac{p_1\cdot \frac{K}{2} h^2 + 3}
        {\frac{p_1+p_i}{2}\cdot \frac{K}{2} h^2 + 3}
        \\&
        =
        \frac{p_1}{p_i}\cdot
        \frac{p_1\cdot \frac{K}{2} h^2 + 3}
        {\frac{1}{2}\pa{1+\frac{p_i}{p_1}}p_1\cdot \frac{K}{2} h^2 + 3}
        \\&
        \geq
        \frac{p_1}{p_i}\cdot
        \frac{6}
        {\frac{1}{2}\pa{1+\frac{p_i}{p_1}}3 + 3}
        , \text{ b/c } \frac{K}{2} p_1 h^2 \geq 3
        \\&
        =
        \frac{p_1}{p_i}\cdot
        \frac{4}
        {3+ \frac{p_i}{p_1}}.
    \end{align*}
    If $\frac{p_1}{p_i}<2$, we have that $\frac{4}
        {3+ \frac{p_i}{p_1}} \geq 1 + \frac{1}{7}(\frac{p_1}{p_i}-1)$, and, therefore, we obtain that
    \[
        \frac{\bbE[C'_1 \mid D = d, C] }{\bbE[C'_i \mid D = d , C] } \geq \frac{p_1}{p_i} + \frac{1}{7}\pa{\frac{p_1}{p_i} -1}.
    \]
    If, instead, $\frac{p_1}{p_i}\geq 2$, we have that
    \begin{align*}
        \frac{\bbE[C'_1 \mid D = d, C] }{\bbE[C'_i \mid D = d , C] }
        &\geq
        \frac{p_1}{p_i}\cdot
        \frac{4}
        {3+ \frac{p_i}{p_1}}
        \\
        &\geq
        \frac{p_1}{p_i}\cdot
        \frac{8}
        {7}.
    \end{align*}
\end{proof}

\begin{lemma}
    Let \(C\) be an ordered configuration.
    Let \(K\coloneq 2^{-11}\).
    Let \(c\coloneq \frac{K}{2(K+4)}\).
    If $h^2\norm{p}_2^2 >1$, then for any integer $d\in [\frac{\bbE[D]}{2},\frac{3\bbE[D]}{2}]$
    \[
        \frac{\bbE[C'_1 \mid C] }{\bbE[C'_i \mid C] }, 
        \frac{\bbE[C'_1 \mid D = d, C] }{\bbE[C'_i \mid D = d , C] } \geq
        \begin{cases}
             \frac{C_1}{C_i} + (1+c) \pa{ \frac{C_1}{C_i} - 1 } \quad & \text{if }\frac{q_i}{Q}\geq \frac{1}{2} p_i
             \\
             (1+c) \frac{C_1}{C_i} \quad & \text{if }\frac{q_i}{Q}< \frac{1}{2} p_i
        \end{cases}
    \]
\end{lemma}
\begin{proof}
    By \Cref{lemma:tight_bound_probability_collision} and our hypothesis, we know that $Q\geq K$, for some constant $K>0$. Therefore, conditioning on the event $D = d$, we have that $\frac{d}{n} \geq \frac{K}{2}$.
We have
\begin{align*}
        \frac{\bbE[C'_1 \mid C] }{\bbE[C'_i \mid C] }, \frac{\bbE[C'_1 \mid D = d, C] }{\bbE[C'_i \mid D = d , C] }
        &=
        \frac{\frac{q_1}{Q}\cdot Q  + p_1(1-Q)}{\frac{q_i}{Q}\cdot Q  + p_i(1-Q)},
        \frac{\frac{q_1}{Q}\cdot \frac{d}{n}  + p_1(1-\frac{d}{n})}{\frac{q_i}{Q}\cdot \frac{d}{n}  + p_i(1-\frac{d}{n})}
        \\&
        \geq
        \frac{\frac{q_1}{Q}\cdot \frac{K}{2}  + p_1}{\frac{q_i}{Q}\cdot \frac{K}{2}  + p_i}
         , \text{ b/c } \frac{q_1}{q_i}\geq \frac{p_1}{p_i}.
    \end{align*}
Therefore, if $\frac{q_i}{Q}\geq \frac{1}{2} p_i$, we obtain
\begin{align*}
        \frac{\bbE[C'_1 \mid C] }{\bbE[C'_i \mid C] }, \frac{\bbE[C'_1 \mid D = d, C] }{\bbE[C'_i \mid D = d , C] }&
        \geq
        \frac{\frac{q_1}{Q}\cdot \frac{K}{2}  + p_1}{\frac{q_i}{Q}\cdot \frac{K}{2}  + p_i}
        \\
        &=
        \frac{p_1}{p_i} + \frac{q_i \frac{K}{2Q}}{q_i\frac{K}{2Q} + p_i}\pa{\frac{q_1}{q_i} - \frac{p_1}{p_i}}
        \\
        &\geq
        \frac{p_1}{p_i} + \frac{K p_i/4}{K p_i/4 + p_i}\pa{\frac{q_1}{q_i} - \frac{p_1}{p_i}}
        \\
        &\geq
        \frac{p_1}{p_i} + \frac{K}{K+4}\pa{\frac{p_1^2}{p_i^2}\frac{(p_1+3 p_i)}{(3p_1 + p_i)} - \frac{p_1}{p_i}}
        , \text{ by \Cref{eq:qdqlb2}}
        \\
        &\geq
        \frac{p_1}{p_i} + \frac{K}{2(K+4)}\pa{\frac{p_1}{p_i} - 1}
        , \text{ b/c } \frac{x(x+3)}{3x+1} \ge \frac{3}{2} - \frac{1}{2x} \quad \text{if \(x \ge 1\)}.
    \end{align*}
If, instead, $\frac{q_i}{Q} < \frac{1}{2} p_i$, we have
\begin{align*}
        \frac{\bbE[C'_1 \mid C] }{\bbE[C'_i \mid C] }, \frac{\bbE[C'_1 \mid D = d, C] }{\bbE[C'_i \mid D = d , C] }&
        \geq
        \frac{\frac{q_1}{Q}\cdot \frac{K}{2}  + p_1}{\frac{q_i}{Q}\cdot \frac{K}{2}  + p_i}
        \\
        &=
        \frac{p_1\pa{\frac{q_1 K}{ 2 p_1 Q} +1 }}{ p_i \pa{ \frac{q_i K}{ 2 p_i Q} +1} }
        \\
        &\geq
        \frac{p_1\pa{\frac{q_1 K}{ 2 p_1 Q} +1 }}{p_i \pa{ \frac{K}{4} +1 }}
        \\
        &\geq
        \frac{p_1\pa{\frac{K}{ 2} +1 }}{p_i \pa{ \frac{K}{4} +1 }}
        , \text{ b/c } \frac{q_1}{Q}>p_1.
    \end{align*}
    Since
    \[
        \frac{\frac{K}{2}+1}{\frac{K}{4}+1}
        =
        \frac{2K+4}{K+4}
        =
        1+\frac{K}{K+4}
        \ge
        1+\frac{K}{2(K+4)}
        =
        1+c,
    \]
    this proves the second branch of the claim.
\end{proof}

\begin{proof}[Proof of \Cref{lemma:growth_expected_ratio}]
    If \(h^2\norm{p}_2^2<1\) and \(Kp_1h^2<6\), the claim follows from the first sublemma, since \(p_1=C_1/n\).

    If \(h^2\norm{p}_2^2<1\) and \(Kp_1h^2\ge 6\), then the second sublemma gives
    \[
        \frac{\bbE[C_1'\mid C]}{\bbE[C_i'\mid C]}
        \ge
        \frac{C_1}{C_i}+\frac{1}{7}\pa{\frac{C_1}{C_i}-1}
    \]
    or even
    \[
        \frac{\bbE[C_1'\mid C]}{\bbE[C_i'\mid C]}
        \ge
        \frac{8}{7}\frac{C_1}{C_i}.
    \]
    In either case,
    \[
        \frac{\bbE[C_1'\mid C]}{\bbE[C_i'\mid C]}
        \ge
        \frac{C_1}{C_i}
        +\eta \min\cpa{\frac{C_1}{n}h^2,1}\pa{\frac{C_1}{C_i}-1},
    \]
    because \(\eta\le 1/7\) and \(\min\cpa{\frac{C_1}{n}h^2,1}\le 1\). The same argument applies conditional on \(D=d\).

    Finally, if \(h^2\norm{p}_2^2>1\), then \(p_1\ge \norm{p}_2^2\), hence \(h^2p_1\ge h^2\norm{p}_2^2>1\), so
    \[
        \min\cpa{\frac{C_1}{n}h^2,1}=1.
    \]
    The third sublemma yields either
    \[
        \frac{\bbE[C_1'\mid C]}{\bbE[C_i'\mid C]}
        \ge
        \frac{C_1}{C_i}+\frac{K}{2(K+4)}\pa{\frac{C_1}{C_i}-1},
    \]
    or
    \[
        \frac{\bbE[C_1'\mid C]}{\bbE[C_i'\mid C]}
        \ge
        \pa{1+\frac{K}{K+4}}\frac{C_1}{C_i}
        \ge
        \frac{C_1}{C_i}+\frac{K}{K+4}\pa{\frac{C_1}{C_i}-1}.
    \]
    Since \(\eta\le \frac{K}{2(K+4)}\), this proves the claim. Again, the same argument works conditional on \(D=d\).
\end{proof}

\subsection{Expected amplification of the additive bias}\label{sec:expected-additive-amplification}
The next lemma translates the expected increase in multiplicative bias into an increase in additive bias, provided that $C_1$ is not too large. Moreover, it quantifies the growth of the majority opinion $C_1$ and of $\frac{\Delta'_j}{\sqrt{C'_1}}$. The growth of the latter quantity will be useful to ensure the bias at the next round continues to satisfy the initial condition.
\begin{lemma}
    \label{lemma:growth_expectation}
    Let \(\gamma\coloneq 2^{-19}\). Then the following holds.
    Consider a configuration $C$ such that $C_1\geq C_2 \geq \ldots \geq C_k>0$. Furthermore, suppose $C_1 < \frac{3}{4}n$. Then, for every \(j\in\cpa{2,\ldots,k}\), we have
    \begin{align}
        \label{eq:bias_growth_balanced_expectation}
        &\bbE\qty[ \Delta'_j \mid C]
        \geq
        \Delta_j \pa{1 + \gamma \min\cpa{\frac{C_1}{n}h^2, 1}},
        \\
        \label{eq:bias_condition_preserved_balanced}
        &\frac{\bbE\qty[ \Delta'_j \mid C]}{\sqrt{\bbE\qty[C'_1]}}
        \geq
         \frac{\Delta_2}{\sqrt{C_1}}\pa{1 + \gamma/3 \min\cpa{\frac{C_1}{n}h^2, 1}},
        \\
        \label{eq:bias_condition_preserved_unbalanced}
        &\frac{\bbE[\Delta_j'\mid C]}{\sqrt{\bbE[C_1'\mid C]}} \geq  \frac{\Delta_j}{\sqrt{C_1}}\sqrt{\frac{\bbE[C_1'\mid C]}{C_1}}
        \\
        \label{eq:growth_majority_opinion_expectation}
        &\bbE[C_1' \mid C] \geq C_1 + \frac{\gamma \alpha(C_1,h)}{7} \pa{ C_1 - C_2}
    \end{align}
    Moreover, the bound in \Cref{eq:bias_condition_preserved_balanced} is uniform in \(j\), since \(\Delta_j \ge \Delta_2\) for every \(j\ge 2\).
\end{lemma}
\begin{proof}
Let \(\alpha\coloneq \alpha(C_1,h)= \eta \min\cpa{\frac{C_1}{n}h^2, 1}\), for \(\eta=2^{-19}\).
By \Cref{lemma:growth_expected_ratio}, we have that for all $j\in\cpa{2,\ldots,k}$
\[
    \frac{\bbE[C_1' \mid C]}{\bbE[C_j' \mid C]} \geq \frac{C_1}{C_j} + \alpha \frac{C_1-C_j}{C_j} \geq \frac{C_1}{C_j} + \alpha \frac{C_1-C_2}{C_j},
\]
and that $\bbE\qty[C_1' \mid C] \geq C_1$.
Take reciprocal and sum over $j\in[k]$ and obtain that
\[
    \frac{n}{\bbE[C_1' \mid C]} \leq \frac{n-C_1}{C_1+\alpha(C_1-C_2)} +1.
\]
This rewrites as
\[
    \bbE[C_1' \mid C] \geq \pa{C_1 + \alpha \pa{C_1 - C_2}}\frac{n}{n+ \alpha\pa{C_1-C_2}}.
\]
Since \(\alpha< 1\) and we are assuming that \(C_1< \frac{3}{4} n\), we have \(\frac{7 C_1 + \alpha(C_1 - C_2)}{6} < \frac{8C_1}{6}<  n\). It is easy to check that this implies
\begin{equation*}
    \bbE[C_1' \mid C] \geq C_1 + \frac{\alpha}{7} \pa{ C_1 - C_2},
\end{equation*}
proving \Cref{eq:growth_majority_opinion_expectation}.
By applying \Cref{lemma:growth_expected_ratio} again, we obtain for all $j\in\cpa{2,\ldots,k}$
\begin{align}
    \nonumber
    \bbE[C_1'-C'_j \mid C] 
    &= 
    \bbE[C_1' \mid C]\pa{1-\frac{\bbE[C_j' \mid C]}{\bbE[C_1' \mid C]}} 
    \\
    \label{eq:expectation_bias}
    &\geq 
    \bbE[C_1' \mid C]\pa{ 1- \frac{C_j}{C_1 + \frac{\alpha}{7}\pa{C_1-C_j}} }
\end{align}
This implies that
\begin{align*}
    \bbE[C_1'-C'_j \mid C] 
    &\geq 
    \pa{C_1 + \frac{\alpha}{7} \pa{ C_1 - C_2}}\pa{ \frac{(C_1-C_j)(1+\frac{\alpha}{7})}{C_1 + \frac{\alpha}{7}\pa{C_1-C_j}} }
    &\pa{ \text{by \Cref{eq:growth_majority_opinion_expectation}}}
    \\
    &\geq 
    (C_1-C_j)(1+\frac{\alpha}{7}),
\end{align*}
proving \Cref{eq:bias_growth_balanced_expectation}.
\Cref{eq:expectation_bias,eq:growth_majority_opinion_expectation} imply also that
\begin{align*}
    \bbE[C_1'-C'_j \mid C] 
    &\geq 
    \bbE[C_1'-C'_2 \mid C] 
    \\
    &\geq
    \sqrt{\bbE[C_1' \mid C] } \frac{\pa{C_1-C_2}(1+\frac{\alpha}{7})}{\sqrt{C_1 + \frac{\alpha}{7} \pa{C_1-C_2} }}
    \\
    &=
    \sqrt{\bbE[C_1' \mid C] } \frac{\Delta_2}{\sqrt{C_1}} \frac{1+\frac{\alpha}{7}}{\sqrt{1 + \frac{\alpha}{7} \frac{C_1-C_2}{C_1} }}
    \\
    &\geq
    \sqrt{\bbE[C_1' \mid C] } \frac{\Delta_2}{\sqrt{C_1}} \sqrt{1+\frac{\alpha}{7}}
    \\
    &\geq
    \sqrt{\bbE[C_1' \mid C] } \frac{\Delta_2}{\sqrt{C_1}} \pa{1+\frac{\alpha}{21}}
    &\pa{\sqrt{1+x}\geq 1+\frac{x}{3} \text{ for } x\in[0,3]}
\end{align*}
proving \Cref{eq:bias_condition_preserved_balanced}.
By \Cref{eq:expectation_bias} we have that
    \begin{align*}
        \bbE[\Delta_j'\mid C]
        &\geq
        \bbE[C_1'\mid C] \pa{ \frac{\Delta_j(1+\alpha/7)}{C_1+\frac{\alpha}{7}\Delta_j } }
        \\
        &\geq
        \bbE[C_1'\mid C] \pa{ \frac{\Delta_j}{C_1} }
        &\pa{ \Delta_j \leq C_1 }.
    \end{align*}
    Rearranging the last equation we proved \Cref{eq:bias_condition_preserved_unbalanced} and concluded the proof of \Cref{lemma:growth_expectation}.
\end{proof}

In the next lemma we show that whenever the configuration is unbalanced, the bias grows considerably more.

\begin{lemma}[Bias growth in the genuinely unbalanced regime]
    \label{lemma:unbalanced_expected_bias}
    Suppose the current ordered configuration is \(C=(C_1,\ldots,C_k)\), with \(C_1=\omega(\log n)\) and \(C_1\le \frac45 n\).
    Fix an opinion \(i\in\cpa{2,\ldots,k}\) such that \(C_i>0\).
    Assume
    \[
        \norm{p}_2^2<\frac{p_1+p_i}{24},
        \qquad
        \Delta_i \ge \lambda \sqrt{\max\cpa{\frac{n}{h^2}, C_1}\log n}
    \]
    for a sufficiently large absolute constant \(\lambda>0\). Then
    \[
        \bbE[\Delta_i'\mid C]
        \ge
        \Delta_i\pa{1 + \frac{(p_1+p_i)\min\cpa{ h^2 \norm{p}_2^2, 1}}{2^{12}\norm{p}_2^2}}.
    \]
\end{lemma}
\begin{proof}
    Since the following three facts hold
    \[
        \frac{q_1-q_i}{Q}
        =
        \frac{q_1+q_i}{Q}\cdot \frac{\frac{q_1}{q_i}-1}{\frac{q_1}{q_i}+1},
    \]
    the function \(f(x)\coloneq (x-1)/(x+1)\) is increasing on \((0,\infty)\), and \(
        \frac{q_1}{q_i} \ge \frac{2p_1^2}{p_i(p_1+p_i)},
    \) 
    by \Cref{eq:qdqlb2},
    we obtain that
    
    \[
        \frac{\frac{q_1}{q_i}-1}{\frac{q_1}{q_i}+1}
        \ge
        f\!\left(\frac{2p_1^2}{p_i(p_1+p_i)}\right)
        =
        \frac{\frac{2p_1^2}{p_i(p_1+p_i)}-1}{\frac{2p_1^2}{p_i(p_1+p_i)}+1}.
    \]
    Therefore
    \[
        \frac{q_1-q_i}{Q}
        \ge
        \frac{q_1+q_i}{Q}\cdot
        \frac{\frac{2p_1^2}{p_i(p_1+p_i)}-1}{\frac{2p_1^2}{p_i(p_1+p_i)}+1}.
    \]
    By \Cref{claim:q_i_over_Q_ub_and_lb_norm2},
    \(
        \frac{q_1+q_i}{Q}\ge \frac{p_1^2+p_i^2}{3\norm{p}_2^2}
    \), and we obtain
    \[
        \frac{q_1-q_i}{Q}
        \ge
        \frac{p_1^2+p_i^2}{3\norm{p}_2^2}\cdot
        \frac{\frac{2p_1^2}{p_i(p_1+p_i)}-1}{\frac{2p_1^2}{p_i(p_1+p_i)}+1}.
    \]
    Moreover,
    \[
        \frac{\frac{2p_1^2}{p_i(p_1+p_i)}-1}{\frac{2p_1^2}{p_i(p_1+p_i)}+1}
        =
        \frac{(p_1-p_i)(2p_1+p_i)}{2p_1^2+p_1p_i+p_i^2}
        \ge
        \frac{p_1-p_i}{p_1+p_i},
    \]
    where the last inequality is equivalent to
    \((2p_1+p_i)(p_1+p_i)\ge 2p_1^2+p_1p_i+p_i^2\).
    Hence
    \[
        \frac{q_1-q_i}{Q}
        \ge
        \frac{p_1^2+p_i^2}{3\norm{p}_2^2}\cdot \frac{p_1-p_i}{p_1+p_i}
        \ge
        \frac{(p_1-p_i)(p_1+p_i)}{6\norm{p}_2^2},
    \]
    where the last step uses \(p_1^2+p_i^2\ge (p_1+p_i)^2/2\). Since we have
    \[
        \bbE[C_r'\mid C]
        =
        n Q \frac{q_r}{Q}+(1-Q) C_r,
    \]
    we obtain that
    \[
        \bbE[\Delta_i'\mid C]-\Delta_i
        =
        nQ\pa{\frac{q_1-q_i}{Q}-(p_1-p_i)}
        \ge
        nQ(p_1-p_i)\pa{\frac{p_1+p_i}{6\norm{p}_2^2}-1}.
    \]
    Since by hypothesis \(\norm{p}_2^2<(p_1+p_i)/24\), we get
    \[
        \bbE[\Delta_i'\mid C]-\Delta_i
        \ge
        \frac{Q(p_1-p_i)(p_1+p_i)}{2\norm{p}_2^2}.
    \]
    Since by \Cref{lemma:tight_bound_probability_collision} $Q \geq 2^{-11}\min\cpa{ h^2 \norm{p}_2^2, 1}$, we conclude the proof of \Cref{lemma:unbalanced_expected_bias}.
\end{proof}

\begin{lemma}
    \label{claim:expectation_growth_C_1_unbalanced}
    Let \(\alpha=\alpha(C_1,h)=\min\cpa{\frac{C_1}{n}h^2,1}\).
    Let $C$ be an ordered configuration such that $nC_1 \geq 24 \|C\|_2^2$. It holds 
    \[
        \bbE[C'_1\mid C] \ge C_1 \pa{ 1 + \gamma \alpha \frac{7}{24} \frac{nC_1}{\|C\|_2^2} }.   
    \]
\end{lemma}
\begin{proof}
    By \Cref{claim:q_i_over_Q_ub_and_lb_norm2} it holds that 
    \[
        \bbE[C'_1\mid C] = (1-Q)C_1 + nQ \frac{q_1}{Q} \ge C_1\pa{ 1+ Q\pa{ \frac{nC_1}{3\|C\|_2^2} -1 } }.
    \]
    If we now apply \Cref{lemma:tight_bound_probability_collision}, we obtain that
    \[
        \bbE[C'_1\mid C] \ge C_1\pa{ 1+ \gamma \alpha \pa{ \frac{7}{24}\frac{nC_1}{\|C\|_2^2} + \frac{1}{24}\frac{nC_1}{\|C\|_2^2} -1 } },
    \]
    and using the hypothesis $nC_1 \geq 24 \|C\|_2^2$ we conclude the proof of \Cref{claim:expectation_growth_C_1_unbalanced}.
\end{proof}

\subsection{Amplification in concentration}
\label{sec:amplification_in_concentration}
In this section, we will prove that the bias amplifies w.h.p. by demonstrating that it is concentrated around the expectation computed in \Cref{lemma:growth_expectation}. Furthermore, we demonstrate that, in the subsequent round, the bias and the counter of the largest opinion satisfy the minimum requirements necessary to enable the iteration of our amplification analysis.

We start quantifying the deviation from the mean of each opinion count. The main tool we use is Bernstein’s inequality.
\begin{lemma}
    \label{lemma:concentration_all_opinions}
    Let $\alpha = \alpha(C_1,h) = \min\cpa{\frac{C_1}{n}h^2, 1}$. Assume $\alpha C_1 \geq \log n$, and $C_1\geq C_j$ for all $j\in[k]$. Let $\mathcal{E}$ be the event
    \[
         \mathcal{E} = \bigcap_{j\in[k]} \cpa{ \qty|C_j' - \bbE[C'_j \mid C]| \leq \lambda_1 \sqrt{\frac{\alpha n C_1^2 \log n}{\|C\|_2^2}} },
    \]
    for some constant $\lambda_1$ large enough.
    We have that $\Pr{\mathcal{E}\mid C} \geq 1 - n^{-100}$.
\end{lemma}
\begin{proof}
    By \Cref{lemma:bound_variance}, it holds for all $j\in[k]$
    \begin{align*}
        \Var[C_j'\mid C] 
        &\leq 
        \Var[C_1'\mid C] 
        \\
        &= 
        O \qty( \frac{n C_1^2}{\|C\|_2^2} \min\cpa{ h^2 \|C/n\|^2_2, 1 }  ) 
        \\
        &= 
        O \qty( \frac{n \alpha C_1^2}{\|C\|_2^2} ),
    \end{align*}
    where in the last inequality we used the fact that $\|C/n\|^2_2 \leq C_1 / n$. By \Cref{lemma:bernstein_non_identical}, we obtain that
    \[
        \pr{ \qty|C_j' - \bbE[C_j'\mid C]| \geq t } \leq \frac{1}{n^{101}},
    \]
    for $t=\lambda_1 \max\cpa{ \log n, \sqrt{\frac{n \alpha C_1^2 \log n}{\|C\|_2^2}} }$, for a constant $\lambda_1$ large enough. Since $\alpha C_1 \geq  \log n$ and $\|C\|^2_2 \leq n C_1$, we obtain $t=\lambda_1 \sqrt{\frac{n \alpha C_1^2 \log n}{\|C\|_2^2}}$, as
    \[
        \sqrt{\frac{n \alpha C_1^2 \log n}{\|C\|_2^2}} \geq \log n.
    \]
    The proof of \Cref{lemma:concentration_all_opinions} now follows by a union bound.
\end{proof}

We continue with the additive version of the bias amplification which holds with high probability.
\begin{lemma}
    [Bias amplification in concentration]
    \label{lemma:bias_concentration_bernstein}
    Let \(\mathcal{E}\) be the event defined in \Cref{lemma:concentration_all_opinions}.
    Suppose the current ordered configuration is \(C = (C_1, \ldots, C_k)\).
    Also, assume $\alpha C_1 = \omega(\log n)$ and $C_1 \leq \frac{4}{5}n$.
    Furthermore, assume that for every opinion \(i\in\{2,\ldots, k\}\), we have
    \[
        \Delta_i \ge \lambda \sqrt{\max\cpa{ \frac{n}{h^2}, C_1 } \log n}
    \]
    for a large enough constant \(\lambda > 0\).
    Then there exists an absolute constant \(c_B>0\) such that
    \[
        \pr{
            \cap_{i\in \cpa{2,\ldots,k}} \cpa{\Delta'_i \geq \Delta_i \pa{1+ c_B\min\cpa{\frac{C_1}{n}h^2, 1}}}
            \,\middle|\,
            C, \mathcal{E}
        }
        =1.
    \]
\end{lemma}
\begin{proof}
    First, notice that the hypothesis on the bias implies that
    \begin{equation}
        \label{eq:condition_on_bias_function_of_alpha}
        \Delta_i \geq \lambda \pa{ \sqrt{\frac{C_1 \log n}{\alpha} } }.
    \end{equation}
    Let $t = \lambda_1 \sqrt{\frac{n \alpha C_1^2 \log n}{\|C\|_2^2}} $. 
    First consider the case $\|C\|_2^2 \leq n C_1 \leq 24 \|C\|_2^2$, and therefore $t \leq \lambda_2 \sqrt{ \alpha C_1 \log n}$, for $\lambda_2 = \lambda_1 \sqrt{24}$.
    Conditioning on $\mathcal{E}$, we have that by \Cref{lemma:growth_expectation} (\Cref{eq:bias_growth_balanced_expectation})
    \begin{align*}
        \Delta'_i 
        &\geq 
        \bbE[\Delta'_i \mid C] - 2t 
        \\
        &\geq \Delta_i\pa{ 1+ \gamma \alpha(C_1,h) } - 2t
        \\
        &\geq \Delta_i\pa{ 1+ \frac{\gamma}{2} \alpha(C_1,h) } + (\lambda \frac{\gamma}{2} - 2\lambda_2) \sqrt{\alpha(C_1,h)C_1 \log n}.
        &\pa{\text{by \Cref{eq:condition_on_bias_function_of_alpha}}}
    \end{align*}
    By taking $\lambda$ s.t. $\lambda \geq 4 \lambda_2 \gamma^{-1}$, we proved the claim whenever $\|C\|_2^2 \leq n C_1 \leq 24 \|C\|_2^2$.
    Now consider the remaining case, $n C_1 \geq 24\|C\|_2^2$. 
    Conditioning on $\mathcal{E}$, we have that by \Cref{lemma:unbalanced_expected_bias}
    \begin{align*}
        \Delta'_i
        & \geq 
        \Delta_i\pa{ 1+ \frac{n C_1\min\cpa{ h^2 \|C/n\|^2_2 , 1 }}{2^{12} \| C \|^2_2 } } - 2t
        \\
        & \geq
        \Delta_i\pa{ 1+ 2^{-12}\alpha } - 2t
        &\pa{ nC_1 \geq \|C\|^2_2 }
        \\
        &\geq
        \Delta_i\pa{ 1+ 2^{-13}\alpha } + \lambda 2^{-13}\sqrt{\alpha C_1 \log n} - 2\lambda_1 \sqrt{\frac{n \alpha C_1^2 \log n}{\|C\|_2^2}}
        \\
        &\geq
        \Delta_i\pa{ 1+ 2^{-13}\alpha } + \pa{\lambda 2^{-13} - 2\lambda_1} \sqrt{\alpha C_1 \log n},
        &\pa{\|C\|^2_2 \leq nC_1}
    \end{align*}
   which concludes the proof of \Cref{lemma:bias_concentration_bernstein} if we take $\lambda \geq 2^{14}\lambda_1$.
\end{proof}
In the next lemma we show that, w.h.p., $C_1' \ge C_1$, which enforces the condition $C_1' = \omega(\log n)$ at the next round.
\begin{lemma}
    [Majority opinion does not decrease in concentration]
    \label{lemma:majority_concentration}
    Let \(\mathcal{E}\) be the event defined in \Cref{lemma:concentration_all_opinions}.
    Suppose the current ordered configuration is \(C = (C_1, \ldots, C_k)\).
    Also, assume $\alpha C_1 = \omega(\log n)$ and $C_1 \leq \frac{4}{5}n$.
    Then $\pr{C_1'\geq C_1 \mid \mathcal{E},C} =1$.
\end{lemma}
\begin{proof}
    Let \(\alpha\coloneq \alpha(C_1,h)=\min\cpa{\frac{C_1}{n}h^2,1}\).
    Let $t = \lambda_1 \sqrt{\frac{n \alpha C_1^2 \log n}{\|C\|_2^2}} $. 
    First consider the case $\|C\|_2^2 \leq n C_1 \leq 24 \|C\|_2^2$, and therefore $t \leq \lambda_2 \sqrt{ \alpha C_1 \log n}$, for $\lambda_2 = \lambda_1 \sqrt{24}$.
    Conditioning on $\mathcal{E}$, we have that by \Cref{lemma:growth_expectation} (\Cref{eq:growth_majority_opinion_expectation})
    \begin{align*}
        C'_1
        &\geq 
        \bbE[C'_1 \mid C] - t 
        \\
        &\geq C_1\pa{ 1+ \frac{\gamma}{7} \alpha } - t
        \\
        &\geq C_1 + \alpha C_1 \pa{ \frac{\gamma}{7} - \lambda_2 \sqrt{\frac{\log n}{\alpha C_1}} }
        \geq C_1.
        &\pa{\alpha C_1 = \omega(\log n)}
    \end{align*}
    Now consider the remaining case, $n C_1 \geq 24\|C\|_2^2$. 
    Conditioning on $\mathcal{E}$, we have that by \Cref{claim:expectation_growth_C_1_unbalanced}
    \begin{align*}
        C'_1
        & \geq 
        C_1\pa{ 1+ \gamma \alpha \frac{7}{24} \frac{nC_1}{\|C\|_2^2}} - t
        \\
        & \geq
        C_1 + \alpha C_1 \sqrt{\frac{nC_1}{\|C\|_2^2}}\pa{ \frac{7 \gamma}{24}\sqrt{\frac{nC_1}{\|C\|_2^2} -  \frac{\lambda_1 \log n}{\alpha C_1} }}
        \\
        &\geq C_1,
        &\pa{ n C_1 \geq 24\|C\|_2^2, \alpha C_1 = \omega(\log n)}
    \end{align*}
   which concludes the proof of \Cref{lemma:majority_concentration}.
\end{proof}

In \Cref{lemma:bias_concentration_bernstein}, we showed that the bias at the next round increases by a multiplicative factor if the current bias satisfies $\Delta_i\ge \lambda \sqrt{\max\cpa{\frac{n}{h^2},C_1}\log n}$.  However, showing that $\Delta_i' \ge \Delta_i$ is not enough to ensure that this condition holds at the next round, as the right hand side depends on $C_1$, which we know grows. Therefore, in the next lemma we show that w.h.p. the bias at the next round satisfies the condition.

\begin{lemma}[Condition on bias is preserved]
    \label{lemma:condition_delta_preserved}
    Let \(\mathcal{E}\) be the event defined in \Cref{lemma:concentration_all_opinions}.
    Suppose the current ordered configuration is \(C=(C_1,\ldots,C_k)\), with
    \[
        \alpha C_1 \geq \log n,
        \qquad
        C_1\le \frac45 n.
    \]
    Assume that for every opinion \(i\in\cpa{2,\ldots,k}\) it holds
    \[
        \Delta_i\ge \lambda \sqrt{\max\cpa{\frac{n}{h^2},C_1}\log n},
    \]
    for a sufficiently large absolute constant \(\lambda>0\), it holds that
    \[
        \pr{
            \cap_{i\in \cpa{2,\ldots,k}}\cpa{\Delta_i' \geq \lambda \sqrt{\max\cpa{\frac{n}{h^2},C_1'}\log n}}
            \,\middle|\,
            C, \mathcal{E}
        }
        = 1.
    \]
\end{lemma}
\begin{proof}
    Let \(\alpha\coloneq \alpha(C_1,h)=\min\cpa{\frac{C_1}{n}h^2,1}\).
    Fix an opinion $i\in\cpa{2,\ldots,k}$. Let $t = \lambda_1 \sqrt{\frac{n \alpha C_1^2 \log n}{\|C\|_2^2}} $. Conditioning on $\mathcal{E}$, we have that
    \begin{align}
        \nonumber
        \frac{\Delta_i'}{\sqrt{C_1'}} 
        &\geq 
        \frac{\bbE[\Delta_i'\mid C]-2t}{\sqrt{\bbE[C_1'\mid C] + t}}
        \\ \nonumber
        &\geq
        \frac{\bbE[\Delta_i'\mid C]}{\sqrt{\bbE[C_1'\mid C]}} - t \frac{\bbE[\Delta_i'\mid C] + 4 \bbE[C_1'\mid C]}{2 \bbE[C_1'\mid C] \sqrt{\bbE[C_1'\mid C]}}
        &\pa{\text{Taylor expansion and convexity}}
        \\ \nonumber
        &=
        \frac{\bbE[\Delta_i'\mid C]}{\sqrt{\bbE[C_1'\mid C]}} - t \frac{5\bbE[C_1'\mid C] - \bbE[C_i'\mid C]}{2 \bbE[C_1'\mid C] \sqrt{\bbE[C_1'\mid C]}}
        \\ \label{eq:first_opinion_divides_bias_taylor_expansion_2}
        &\geq
        \frac{\bbE[\Delta_i'\mid C]}{\sqrt{\bbE[C_1'\mid C]}} - \frac{5 t}{2 \sqrt{\bbE[C_1'\mid C]}}
        \\ 
        \label{eq:first_opinion_divides_bias_taylor_expansion}
        &=
        \frac{\bbE[\Delta_i'\mid C]}{\sqrt{\bbE[C_1'\mid C]}} \pa{ 1- \frac{5 t}{2 \bbE[C_1'-C_i'\mid C]} }.
    \end{align}
    First consider the case $\|C\|_2^2 \leq n C_1 \leq 24 \|C\|_2^2$, and therefore $t \leq \lambda_2 \sqrt{ \alpha C_1 \log n}$, for $\lambda_2 = \lambda_1 \sqrt{24}$. By \Cref{lemma:growth_expectation} (\Cref{eq:bias_condition_preserved_balanced}), \Cref{eq:first_opinion_divides_bias_taylor_expansion} implies
    \begin{align*}
        \frac{\Delta_i'}{\sqrt{C_1'}} 
        &\geq 
        \frac{\Delta_2}{\sqrt{C_1}}\pa{1 + \frac{\gamma \alpha}{3}}\pa{ 1- \frac{5 \lambda_2 \sqrt{ \alpha C_1 \log n}}{2 \Delta_i\pa{ 1 + \gamma \alpha }} }
        \\
        &\geq
        \frac{\Delta_2}{\sqrt{C_1}}\pa{1 + \frac{\gamma \alpha}{3}}\pa{ 1- \frac{5 \lambda_2 \sqrt{ \alpha C_1 \log n}}{2 \lambda \sqrt{\frac{C_1 \log n}{\alpha}}\pa{ 1 + \gamma \alpha }} }
        &\pa{\text{By \Cref{eq:condition_on_bias_function_of_alpha}}}
        \\
        &=
        \frac{\Delta_2}{\sqrt{C_1}}\pa{1 + \frac{\gamma \alpha}{3}}\pa{\frac{2 \lambda \sqrt{\frac{C_1 \log n}{\alpha}} +  \pa{ 2\lambda \gamma - 5 \lambda_2 }\sqrt{ \alpha C_1 \log n}}{2 \lambda \sqrt{\frac{C_1 \log n}{\alpha}}\pa{ 1 + \gamma \alpha }} }
        \\
        &\geq
        \frac{\Delta_2}{\sqrt{C_1}}\pa{1 + \frac{\gamma \alpha}{3}}\pa{\frac{1 + \frac{3}{4} \gamma \alpha }{ 1 + \gamma \alpha } }
        &\pa{\text{by taking $\lambda$ large enough}}
        \\
        &\geq
        \frac{\Delta_2}{\sqrt{C_1}}
        \\
        &\geq
        \lambda \sqrt{\frac{\log n}{\alpha}}.
    \end{align*}
    This proves the claim whenever \(\|C\|_2^2 \leq n C_1 \leq 24 \|C\|_2^2\).
    Now consider the remaining case, $n C_1 \geq 24\|C\|_2^2$. By \Cref{lemma:growth_expectation} (\Cref{eq:bias_condition_preserved_unbalanced}), \Cref{eq:first_opinion_divides_bias_taylor_expansion} implies
    \begin{align*}
        \frac{\Delta_i'}{\sqrt{C_1'}} 
        &\geq 
        \frac{\Delta_2}{\sqrt{C_1}} \sqrt{\frac{\bbE[C'_1\mid C]}{C_1}}\pa{ 1- \frac{5 \lambda_1 \sqrt{\frac{n \alpha C_1^2 \log n}{\|C\|_2^2}} }{\Delta_i\pa{ 1 + \gamma \alpha \frac{nC_1}{\|C\|_2^2} }} }
        \\
        &\geq 
        \frac{\Delta_2}{\sqrt{C_1}} \sqrt{\frac{\bbE[C'_1\mid C]}{C_1}}\pa{ 1- \frac{5 \lambda_1 \sqrt{\frac{n \alpha C_1^2 \log n}{\|C\|_2^2}} }{2 \lambda \sqrt{\frac{C_1 \log n}{\alpha}}\pa{ 1 + \gamma \alpha \frac{nC_1}{\|C\|_2^2} }} }
        &\pa{\text{By \Cref{eq:condition_on_bias_function_of_alpha}}}
        \\
        &=
        \frac{\Delta_2}{\sqrt{C_1}}\sqrt{\frac{\bbE[C'_1\mid C]}{C_1}}\pa{ \frac{2 \lambda \sqrt{\frac{C_1 \log n}{\alpha}} + \pa{ 2 \lambda \gamma \sqrt{ \frac{nC_1}{\|C\|^2_2} }- 5 \lambda_1 }\sqrt{\frac{n \alpha C_1^2 \log n}{\|C\|_2^2}} }{2 \lambda \sqrt{\frac{C_1 \log n}{\alpha}}\pa{ 1 + \gamma \alpha \frac{nC_1}{\|C\|_2^2} }} }
        \\
        &\geq
        \frac{\Delta_2}{\sqrt{C_1}}\sqrt{\frac{\bbE[C'_1\mid C]}{C_1}}\pa{ \frac{ 1+ \frac{3}{4} \gamma \alpha \frac{nC_1}{\|C\|_2^2} }{ 1 + \gamma \alpha \frac{nC_1}{\|C\|_2^2} } }
        &\pa{\text{for large $\lambda$}}
        \\
        &\geq
        \frac{\Delta_2}{\sqrt{C_1}}\sqrt{1+ \gamma \alpha \frac{7}{24} \frac{nC_1}{\|C\|^2_2}}\pa{ \frac{ 1+ \frac{3}{4} \gamma \alpha \frac{nC_1}{\|C\|_2^2} }{ 1 + \gamma \alpha \frac{nC_1}{\|C\|_2^2} } }
        &\pa{\text{by \Cref{claim:expectation_growth_C_1_unbalanced}} }
        \\
        &\geq
        \frac{\Delta_2}{\sqrt{C_1}},        
        \\
        &\geq
        \lambda \sqrt{\frac{\log n}{\alpha}}.
    \end{align*}
    where in the second-last inequality we used the fact that, if \(x\geq 24\), then \(\frac{\sqrt{1+\frac{7}{24} x} \pa{1+\frac{3}{4}x}}{1+x} \ge 1\).
    Moreover, by \Cref{lemma:majority_concentration}, on the event \(\mathcal{E}\) we have \(C_1'\ge C_1\). Therefore,
    \[
        \frac{1}{\alpha}
        =
        \max\cpa{\frac{n}{h^2 C_1},1}
        \ge
        \max\cpa{\frac{n}{h^2 C_1'},1},
    \]
    and hence
    \[
        \lambda \sqrt{\frac{\log n}{\alpha}}
        \ge
        \lambda \sqrt{\max\cpa{\frac{n}{h^2 C_1'},1}\log n}.
    \]
    Multiplying both sides by \(\sqrt{C_1'}\), we conclude that
    \[
        \Delta_i' \ge \lambda \sqrt{\max\cpa{\frac{n}{h^2},C_1'}\log n},
    \]
    as claimed.
\end{proof}

\subsection{Convergence time of \dejavu}

To prove \cref{thm:convergence_time_dejavu}, we need one last lemma to handle the case in which $C_1 \ge \frac{3n}{4}$. In this regime, we reduce the analysis to the better-understood binary case by merging all remaining opinions.
\begin{lemma}[Binary merging domination]
    \label{lemma:binary_merging_domination}
    Let \(C^{(0)}=(C_1^{(0)},\ldots,C_k^{(0)})\) be an initial configuration, and let \(\bar C^{(0)}=(C_1^{(0)},n-C_1^{(0)})\) be the binary configuration obtained by merging all opinions \(i\neq 1\) into a single competing opinion.
    Let \(\pa{C^{(t)}}_{t\geq 0}\) be the \dejavu process started from \(C^{(0)}\), and let \(\pa{\bar C^{(t)}}_{t\geq 0}\) be the binary \dejavu process started from \(\bar C^{(0)}\).
    Then there exists a coupling such that
    \[
        \bar C_1^{(t)} \leq C_1^{(t)}
    \]
    for all \(t\geq 0\). In particular, the consensus time on opinion \(1\) in the original process is stochastically dominated by the consensus time in the merged binary process.
\end{lemma}
\begin{proof}
    We realize the two processes on the same set of \(n\) agents. For each round \(t\), let
    \[
        A_t\coloneq \cpa{u\in[n]:\text{agent \(u\) has opinion \(1\) in the original process at time \(t\)}}
    \]
    and
    \[
        B_t\coloneq \cpa{u\in[n]:\text{agent \(u\) has opinion \(1\) in the binary process at time \(t\)}}.
    \]
    We construct the coupling so that
    \[
        B_t\subseteq A_t
    \]
    for every \(t\geq 0\). At time \(t=0\), choose \(B_0=A_0\), so that \(|B_0|=\bar C_1^{(0)}=C_1^{(0)}\).

    Assume inductively that \(B_t\subseteq A_t\). For every agent \(u\), reveal the same sampled agents \(v_1,\ldots,v_h\) in both processes during round \(t\). In the original process, the sampled opinions are the actual opinions of \(v_1,\ldots,v_h\). In the binary process, we declare that a sampled agent has opinion \(1\) if it belongs to \(B_t\), and opinion \(0\) otherwise. Since \(|B_t|=\bar C_1^{(t)}\), this gives exactly the correct law for one round of the binary process started from \(\bar C^{(t)}\).

    Fix an agent \(u\). We claim that, under this coupling,
    \[
        \mathbf{1}\cpa{u\in B_{t+1}}
        \leq
        \mathbf{1}\cpa{u\in A_{t+1}}.
    \]
    If \(u\in B_{t+1}\), then either:
    \begin{enumerate}
        \item the binary process observes a repeated \(1\) before any repeated \(0\), or
        \item no opinion is repeated within the first \(h\) samples and \(u\in B_t\), so \(u\) keeps opinion \(1\).
    \end{enumerate}
    In the first case, before the second sampled agent from \(B_t\) appears, there is at most one sampled agent outside \(B_t\). Since \(B_t\subseteq A_t\), this implies that before that same time there is at most one sampled agent outside \(A_t\), hence in the original process no opinion different from \(1\) can appear twice. On the other hand, the two samples from \(B_t\) are also two samples of opinion \(1\) in the original process. Therefore the first repeated opinion in the original process is also \(1\), and \(u\in A_{t+1}\).
    In the second case, \(u\in B_t\subseteq A_t\), so \(u\) starts the round with opinion \(1\) also in the original process. Moreover, if no opinion is repeated in the binary process within the first \(h\) samples, then necessarily \(h=2\), with one sampled agent in \(B_t\) and one sampled agent outside \(B_t\). Hence in the original process either the two samples have different opinions, in which case \(u\) keeps opinion \(1\), or the sampled agent outside \(B_t\) belongs to \(A_t\setminus B_t\), in which case the two original samples are both equal to \(1\) and \(u\) adopts opinion \(1\). Thus again \(u\in A_{t+1}\).

    Since this holds for every agent \(u\), we obtain \(B_{t+1}\subseteq A_{t+1}\). By induction, \(B_t\subseteq A_t\) for all \(t\geq 0\), and therefore
    \[
        \bar C_1^{(t)}=|B_t|\le |A_t|=C_1^{(t)}
    \]
    for all \(t\geq 0\).

    Finally, if the merged binary process reaches consensus on opinion \(1\) at time \(t\), then \(\bar C_1^{(t)}=n\). Since \(\bar C_1^{(t)}\le C_1^{(t)}\le n\), it follows that \(C_1^{(t)}=n\) as well.
\end{proof}

We are now ready to prove \cref{thm:convergence_time_dejavu}, which we now restate for convenience.

\begin{theorem*}
    Let \(h\geq 2\) and \(C=(C_1,\ldots,C_k)\) be an initial system configuration where each agent supports an opinion in \(\{1,\ldots,k\}\).
    Without loss of generality, let \(C_1\geq C_2 \geq \cdots \geq C_k\).
    Assume that \(C_1 = \omega\pa{\log n}\) and that, for a large enough constant \(\lambda > 0\),
    \[
        C_1 - C_2 \ge \lambda \sqrt{ \max\cpa{\frac{n}{h^2}, C_1} \log n }.
    \]
    \dejavu converges to consensus on the first opinion w.h.p.\ in \(O\pa{ (\frac{n}{h^2 C_1}+1) \log n}\) rounds.
\end{theorem*}
\begin{proof}[Proof of \Cref{thm:convergence_time_dejavu}]
    If \(C_1^{(0)}> 3n/4\), let \(\pa{\bar C^{(t)}}_{t\geq 0}\) be the binary process obtained by merging all opinions \(i\neq 1\) into a single competing opinion, as in \Cref{lemma:binary_merging_domination}.
    In the binary setting, \dejavu coincides with \twochoices if \(h=2\), and with \threemaj if \(h>2\).
    Since \(\bar C_1^{(0)}>3n/4\), the known high-probability bounds for these two binary dynamics imply that \(\pa{\bar C^{(t)}}_{t\geq 0}\) reaches consensus on opinion \(1\) within \(O(\log n)\) rounds w.h.p.~\cite{Shimizu2025}.
    By \Cref{lemma:binary_merging_domination}, the original process reaches consensus on opinion \(1\) no later than the merged binary process.
    Therefore the theorem follows in this case.

    We now consider the remaining case \(C_1^{(0)}\le 3n/4\).
    For each round \(t\ge 0\), define
    \[
        \alpha_t\coloneq \min\cpa{\frac{C_1^{(t)}}{n}h^2,1}.
    \]
    Also define the event
    \[
        \mathcal{E}^{(t)}
        =
        \bigcap_{j\in[k]}
        \cpa{
            \qty|C_j^{(t+1)}-\bbE[C_j^{(t+1)}\mid C^{(t)}]|
            \le
            \lambda_1\sqrt{\frac{\alpha_t\, n (C_1^{(t)})^2\log n}{\|C^{(t)}\|_2^2}}
        }.
    \]
    For every round \(t\) such that \(C_1^{(t)}\le 3n/4\), let
    \begin{align*}
        \mathcal{F}^{(t)}
        &=
        \bigcap_{i=2}^k
        \cpa{
            \Delta_i^{(t+1)}
            \ge
            \Delta_i^{(t)}(1+c_B\alpha_t)
        },
        \\
        \mathcal{G}^{(t)}
        &=
        \bigcap_{i=2}^k
        \cpa{
            \Delta_i^{(t+1)}
            \ge
            \lambda\sqrt{\max\cpa{\frac{n}{h^2},C_1^{(t+1)}}\log n}
        },
        \\
        \mathcal{H}^{(t)}
        &=
        \cpa{C_1^{(t+1)}\ge C_1^{(t)}}.
    \end{align*}
    Assume that, at some round \(t\), the current configuration satisfies
    \[
        \bigcap_{i=2}^k
        \cpa{
            \Delta_i^{(t)}
            \ge
            \lambda\sqrt{\max\cpa{\frac{n}{h^2},C_1^{(t)}}\log n}
        }
        \qquad\text{and}\qquad
        C_1^{(t)}\le \frac34 n.
    \]
    Then, by \Cref{lemma:concentration_all_opinions},
    \[
        \Pr(\mathcal{E}^{(t)}\mid C^{(t)})\ge 1-n^{-100}.
    \]
    Moreover, by \Cref{lemma:bias_concentration_bernstein,lemma:condition_delta_preserved,lemma:majority_concentration},
    \[
        \Pr(\mathcal{F}^{(t)}\cap \mathcal{G}^{(t)}\cap \mathcal{H}^{(t)} \mid C^{(t)},\mathcal{E}^{(t)})=1.
    \]
    Therefore,
    \[
        \Pr(\mathcal{F}^{(t)}\cap \mathcal{G}^{(t)}\cap \mathcal{H}^{(t)} \mid C^{(t)})
        \ge
        1-n^{-100}.
    \]

    Let
    \[
        \tau\coloneq \min\cpa{t\ge 0: C_1^{(t)}>3n/4}.
    \]
    As long as \(t<\tau\), event \(\mathcal{G}^{(t)}\) implies that the bias condition needed to reapply the previous argument holds at round \(t+1\), while \(\mathcal{H}^{(t)}\) implies that \(C_1^{(t)}\) is nondecreasing.

    Set
    \[
        \beta\coloneq c_B\min\cpa{\frac{C_1^{(0)}}{n}h^2,1}.
    \]
    Since \(C_1^{(t)}\ge C_1^{(0)}\) for all \(t<\tau\), we have \(\alpha_t\ge \min\cpa{\frac{C_1^{(0)}}{n}h^2,1}\), and hence on every event \(\mathcal{F}^{(t)}\) with \(t<\tau\),
    \[
        \Delta_2^{(t+1)}
        \ge
        \Delta_2^{(t)}(1+\beta).
    \]
    Iterating, we obtain
    \[
        \Delta_2^{(t)}
        \ge
        \Delta_2^{(0)}(1+\beta)^t
        \qquad\text{for every }t<\tau.
    \]

    Let
    \[
        \tau_2\coloneq \min\cpa{t\ge 0:\Delta_2^{(t)}>3n/4}.
    \]
    Since \(\Delta_2^{(t)}\le C_1^{(t)}\) for every \(t\), the event \(\cpa{\Delta_2^{(t)}>3n/4}\) implies \(\cpa{C_1^{(t)}>3n/4}\), and hence \(\tau\le \tau_2\).
    Also define
    \[
        T\coloneq \min\cpa{t\in\bbN: \Delta_2^{(0)}(1+\beta)^t>3n/4}.
    \]
    By the previous growth estimate, on the event \(\bigcap_{s=0}^{T-1}(\mathcal{F}^{(s)}\cap\mathcal{G}^{(s)}\cap\mathcal{H}^{(s)})\) we have \(\tau_2\le T\), and therefore \(\tau\le T\).
    Moreover, by a union bound,
    \[
        \Pr\!\left(
            \bigcap_{s=0}^{T-1}(\mathcal{F}^{(s)}\cap\mathcal{G}^{(s)}\cap\mathcal{H}^{(s)})
        \right)
        \ge
        1-Tn^{-100}
        =
        1-n^{-\Theta(1)}.
    \]
    Since \(\Delta_2^{(0)}\ge \lambda\sqrt{\max\cpa{n/h^2,C_1^{(0)}}\log n}\ge 1\), we obtain
    \[
        T
        =
        O\pa{\frac{\log n}{\log(1+\beta)}}
        =
        O\pa{\qty(\frac{n}{h^2C_1^{(0)}}+1)\log n}.
    \]
    Therefore, with high probability, the process reaches a configuration with \(C_1^{(t)}>3n/4\) within
    \[
        O\pa{\qty(\frac{n}{h^2C_1^{(0)}}+1)\log n}
    \]
    rounds.
    We are now back to the case where the majority opinion has size larger than \(3n/4\), and therefore we can apply the previous argument to conclude that consensus on opinion \(1\) is reached within an additional \(O(\log n)\) rounds w.h.p.

    Hence the overall convergence time is
    \[
        O\pa{\qty(\frac{n}{h^2C_1^{(0)}}+1)\log n},
    \]
    concluding the proof of \Cref{thm:convergence_time_dejavu}.
\end{proof}

\section{\texorpdfstring{Lower bound for \(h\)-majority}{Lower bound for h-majority}}
\label{sec:lower_bound_h_majority}

In this section, 
we prove \Cref{thm:app:lb:h-majority}, which we restate here for convenience.

\begin{theorem*}
    Let \(\varepsilon > 0\) be any arbitrarily small constant.
    Let \(C = (C_1, \ldots, C_k)\) be the starting configuration of the system, with \(C_1 \ge \ldots \ge C_k\), and \(C_1 \le n/100\). For any \(h \ge n^{0.75 + \varepsilon} / C_1 \), 
    the \(h\)-majority process converges in time \(\Omega(n / (h^2 C_1) + 1)\) with high probability. 
\end{theorem*}

To prove \cref{thm:app:lb:h-majority}, we revisit the proof given by \cite{becchettiSimpleDynamicsPlurality2017}. 
We will generally refer to the number of nodes supporting an opinion \(j\) at any given time \(t\) by \(C_j\), and we denote by \(C_j'\) the number of nodes supporting opinion \(j\) at time \(t+1\) conditional on the system configuration at time \(t\).
We rely on the following lemma, adapted from \cite[Lemma 9]{becchettiSimpleDynamicsPlurality2017}.

\begin{lemma}
\label{lem:app:max-growth:h-maj}
    Let \(C = (C_1, \ldots, C_k)\) be the starting configuration of the system, with \(C_1 \ge \ldots \ge C_k\).
    If \(h \ge n^{0.75 + \varepsilon} / C_1\), then after one round of the \(h\)-majority protocol, 
    \[
        \Pr\left(C_j' \ge \left(1 + \frac{h^2 C_1}{n}\right)C_1 \mid C\right) \le \exp[-n^{\Theta(\varepsilon)}].
    \]
\end{lemma}
\begin{proof}
    Let \(u \in [n]\) be any specific node, and let \(N_j\) be the number of nodes with opinion \(j\) picked by \(u\) during the sampling stage of \(h\)-majority protocol.
    Let \(Y_u\) be the indicator random variable of the event that node \(u\) adopts opinion \(j\). 
    We give an upper bound on the probability of the event \(Y_u = 1\) by conditioning it on \(N_j = 1\) and \(N_j > 2\) (observe that if \(N_j = 0\) node \(u\) cannot choose \(j\) as its opinion).
    \begin{align}
        \Pr(Y_u = 1) \le \Pr(Y_u \mid N_j = 1)\Pr(N_j = 1) + \Pr(N_j \ge 2). \label{eq:app:est-1-growth}
    \end{align}
    The term \(\Pr(Y_u = 1  \mid N_j = 1)\) is at most \(1/h\) since the only event in which \(u\) picks opinion \(j\) is that all the \(h\) sampled opinions are different.
    The term \(\Pr(N_j = 1)\) can be upper bounded by the probability that at least \(1\) opinion among the sampled ones is \(1\), which is, by the union bound, at most \(h C_j / n\).
    Finally, the term \(\Pr(N_j \ge 2)\) is at most \(\binom{h}{2} C_j^2 / n^2\).
    Hence, \cref{eq:app:est-1-growth} can be rewritten as 
    \begin{align*}
        \Pr(Y_u = 1) & \le \frac{C_j}{n} + \binom{h}{2}\frac{C_j^2}{n^2} \\
        & \le \frac{C_j}{n} + \frac{h^2 C_j^2}{2 n^2}.
    \end{align*}
    Hence, it holds that 
    \begin{align*}
        \bbE[C_j' \mid C] & \le C_j + \frac{h^2 C_j^2}{2n} \\
        & = C_j \left(1 + \frac{h^2 C_j}{2n}\right) \\
        & \le C_j \left(1 + \frac{h^2 C_1}{2n}\right),
    \end{align*}
    where the latter inequality holds by the hypothesis on \(C_j\).
    
    We now consider two cases.
    First, assume \(C_j \ge C_1 / 2\).
    By the Hoeffding bound (\cref{lemma:app:hoeffding}), we obtain that
    \begin{align*}
        \Pr\left(C_j ' \ge C_j \left(1 + \frac{h^2 C_1}{n}\right) \mid C\right) & \le \exp\left[ - \frac{2 \left(C_j \cdot \frac{h^2 C_1}{2n}\right)^2}{n}\right] \\
        & \le \exp\left[ - \frac{h^4 C_1^2 C_j^2}{2n^3}\right] \\
        & \le \exp\left[ - \frac{h^4 C_1^4}{8 n^3}\right],
    \end{align*}
    where the latter inequality holds by the hypothesis on \(C_j\).
    To conclude, observe that the hypothesis on \(h\) implies that \(h C_1 \ge n^{0.75 + \varepsilon}\) and, hence, 
    \begin{align*}
        \exp\left[ - \frac{h^4 C_1^4}{8 n^3}\right] \le \exp\left[- \frac{n^{4\varepsilon}}{8}\right].
    \end{align*}

    In the second case, we assume \(C_j < C_1 / 2\).
    Then, 
    \begin{align*}
        & \Pr\left( C_j' \ge C_1 \left(1 + \frac{h^2C_1}{2n}\right) \mid C \right)\\
        = \ & \Pr\left( C_j' \ge C_j \left(1 + \frac{h^2C_1}{2n}\right) + C_1 \left(1 + \frac{h^2C_1}{2n}\right) - C_j \left(1 + \frac{h^2C_1}{2n}\right) \mid C \right) \\
        = \ &  \Pr\left( C_j' \ge C_j \left(1 + \frac{h^2C_1}{2n}\right) + (C_1 - C_j) \left(1 + \frac{h^2C_1}{2n}\right) \mid C \right) \\
        \le \ & \Pr\left( C_j' \ge C_j \left(1 + \frac{h^2C_1}{2n}\right) + \frac{C_1}{2} \left(1 + \frac{h^2C_1}{2n}\right) \mid C \right),
    \end{align*}
    where, in the latter inequality, we used that \(C_j < C_1 / 2\).
    We can continue by observing that 
    \begin{align*}
        \Pr\left( C_j' \ge C_1 \left(1 + \frac{h^2C_1}{2n}\right) \mid C \right) & \le \Pr\left( C_j' \ge C_j \left(1 + \frac{h^2C_1}{2n}\right) + \frac{C_1}{2} \left(1 + \frac{h^2C_1}{2n}\right) \mid C \right) \\
        & \le \Pr\left( C_j' \ge C_j \left(1 + \frac{h^2C_1}{2n}\right) +  \frac{h^2C_1^2 }{4n} \mid C \right).
    \end{align*}
    By the Hoeffding bound (\cref{lemma:app:hoeffding}), we obtain that
    \begin{align*}
        \Pr\left( C_j' \ge C_1 \left(1 + \frac{h^2C_1}{2n}\right) \mid C\right) & \le \exp\left[ - \frac{2 \left(\frac{h^2 C_1^2}{4n}\right)^2}{n}\right] \\
        & = \exp\left[ - \frac{ h^4 C_1^4}{8n^3}\right] \\
        & \le \exp\left[ - \frac{ n^{4\varepsilon}}{8}\right],
    \end{align*}
    where the latter inequality comes from the hypothesis on \(h\).
    
    Hence, for any \(j \in [n]\), it holds that 
    \begin{align*}
        \Pr\left(C_j' \ge C_1 \left(1 + \frac{h^2 C_1}{n}\right)\mid C \right) \le \exp[- n^{\Theta(\varepsilon)}].
    \end{align*}    
\end{proof}

We are now ready to prove \cref{thm:app:lb:h-majority}.

\begin{proof}[Proof of \cref{thm:app:lb:h-majority}]
    It is sufficient to iteratively apply \cref{lem:app:max-growth:h-maj} and the union bound.
    More formally,
    for any \(i \in [k]\), let \(c_i^{(t)}\) be the number of nodes supporting opinion \(i\) at the end of round \(t\).
    Note that \(c_i^{(0)} = C_i\) for all \(i \in [k]\), and that \(C_1 \le n / 100\).
    Let \(X^{(t)} = \max_{j \in [k]} \{c_j^{(t)}\}\).
    We lower bound the convergence time of the \(h\)-majority process by the time \(\tau\) any opinion reaches at least \(100C_1\) 
    supporting nodes.
    Note that, whenever \(X^{(t)}<100C_1\), by the union bound and \cref{lem:app:max-growth:h-maj}, it holds that
    \[
        X^{(t+1)} \le X^{(t)} \left(1 + \frac{h^2 (100C_1)}{n}\right)
    \]
    with probability \(1 - \exp[-n^{\Theta(\varepsilon)}]\).
    Let \(T = \lfloor n \ln 99 / (100 h^2 C_1) \rfloor\), and denote by \(E_t\) the event \(X^{(t)} \le C_1(1 + h^2 (100C_1)/n)^t\).
    Note that \(E_T\) implies that \(X^{(T)} < 100 C_1\), and, hence, \(\cap_{t = 0}^{T} E_t\) implies that \(\tau \ge T\).
    It holds that 
    \begin{align*}
        \Pr\left(\tau \ge T \right) & \ge \Pr(\cap_{t = 0}^T E_t) \\
        & = \prod_{t = 0}^{T-1} \Pr\left(E_{t+1} \mid \cap_{j=0}^t E_j\right) \\
        & \ge \prod_{t = 0}^{T-1}\Pr\left(X^{(t+1)} \le X^{(t)} \left(1 + \frac{h^2 (100C_1)}{n}\right) \mid \cap_{j=0}^{t} E_j \right) \\
        & \ge \prod_{t = 0}^{T-1}\left[1 - \exp[-n^{\Theta(\varepsilon)}]\right] \\
        & \ge 1 - T \exp[-n^{\Theta(\varepsilon)}] \\
        & \ge 1 - \exp[-n^{\Theta(\varepsilon)}],
    \end{align*}
    where the third inequality follows because on the event \(\cap_{j=0}^{t} E_j\) we have
    \[
        X^{(t)} \le C_1\left(1 + \frac{h^2 (100C_1)}{n}\right)^t < 100 C_1,
    \]
    so \cref{lem:app:max-growth:h-maj} applies at round \(t\), while the last two inequalities follow from Bernoulli's inequality, that is, \((1 - x)^m \ge 1 - mx\) for any \(x \in [0,1]\) and \(m \ge 1\), together with the definition of \(T\) and the hypothesis on \(h\) and \(C_1\).
    This concludes the proof.
\end{proof}

\section{\texorpdfstring{Number of samples required by \dejavu}{Number of samples of DéjàVu}}
\label{sec:samplesdejavu}

The goal of this section is to prove \Cref{thm:sample-efficiency}. To this end, we derive a lower bound on $S_m$, the number of samples required for convergence of \hmaj, and an upper bound on $S_d$, the number of samples required for convergence of \dejavu.

While $S_m$ is simply $h$ times the convergence time of \hmaj, the random variable $S_d$ is considerably more intricate, as it depends on the configuration at every round prior to consensus. Consequently, the main technical effort of this section is devoted to proving the following theorem.

\begin{theorem}
    \label{thm:dejavu_samples}
		Let \(h\geq 2\) and \(C=(C_1,\ldots,C_k)\) be an initial system configuration where each agent supports an opinion in \(\{1,\ldots,k\}\).
    Without loss of generality, let \(C_1\geq C_2 \geq \cdots \geq C_k\).
    Assume that \(C_1 = \omega\pa{\log^2 n}\) and that, for a large enough constant \(\lambda > 0\),
    \[
        C_1 - C_2 \ge \lambda \sqrt{ \max\cpa{\frac{n}{h^2}, C_1} \log n } .
    \]
    Let $S_d$ be the number of samples of \dejavu in \pullh. Then, w.h.p.,
    \[
        S_d = O\pa{ \min\cpa{h, \frac{n}{\norm{C}_2}} \pa{\frac{n}{h^2 C_1}+1 } \log n }
    \]
\end{theorem}

The key ingredient in proving this theorem is to show that the $\ell_2$-norm of the configuration does not decrease by more than a constant multiplicative factor. More precisely, we prove that
\(
\|C^{(t)}\| \geq c^\downarrow \|C^{(0)}\|
\)
for some absolute constant $c^\downarrow > 0$ and for all $t \leq T$, where $T$ denotes the convergence time of \dejavu. 

Indeed, by \Cref{lemma:tight_bound_probability_collision}, once \(\|p^{(t)}\|\) is bounded from below, an agent sees a repeated opinion within \(O(\|p^{(t)}\|^{-1})\) samples with constant probability. Therefore, once we establish a uniform lower bound on \(\|p^{(t)}\|\), we can conclude \Cref{thm:dejavu_samples} by combining this per-round bound with the upper bound on the convergence time.

Let us first define the stopping time we will bound.
\begin{definition}[Stopping times for basic quantities] \label{def:stopping times}

  For constants $c^\uparrow,c^\downarrow>0$ define
    \begin{align*}
        &\tau^\uparrow = \inf\{t\geq 0: \|p^{(t)}\|\geq (1+c^\uparrow)\|p^{(0)}\|\}, \\
        &\tau^\downarrow = \inf\cpa{t\geq 0: \|p^{(t)}\|\leq (1-c^\downarrow)\|p^{(0)}\|}.
    \end{align*}
\end{definition}
Here is the aforementioned key lemma to prove \Cref{thm:dejavu_samples}.
\begin{lemma}[Bounded decrease of $\|C^{(t)}\|$]
\label{lem:taunormdown is large}
  Consider the stopping times $\tau^\downarrow$ defined in \cref{def:stopping times}.
  Then, for any $T>0$, we have
    \begin{align*}
        \Pr\qty[\tau^\downarrow\leq T]
        \le \begin{cases}
            T\exp\qty(-\Omega\qty(\frac{n}{h^2 T + n \|C\|^{-1}}))
            \quad & 
            \text{if } h \|C^{(0)}\| \leq n
            \\
            T\exp\qty(-\Omega\qty(\frac{\|C^{(0)}\|^2}{nT}))
            \quad & 
            \text{if } h \|C^{(0)}\| > n
        \end{cases}
    \end{align*}
\end{lemma}

We follow the same drift analysis approach presented in \cite{Shimizu2025} used to analyze \textsc{3-majority} and \textsc{2-choices}.
In that paper, the authors use a drift analysis based on the Bernstein condition.
First, they prove that the euclidean norm of the configuration at round $t$, $\|C_t\|$, is a sub-martingale.
Then they prove that the difference $\|C_t\| - \|C_{t-1}\|$ conditioned on the $(t-1)$-th configuration satisfies the Bernstein condition. Under these conditions, they show that, it holds that $\|C_t\| = \Omega \pa{\|C_0\|} $ w.h.p. for all $0 \leq t\leq n$.

Our main contribution to adapt their analysis to \dejavu is to show that the sub-martingale condition for $\|C_t\|$ still holds for \dejavu. More generally, the same argument applies to any \emph{Majority Boosting Protocol}, i.e. any protocol such that for all opinions \(i<j\),
\[
    \frac{\bbE(C_i')}{\bbE(C_j')} \ge \frac{C_i}{C_j}.
\]
\begin{lemma}
    \label{lemma:expected_value_norm2_inequality}
    Suppose the protocol is Majority Boosting, namely that for all \(i<j\),
    \[
        \frac{\bbE\pa{C'_i}}{\bbE\pa{C'_j}} \geq \frac{C_i}{C_j},
    \]
    Then
    \[
        \bbE\pa{\| C' \|^2 \mid C} \geq \norm{C}^2.
    \]
\end{lemma}

\subsection{The configuration norm is a submartingale for Majority Boosting Protocols}\label{sec:norm-submartingale}

In this section we prove \Cref{lemma:expected_value_norm2_inequality}.
To do so we first prove in \cref{lemma:norm_2_expected_value_inequality} that $(C_t,\bbE(C_{t+1}\mid C_t) ) \geq (C_t,C_t)$, where $(\cdot,\cdot)$ is the scalar product. Using the Cauchy-Schwarz inequality, we obtain that $\|\bbE\pa{ C_{t+1}\mid C_t }\|^2 \geq \|C_t\|^2$. By the definition of variance, we finally obtain the sub-martingale condition in \cref{lemma:expected_value_norm2_inequality}, i.e. $\bbE\pa{ \|C_{t+1}\|^2 } \geq \|C_t\|^2$.

We need some technical lemmas.

\begin{lemma}
    \label{lemma:norm_2_expected_value_inequality}
    It holds
    \[
        \sum_{j\in[k]} C_j^2 \leq \sum_{j\in[k]} C_j \cdot \bbE\pa{C'_j}.
    \]
\end{lemma}

Before proving \cref{lemma:norm_2_expected_value_inequality}, we need the following technical claims.

\begin{claim}
    \label{claim:auxiliary_monotonicity_expectations_cumulative_difference}
    For all opinions $1\leq i < j \leq k$, we have that $\bbE\pa{C'_i}-C_i \leq 0$ implies that $\bbE\pa{C'_j}-C_j \leq 0$.
\end{claim}
\begin{proof}
    W.l.o.g. we can assume $C_i,C_j>0$.
    By the Majority Boosting assumption, we have that
    \[
        \frac{\bbE\pa{C'_i}}{\bbE\pa{C'_j}} \geq \frac{C_i}{C_j}.
    \]
    This implies that
    \[
        \frac{\bbE\pa{C'_j}}{C_j} \leq \frac{\bbE\pa{C'_i}}{C_i}  \leq 1,
    \]
    where the last inequality follows if we assume $\bbE\pa{C'_i}-C_i \leq 0$. This concludes the proof of \cref{claim:auxiliary_monotonicity_expectations_cumulative_difference}.
\end{proof}

\begin{claim}
    \label{claim:monotonicity_expectations_cumulative_difference}
    For all opinions $1\leq i \leq k$, it holds
    \[
        \sum_{j=1}^i \left(\bbE\pa{C'_j}-C_j\right) \geq 0.
    \]
\end{claim}
\begin{proof}
    Let $i^*=\min\cpa{i\in[k]: \bbE\pa{C'_i}-C_i \leq 0}$. The claim is trivially true for all $i<i^*$ or if the set is empty. By \cref{claim:auxiliary_monotonicity_expectations_cumulative_difference}, for all $i^*\leq i \leq k$, $\bbE\pa{C'_i}-C_i \leq 0$. This implies that for $i^*\leq i \leq k$, $\sum_{j=1}^i \left( \bbE\pa{C'_j}-C_j \right)$ is a non-increasing function in $i$. Since for $i=k$ we have $\sum_{j=1}^i \left( \bbE\pa{C'_j}-C_j \right)
    = n - n = 0$, we conclude the proof of \cref{claim:monotonicity_expectations_cumulative_difference}.
\end{proof}

\begin{proof}[Proof of \cref{lemma:norm_2_expected_value_inequality}]
    Let
    \[
        d_j\coloneq \bbE\pa{C'_j}-C_j,
        \qquad
        S_i\coloneq \sum_{j=1}^i d_j
        \qquad (i\in[k]).
    \]
    By \Cref{claim:monotonicity_expectations_cumulative_difference}, we have \(S_i\ge 0\) for every \(i\in[k]\), and since \(\sum_{j=1}^k C'_j=n=\sum_{j=1}^k C_j\), we also have \(S_k=0\).

    Summation by parts yields
    \begin{align*}
        \sum_{j=1}^k C_j d_j
        &= \sum_{i=1}^{k-1} (C_i-C_{i+1}) S_i + C_k S_k
        \\
        &= \sum_{i=1}^{k-1} (C_i-C_{i+1}) S_i
        \\
        &\ge 0,
    \end{align*}
    because \(C_1\ge \cdots \ge C_k\) and \(S_i\ge 0\) for all \(i\le k-1\).
    This concludes the proof of \cref{lemma:norm_2_expected_value_inequality}.
\end{proof}

The Cauchy-Schwarz inequality and \cref{lemma:norm_2_expected_value_inequality} imply this corollary. Note that this is not yet the sub-martingale condition, since the expectation is inside the norm operator.
\begin{corollary}
    \label{corollary:norm_2_expected_value_inequality}
    It holds
    \(\| \bbE\pa{C'} \| \geq \| C \|\).
\end{corollary}
\begin{proof}
    By \cref{lemma:norm_2_expected_value_inequality} and the Cauchy-Schwarz inequality, we have
    \[
        \| C \|^2
        = (C , C)
        \leq ( C, \bbE\pa{C'} )
        \leq \| C \| \| \bbE\pa{C'} \|.
    \]
    By dividing both terms by $\| C \| $, we conclude the proof of \cref{corollary:norm_2_expected_value_inequality}.
\end{proof}

Now we have all the ingredients to show that $\|C_t\|$ is a sub-martingale.

\begin{proof}[Proof of \Cref{lemma:expected_value_norm2_inequality}]
    We have that
    \begin{align*}
        \bbE\pa{\| C' \|^2 \mid C} - \| C \|^2
        &= \sum_{j\in [k]} \bbE\pa{{C_j'}^2 \mid C} - \| C \|^2
        \\
        &= \sum_{j\in [k]} \pa{\bbE\pa{C_j' \mid C}^2 + \Var\pa{C_j' \mid C}} - \| C \|^2
        \\
        &\geq \sum_{j\in [k]} \pa{\bbE\pa{C_j' \mid C}^2 } - \| C \|^2
        \\
        &= \|\bbE\pa{C' \mid C} \|^2 - \| C \|^2
        \\
        &\geq  0.
        &\pa{\text{by \cref{corollary:norm_2_expected_value_inequality}}}
    \end{align*}
    By rearranging the last inequality, we conclude the proof of \cref{lemma:expected_value_norm2_inequality}.
\end{proof}

\subsection{\dejavu does not decrease the configuration norm}
\label{sec:structure_from_shimizu}

In this section we follow the steps in \cite{Shimizu2025} to show that, for the \dejavu protocol, $\|C_t\| =\Omega \pa{\|C_0\|}$. Some of the lemmas are identical to those in \cite{Shimizu2025}; for this reason we omitted their proofs. The remaining lemmas needed to be adapted to the \dejavu protocol.

The following definition and lemmas are the primary mathematical tools used for the drift analysis of $\|C_t\|$.
\begin{definition}[Bernstein condition and one-sided Bernstein condition]
    \label{def:Bernstein condition}
    Let $D,s\ge 0$ be parameters.
    A random variable $X$ satisfies \emph{$(D,s)$-Bernstein condition} if,
    for any $\lambda\in \mathbb{R}$ such that $\abs{\lambda} D < 3$,
    $\bbE\qty[ \exp^{\lambda X} ]  \le \exp\qty( \frac{\lambda^2 s / 2}{1- (\abs{\lambda}D)/3})$.
    We say that $X$ satisfies \emph{one-sided} $(D,s)$-Bernstein condition if,
    for any $\lambda\geq 0$ such that $\lambda D < 3$,
    $\bbE\qty[ \exp^{\lambda X} ]  \le \exp\qty( \frac{\lambda^2 s / 2}{1- (\lambda D)/3})$.
    \end{definition}

% ========== Lemma 3.4 (statement only) ==========
\begin{lemma}[Lemma 3.4 (Closure properties of one-sided Bernstein) from \cite{Shimizu2025} ]
\label{lem:Bernstein condition}
        Let $X,Y$ be random variables. We have the following:
        \begin{enumerate}
            \renewcommand{\labelenumi}{(\roman{enumi})}
            \item \label{item:Bernstein condition for bounded random variables}
            If $\bbE[X]=0$ and $\abs{X}\leq D$ for some $D$, then $X$ satisfies $\qty(D,\Var\qty[X])$-Bernstein condition.
            \item \label{item:Bernstein condition 1}
            If $X$ satisfies $\qty(D,s)$-Bernstein condition, then
            $X$ satisfies $\qty(D',s')$-Bernstein condition for any
            $D' \ge D$ and $s' \ge s$.
            Similarly, if $X$ satisfies one-sided $\qty(D,s)$-Bernstein condition, then
            $X$ satisfies one-sided $\qty(D',s')$-Bernstein condition for any
            $D' \ge D$ and $s' \ge s$.
            \item \label{item:Bernstein condition 2}
            If $X$ satisfies $\qty(D,s)$-Bernstein condition, then $aX$ satisfies $(\abs{a}D,a^2s)$-Bernstein condition for any $a \in \bbR$.
            If $X$ satisfies one-sided $\qty(D,s)$-Bernstein condition, then $aX$ satisfies one-sided $(a D,a^2s)$-Bernstein condition for any $a \geq 0$.
            \item \label{item:dominated Bernstein condition}
            If $X$ satisfies one-sided $\qty(D,s)$-Bernstein condition and $Y\preceq X$,
            then $Y$ satisfies one-sided $\qty(D, s)$-Bernstein condition.
            In particular,
              if $X$ satisfies one-sided $\qty(D,s)$-Bernstein condition and $Y\leq  X$,
              then $Y$ satisfies one-sided $\qty(D, s)$-Bernstein condition.
            \item \label{item:BC for independent rvs}
            If a sequence of $n$ random variables $X_1,\ldots,X_n$ are independent and $X_i$ satisfies $\qty(D,s_i)$-Bernstein condition for $i\in [n]$,
            then $\sum_{i\in [n]}X_i$ satisfies $\qty(D,\sum_{i\in [n]}s_i)$-Bernstein condition.
            \item \label{item:BC for NA rvs}
            If a sequence of $n$ random variables $X_1,\ldots,X_n$ are negatively associated and $X_i$ satisfies one-sided $\qty(D,s_i)$-Bernstein condition for $i\in [n]$,
            then $\sum_{i\in [n]}X_i$ satisfies one-sided $\qty(D,\sum_{i\in [n]}s_i)$-Bernstein condition.
        \end{enumerate}
    \end{lemma}
% ========== Lemma 3.5 (statement only) ==========

\begin{lemma}[Lemma 3.5 (Additive drift under one-sided Bernstein) from \cite{Shimizu2025}]
\label{lem:Freedman stopping time additive}
    Let $(X_t)_{t\ge 0}$ be an adapted process and $\tau$ a stopping time. Suppose there exist $R\in\mathbb{R}$, $D,s>0$ such that for every $t$:
\begin{description}
\item[(C1)] $\mathbf{1}_{\{\tau>t-1\}}\big(\mathbf{E}_{t-1}[X_t]-X_{t-1}-R\big)\le 0$.
\item[(C2)] $\mathbf{1}_{\{\tau>t-1\}}\big(X_t-X_{t-1}-R\big)$ satisfies the $(D,s)$-one-sided Bernstein condition .
\end{description}
For a parameter $h>0$, define stopping times
\begin{align*}
    & \tau^+_X \coloneq  \inf\cpa{ t\geq 0 : X_t \geq X_0 +h }
    \\
    & \tau^-_X \coloneq  \inf\cpa{ t\geq 0 : X_t \leq X_0 -h }
\end{align*}
Then, we have the following:
\begin{enumerate}
\item Suppose $ R\geq 0$. Then, for any $h,T>0$ such that $z\coloneq  h - R \cdot T>0$, we have
\[\Pr\!\big[\tau^+_X \le \min\{T,\tau\}\big]\le \exp\!\pa{-\frac{z^2/2}{sT+(zD)/3}}.\]
\item Suppose $R<0$. Then, for any $h,T>0$ such that $z\coloneq  (-R) \cdot T - h >0$, we have
\[
\Pr\!\big[\min\{\tau^-_X,\tau\}>T\big]\le  \exp\!\pa{-\frac{z^2/2}{sT+(zD)/3}}..
\]
\end{enumerate}
\end{lemma}

% ========== Lemma 4.2 (with proof) ==========
In the following lemma, we show that some fundamental quantities satisfy the Bernstein condition. We will use our notation $p^{(t)}=C^{(t)}/n$.

\begin{lemma}[Adapted from Lemma 4.2 (One-sided Bernstein for the basic quantities) from \cite{Shimizu2025}]
    \label{lem:Bernstein_condition_for_sync_processes}
We have the following for any $t\geq 1$:
\begin{enumerate}
            \renewcommand{\labelenumi}{(\roman{enumi})}
    \item \label{item:BC for alpha}
    For any opinion $i\in [k]$, the random variable
    \[
      p_i^{(t)}-\bbE_{t-1}\pa{p^{(t)}(i)}
    \]
    conditioned on round $t-1$ satisfies the
    \[
      \qty(\frac{1}{n},\Var_{t-1}\qty[p_i^{(t)}])
    \]
    Bernstein condition.
    \item \label{item:BC for alphanorm}
    If $h \|p^{(t-1)}\|\leq 1$,
     $\|p^{(t-1)}\|^2-\|p^{(t)}\|^2$ conditioned on round $t-1$ satisfies one-sided \\
     $ \qty(\frac{2\|p^{(t-1)}\|}{n},\frac{c h^2\|p^{(t-1)}\|^4}{n})$-Bernstein condition.
     \item \label{item:BC for alphanorm 2}
     If $h \|p^{(t-1)}\|>1$,
     $\|p^{(t-1)}\|^2-\|p^{(t)}\|^2$ conditioned on round $t-1$ satisfies one-sided \\
     $ \qty(\frac{2\|p^{(t-1)}\|}{n},\frac{c \|p^{(t-1)}\|^2}{n})$-Bernstein condition.
\end{enumerate}
\end{lemma}

\begin{proof}[Proof of \cref{item:BC for alpha}]
      By definition,
      \[
        np_i^{(t)}=\sum_{v\in V}\mathbb{1}_{\text{opinion}_t(v)=i}.
      \]
      Hence
      \[
        p_i^{(t)}-\bbE_{t-1}\qty[p_i^{(t)}]=\sum_{v\in V}X_t(v),
      \]
      where
      \[
        X_t(v)\coloneq \frac{\mathbb{1}_{\text{opinion}_t(v)=i}-\bbE_{t-1}\qty[\mathbb{1}_{\text{opinion}_t(v)=i}]}{n}.
      \]
      Since $\abs{X_t(v)}\leq 1/n$ for all $v\in V$, $X_t(v)$ conditioned on round $t-1$ satisfies $\qty(\frac{1}{n},\Var_{t-1}[X_t(v)])$-Bernstein condition (\cref{item:Bernstein condition for bounded random variables} of \cref{lem:Bernstein condition}).
      Furthermore, since $(X_t(v))_{v\in V}$ conditioned on round $t-1$ are $n$ mean-zero independent random variables,
      $p_i^{(t)}-\bbE_{t-1}\qty[p_i^{(t)}]=\sum_{v\in V}X_t(v)$ satisfies $\qty(\frac{1}{n},\sum_{v\in V}\Var_{t-1}[X_v])$-Bernstein condition from \cref{item:BC for independent rvs} of \cref{lem:Bernstein condition}.
      Since
      \[
      \sum_{v\in V}\Var_{t-1}[X_t(v)]=\Var_{t-1}\qty[\sum_{v\in V}X_t(v)]=\Var_{t-1}\qty[p_i^{(t)}-\bbE_{t-1}\qty[p_i^{(t)}]]=\Var_{t-1}\qty[p_i^{(t)}],
      \]
      we obtain the claim.
  \end{proof}

  \begin{proof}[Proof of \cref{item:BC for alphanorm,item:BC for alphanorm 2}]
    We have
    \begin{align}
      \|p^{(t-1)}\|^2-\|p_{t}\|^2
        &=\sum_{i\in [k]}\cpa{{p^{(t-1)}_i}^2-{p^{(t)}_i}^2} \nonumber\\
        &\leq \sum_{i\in [k]}2p^{(t-1)}_i\cpa{p^{(t-1)}_i-p^{(t)}_i} & & \pa{\text{for all $x,y\in\bbR$, $x^2-y^2\le 2x(x-y)$}}\nonumber\\
        &\leq \sum_{i\in [k]}2p^{(t-1)}_i\cpa{\bbE_{t-1}\pa{p^{(t)}_i}-p^{(t)}_i} & & \pa{\text{by \cref{lemma:norm_2_expected_value_inequality}}}\nonumber\\
        &=\sum_{i\in [k]}Y_t(i), \nonumber %\label{eq:Bernstein alphanorm}
    \end{align}
    where \[
    Y_t(i)\coloneq  2p^{(t-1)}_i\cpa{\bbE_{t-1}\pa{p^{(t)}_i}-p^{(t)}_i}=\sum_{v\in V}\frac{2p^{(t-1)}_i}{n}\qty(\bbE_{t-1}[\mathbb{1}_{o_t(v)=i}]-\mathbb{1}_{o_t(v)=i}).
    \]
    $Y_t(i)$ conditioned on round $t-1$ satisfies $\qty(\frac{2p^{(t-1)}_i}{n},4{p^{(t-1)}_i}^2\Var_{t-1}[p^{(t)}(i)])$-Bernstein condition from \cref{item:Bernstein condition 2} of \cref{lem:Bernstein condition} and \cref{item:BC for alpha} of \cref{lem:Bernstein_condition_for_sync_processes}.
    Furthermore, from \(p^{(t-1)}_i\leq \|p^{(t-1)}\|\), \cref{item:Bernstein condition 1} of \cref{lem:Bernstein condition} implies that
    $Y_t(i)$ conditioned on round $t-1$ satisfies $\qty(\frac{2\|p^{(t-1)}\|}{n},4{p^{(t-1)}_i}^2\Var_{t-1}[p^{(t)}(i)])$-Bernstein condition.

    From \cref{lem:zero one lemma}, the random variables $(\mathbb{1}_{\text{opinion}_t(v)=i})_{i\in [k]}$ are negatively associated for each $v\in V$.
    From \cref{lem:concatenation of NA}, $\qty((\mathbb{1}_{\text{opinion}_t(v)=i})_{i\in [k]})_{v\in V}$, a sequence of $kn$ random variables, are also negatively associated.
    Since $Y_t(i)=h_i\qty((\mathbb{1}_{\text{opinion}_t(v)=i})_{v\in V})$, i.e., non-increasing functions of disjoint subsets of negatively associated random variables $\qty((\mathbb{1}_{\text{opinion}_t(v)=i})_{i\in [k]})_{v\in V}$, $(Y_t(i))_{i\in [k]}$ are negatively associated (\cref{lem:concatenation of NA}).
    Thus, from \cref{item:BC for NA rvs} of \cref{lem:Bernstein condition}, $\sum_{i\in [k]}Y_t(i)$ conditioned on round $t-1$ satisfies
    one-sided $\qty(\frac{2\|p^{(t-1)}\|}{n}, 4\sum_{i\in [k]}{p^{(t-1)}_i}^2\Var_{t-1}[p^{(t)}_i])$-Bernstein condition.
    From \cref{item:dominated Bernstein condition} of \cref{lem:Bernstein condition}, $\|p^{(t-1)}\|^2-\|p_{t}\|^2\leq \sum_{i\in [k]}Y_t(i)$ conditioned round $t-1$ also satisfies
one-sided $\qty(\frac{2\|p^{(t-1)}\|}{n},4\sum_{i\in [k]}{p^{(t-1)}_i}^2\Var_{t-1}[p^{(t)}_i])$-Bernstein condition.

    Since by \Cref{lemma:bound_variance} we have that $\Var_{t-1}\pa{p^{(t)}_i} \leq \frac{p^{(t-1)}_i}{n} \min\cpa{c \,h^2 \|p^{(t-1)}\|_2^2,1} \pa{ \frac{p^{(t-1)}_i}{\|p^{(t-1)}\|^2} + 1 }$, applying \cref{item:Bernstein condition 1} of \cref{lem:Bernstein condition} and that $\|p^{(t-1)}\|_3^3 \leq \|p^{(t-1)}\|^3$, $\|p^{(t-1)}\|_4^4 \leq \|p^{(t-1)}\|^4$, we obtain the claim.
  \end{proof}

  The following lemmas bound the stopping time until the Euclidean norm of the configuration decreases by a constant multiplicative factor.

\begin{lemma}[Adapted from Lemma 4.5 from \cite{Shimizu2025}]
  \label{lem:drift analysis for basic}
  Consider stopping times defined in \cref{def:stopping times}.
    For any $T>0$, we have
    \begin{align*}
        \Pr\qty[\tau^\downarrow\leq \min\cpa{T,\tau^\uparrow}]\le
          \begin{cases}
            \exp\qty(-\Omega\qty(\frac{n}{h^2 T + \|p^{(0)}\|^{-1}}))
            \quad & 
            \text{if } h \|p^{(0)}\| \leq 1
            \\
            \exp\qty(-\Omega\qty(\frac{n\|p^{(0)}\|^2}{T}))
            \quad & 
            \text{if } h \|p^{(0)}\| > 1
        \end{cases}
    \end{align*}
\end{lemma}
\begin{proof}[Proof of \cref{lem:drift analysis for basic}]
  Let $\tau = \tau^\uparrow$, $X_t=-\|p^{(t\wedge \tau)}\|^2$, and $R = 0$.
  From \cref{lemma:expected_value_norm2_inequality},
    \begin{align*}
        \mathbb{1}_{\tau>t-1}\qty(\bbE_{t-1}\qty[ X_t ]-X_{t-1}-R)
        &=\mathbb{1}_{\tau>t-1}\qty(\|p^{(t-1)}\|^2-\bbE_{t-1}\qty[\|p^{(t)}\|^2])
        \leq 0.
    \end{align*}
    Furthermore, from \cref{item:BC for alphanorm} of \cref{lem:Bernstein_condition_for_sync_processes}
    and \cref{item:Bernstein condition 1} of \cref{lem:Bernstein condition},
    the random variable
    \begin{align*}
        \mathbb{1}_{\tau>t-1}\qty( X_t - X_{t-1}-R)
        &=\mathbb{1}_{\tau>t-1}\qty(\|p^{(t-1)}\|^2-\|p^{(t)}\|^2)
    \end{align*}
    satisfies one-sided $\qty(O\qty(\frac{\|p^{(0)}\|}{n}),O\pa{\frac{ \min\cpa{h^2 \|p^{(0)}\|^{2},1} \|p^{(0)}\|^2}{n}})$-Bernstein condition.
    Here, we used $ \|p^{(t-1)}\| \le (1+c^\uparrow)\|p^{(0)}\| $ for $ t-1 < \tau $.

    Set
    \[
        \varepsilon\coloneq 
        \|p^{(0)}\|^2-\pa{(1-c^\downarrow)\|p^{(0)}\|}^2
        =
        \pa{2c^\downarrow-(c^\downarrow)^2}\|p^{(0)}\|^2.
    \]
    Applying \cref{lem:Freedman stopping time additive} with $D = O\qty(\frac{\|p^{(0)}\|}{n})$ and threshold \(\varepsilon\),
    \begin{align*}
        \Pr\qty[\tau_X^\uparrow\leq \min\{T,\tau\}]
        \leq 
        \begin{cases}
            \exp\qty(-\Omega\qty(\frac{\|p^{(0)}\|^4}{h^2\|p^{(0)}\|^4 T/n+\|p^{(0)}\|^{3}/n})) 
            \quad & 
            \text{if } h \|p^{(0)}\| \leq 1
            \\
            \exp\qty(-\Omega\qty(\frac{\|p^{(0)}\|^4}{\|p^{(0)}\|^2 T/n+\|p^{(0)}\|^{3}/n})) 
            \quad & 
            \text{if } h \|p^{(0)}\| > 1
        \end{cases}
    \end{align*}
    holds. Moreover, if $\tau^\downarrow\leq \min\{T,\tau\}$, then for some $t\leq \min\{T,\tau\}$ we have
    \[
        \|p^{(t)}\|\leq (1-c^\downarrow)\|p^{(0)}\|,
    \]
    and therefore
    \[
        X_t-X_0
        =
        \|p^{(0)}\|^2-\|p^{(t)}\|^2
        \geq
        \|p^{(0)}\|^2-\pa{(1-c^\downarrow)\|p^{(0)}\|}^2
        =
        \varepsilon.
    \]
    Hence \(\tau^\downarrow\leq \min\{T,\tau\}\) implies \(\tau_X^\uparrow\leq \min\{T,\tau\}\), and so
    \[
        \Pr\qty[\tau^\downarrow\leq \min\{T,\tau\}]
        \leq
        \Pr\qty[\tau_X^\uparrow\leq \min\{T,\tau\}],
    \]
    which yields the claim.
\end{proof}

\begin{proof}[Proof of \cref{lem:taunormdown is large}]
    For each $0\le s \le T$,
    let $\sigma^{\downarrow}_s = \inf\{ t\ge s \colon \|p^{(t)}\| \le (1-c^\downarrow) \|p_s\| \}$,
    $\sigma^{\uparrow}_s = \inf\{t \ge s \colon \|p^{(t)}\| \ge 2\|p_s\|\}$,
    and
    let $\mathcal{E}^{(s)}$ be the event that $\|p_s\| \ge \|p^{(0)}\|$ and $\sigma^{\downarrow}_s \le \min\{T, \sigma^{\uparrow}_s\}$.
    Note that $\tau^\downarrow = \sigma^{\downarrow}_0$ and $\tau^\uparrow = \sigma^{\uparrow}_0$ (for $c^\uparrow=1$).

    The key observation is that the partial process $(p^{(t)})_{t\ge s}$ is again a \dejavu process and $\sigma_t^{\uparrow},\sigma_t^{\downarrow}$ can be seen as the stopping times of \cref{def:stopping times} for the partial process.
    Moreover, the event $\mathcal{E}^{(s)}$ depends only on the partial process $(p^{(t)})_{t\ge s}$.
    Therefore, from \cref{lem:drift analysis for basic}, we have
    \begin{align*}
        \Pr_{(p^{(t)})_{t\ge s}}\qty[ \mathcal{E}^{(s)} ] &\le \Pr_{(p^{(t)})_{t\ge s}}\qty[ \sigma_s^{\downarrow} \le \min\{T, \sigma^{\uparrow}_s\} \middle| \|p^{(s)}\| \ge \|p^{(0)}\| ] \\
        &\le \begin{cases}
            \exp\qty(-\Omega\qty(\frac{n}{h^2 T + \|p^{(0)}\|^{-1}}))
            \quad & 
            \text{if } h \|p^{(0)}\| \leq 1
            \\
            \exp\qty(-\Omega\qty(\frac{n\|p^{(0)}\|^2}{T}))
            \quad & 
            \text{if } h \|p^{(0)}\| > 1
        \end{cases}
    \end{align*}

    If $\tau^\downarrow \le T$ occurs, then $\mathcal{E}^{(s)}$ occurs for some $0\le s \le T$.
    For example, if $s\le \tau^\downarrow$ is the round such that $\|p_s\| = \max_{0\le t \le \tau^\downarrow} \|p^{(t)}\|$, then $\mathcal{E}^{(s)}$ holds.
    Therefore, we have
    \begin{align*}
        \Pr\qty[ \tau^\downarrow \le T ] &\le \Pr\qty[ \bigvee_{0\le s \le T} \mathcal{E}^{(s)} ] \\
        &\le \sum_{0\le s \le T} \Pr\qty[ \mathcal{E}^{(s)}] \\
        &\le \begin{cases}
            T\exp\qty(-\Omega\qty(\frac{n}{h^2 T + \|p^{(0)}\|^{-1}}))
            \quad & 
            \text{if } h \|p^{(0)}\| \leq 1
            \\
            T\exp\qty(-\Omega\qty(\frac{n\|p^{(0)}\|^2}{T}))
            \quad & 
            \text{if } h \|p^{(0)}\| > 1
        \end{cases}
    \end{align*}
\end{proof}
In the following, we remark why \Cref{lem:taunormdown is large} implies that $\|p^{(t)}\| \geq c^\downarrow \|p^{(0)}\|$ w.h.p. 
In the regime $h \leq  \frac{1}{\sqrt{p_1} }$, by \Cref{thm:convergence_time_dejavu}, $T_d = O\pa{ \frac{\log n}{h^2 p_1}}$. Hence, whenever \(p_1=\omega\pa{\log^2 n/n}\), \Cref{lem:taunormdown is large} implies that w.h.p. \(\|p^{(t)}\| \geq c^\downarrow \|p^{(0)}\|\) for all \(0\le t\le T_d\).

In the regime $ \frac{1}{\sqrt{p_1}} \leq h \leq \frac{1}{\|p\|\log n}$ (which exists only for initial unbalanced configurations), by \Cref{thm:convergence_time_dejavu}, it holds $T_d = O\pa{\log n }$. Therefore, \Cref{lem:taunormdown is large} implies that w.h.p. $\|p^{(t)}\| \geq c^\downarrow \|p^{(0)}\|$.

In the regime $h >  \frac{1}{\|p\|\log n} , \frac{1}{\sqrt{p_1}}$, by \Cref{thm:convergence_time_dejavu}, it holds $T_d = O\pa{\log n }$, so we distinguish two cases. If $\|p^{(0)}\|^2 = \omega\pa{\frac{\log^2 n}{n}}$, for some sufficiently large constant $C>0$, \Cref{lem:taunormdown is large} implies that w.h.p. $\|p^{(t)}\|^2 = \Omega\pa{ \|p^{(0)}\|^2 }$ for all $0 \leq t \leq T_d$. Instead, if $\|p^{(0)}\|^2 \leq \frac{C \log^2 n}{n}$, we have that $\|p^{(t)}\|^2 = \Omega\pa{ \frac{\|p^{(0)}\|^2}{\log^2 n} }$ for all $t\geq 0$, as it always holds $\|p^{(t)}\|\geq \frac{1}{k} \geq \frac{1}{n}$.

\subsection{Comparing the number of samples required by \texorpdfstring{\dejavu and \hmaj}{DéjàVu and h-majority}}\label{sec:comparison-samples}
We have all the ingredients to conclude the proof of \Cref{thm:dejavu_samples,thm:sample-efficiency}, which we restate for convenience.
\begin{theorem*}
  Let \(C=(C_1,\ldots,C_k)\) be a system configuration such that \(C_1\geq \cdots \geq C_k\), \(C_1 = \omega\pa{\log^2 n}\), and that, for a large enough constant \(\lambda > 0\),
  \[
    C_1 - C_2 \ge \lambda \sqrt{ \max\cpa{\frac{n}{h^2}, C_1} \log n } .
  \]
  Let \(S_{d}\) and \(S_{m}\) be the numbers of samples until consensus of, respectively, \dejavu and \hmaj.
  Fix any arbitrarily small constant \(\varepsilon > 0\).
  For \(h = \Omega({\min\{n^{3/4+\varepsilon}/{C_1}, \sqrt{n/C_1}\}})\), w.h.p. it holds
  \begin{align*}
      \begin{cases}
          & S_{d} \cdot \frac{O\left(\max\{1,h\frac{\norm{C}_2}{n}\right)\}}{\log^3 n}\leq S_{m} \text{ if } \norm{C}_2 = O(\sqrt{n} \log n) \text{ and } h \norm{C}_2 \ge \frac{n}{\log n}, \\
          & S_{d} \cdot \frac{O\left(\max\{1,h\frac{\norm{C}_2}{n}\right)\}}{\log n}\leq S_{m} \text{ otherwise}.
      \end{cases}
  \end{align*}
\end{theorem*}
\begin{proof}[Proof of \Cref{thm:dejavu_samples}]
    Let $T_d$ be the convergence time of \dejavu. By \Cref{thm:convergence_time_dejavu}, there exists an absolute constant \(c_T>0\) such that
    \[
        \Pr\!\left(
            T_d \le c_T \max\cpa{\frac{n}{h^2 C_1},1}\log n
        \right)
        \ge 1-n^{-\Theta(1)}.
    \]
    Define
    \[
        T_*\coloneq c_T \max\cpa{\frac{n}{h^2 C_1},1}\log n.
    \]
    On the event
    \[
        \cpa{T_d\le T_*}
        \cap
        \cpa{\|p^{(t)}\|\ge c^\downarrow \|p^{(0)}\| \text{ for all }0\le t\le T_*},
    \]
    let
    \[
        m\coloneq \min\cpa{h,\left\lceil 1/\|p^{(0)}\| \right\rceil}.
    \]
    By \Cref{lemma:tight_bound_probability_collision}, for every \(t\le T_*\),
    \[
        \Pr(H\le m\mid C^{(t)})
        \ge 2^{-11}\min\cpa{m^2\|p^{(t)}\|^2,1}
        \ge 2^{-11}\min\cpa{(c^\downarrow)^2,1}
        =:\kappa>0.
    \]
    Therefore, the number of samples performed by a fixed agent in round \(t\) is stochastically dominated by \(mG_t\), where \(G_t\sim\mathrm{Geom}(\kappa)\): indeed, every fresh block of \(m\) samples contains an internal repeat with conditional probability at least \(\kappa\), and such an internal repeat already stops the round.

    Consequently, on the same event,
    \[
        S_d \preceq m\sum_{t=1}^{T_*} G_t.
    \]
    Since \(\kappa\) is an absolute constant, \Cref{thm:chernoff_geometric} yields
    \[
        \sum_{t=1}^{T_*} G_t = O(T_*)
    \]
    w.h.p. Combining this estimate with the high-probability event \(T_d\le T_*\) and with the norm lower bound from \Cref{lem:taunormdown is large} for \(T=T_*\), a union bound implies that, w.h.p.,
    \[
        S_d
        = O\!\left(
            T_* \cdot \min\cpa{h,\frac{1}{\|p^{(0)}\|}}
        \right)
        =
        O\!\left(
            \max\cpa{\frac{n}{h^2 C_1},1}\log n
            \cdot
            \min\cpa{h,\frac{n}{\|C\|}}
        \right)
    \]
    w.h.p., which is the claim.
\end{proof}

We restate \cref{thm:sample-efficiency} for convenience.

\begin{theorem*}
    Let \(C=(C_1,\ldots,C_k)\) be a system configuration such that \(C_1\geq \cdots \geq C_k\), \(C_1 = \omega\pa{\log^2 n}\), and bias
    \[
        C_1 - C_2 = \Omega\pa{ \sqrt{ \max\cpa{\frac{n}{h^2}, C_1} \log n } }.
    \]
    Let \(S_d\) and \(S_m\) be the numbers of samples until consensus of, respectively, \dejavu and \hmaj.
    Fix any arbitrarily small constant \(\varepsilon > 0\).
    For \(h = \Omega({\min\{n^{3/4+\varepsilon}/{C_1}, \sqrt{n/C_1}\}})\), w.h.p. it holds
    \[
        S_d \cdot \frac{O\left(\max\{1,h\frac{\norm{C}_2}{n}\right)\}}{\log n}\leq S_m~.
    \]
\end{theorem*}
\begin{proof}[Proof of \Cref{thm:sample-efficiency}]
   Let $T_m$ be the convergence time of the \hmaj. By \Cref{thm:app:lb:h-majority} and by the trivial lower bound $T_m \ge 1$, we obtain
   \[
       S_m= T_m \cdot h =\Omega 
       \begin{cases}
           \max\cpa{ \frac{n}{h\,C_1}, h } &\text{if } h=\Omega\pa{ n^{3/4+\varepsilon}/C_1 }
           \\
           h &\text{otherwise }
       \end{cases}
   \]
   Let's first assume that $C_1\leq \sqrt{n}$. Consequently, we have that
   \[
       \min\cpa{\frac{n^{3/4+\varepsilon}}{C_1}, \sqrt{\frac{n}{C_1}}} =  \sqrt{\frac{n}{C_1}}.
   \]
   Then we assume that $h\geq \sqrt{n/C_1}$.
   To find the ratio $\frac{S_d}{S_m}$, we analyze the two main regimes of $h$.

   \textbf{Low sampling regime $\pa{h \le \frac{n}{\norm{C}_2}}$}.
    The ratio is:
\[
    \frac{S_d}{S_m} = O\left( \frac{\log n \cdot h}{h} \right) =O(\log n).
\]

    \textbf{Large sampling regime $\pa{h > \frac{n}{\norm{C}_2}}$}.
    The ratio is:
    \[
        \frac{S_d}{S_m} = O\pa{\frac{ \log n \frac{n}{\norm{C}_2} }{h}}.
    \]
    If instead, $C_1 > \sqrt{n}$, assuming $h\geq n^{3/4+\varepsilon}/C_1 $, we obtain the same results similarly.
\end{proof}

\section{Conclusion and Open Questions}
\label{sec:conclusions}

We conclude by discussing some limitations of our work and several directions for future research.

A complete analysis of the convergence time of \hmaj for arbitrary values of \(h\) and \(k\) remains a major open problem, since the regime \(h \le k\) is still missing. Our analysis of \dejavu provides several tools, in particular the Poisson race framework and the monotonicity of the probability ratio via Newton's inequalities, that may be useful for attacking this question.

As discussed in \cref{ssec:contribution}, our analysis requires an initial additive bias of
\[
  \Omega\!\left(\sqrt{\max\cpa{n/h^2, C_1} \log n}\right).
\]
For constant \(h\), this is comparable to the \(\omega(\sqrt{n\log n})\) bias required by the \threemaj dynamics.
We conjecture that this requirement can be significantly reduced for \dejavu. More precisely, we believe that the minimum bias needed for plurality consensus is always \(\omega(\sqrt{C_1 \log n})\), matching the optimal bias of the \twochoices dynamics (see \cref{sec:related}).
The intuition is that \dejavu can be viewed as a generalization of \twochoices in which more samples are allowed, and we see no fundamental reason why its bias requirement should deteriorate as \(h\) grows. By contrast, we expect that \hmaj genuinely requires a larger bias in some regimes.

Finally, we stress that \dejavuh is specifically designed for \emph{plurality consensus}: it amplifies and preserves the initial majority opinion. Without sufficient initial bias, the protocol can be slow to converge because of its conservative behavior: when no opinion is repeated within the first \(h\) samples of a round, the agent keeps its current opinion. This is analogous to the comparison between \twochoices and \threemaj in \cite{BerenbrinkCEKMN17}, where \threemaj is shown to converge much faster than \twochoices when the number of initial opinions approaches \(n\).

One can also consider a variant of \dejavuh in which, when no opinion is repeated within the first \(h\) samples, the agent updates to one of the sampled opinions chosen uniformly at random. This variant may be more effective for achieving general consensus when the number of initial opinions is extremely large and there is no initial bias, and it appears to be an interesting direction for future work.

\section*{Acknowledgments}

This work has been supported the French government National Research Agency (ANR) through the UCA JEDI (ANR-15-IDEX-01) and EUR DS4H (ANR-17-EURE-004), through the 3IA Cote d'Azur Investments in the project with the reference number ANR-23-IACL-0001, and by the AID INRIA-DGA project n°2023000872 “BioSwarm".
Support was received also by the MUR (Italy) Department of Excellence 2023 -
2027 for GSSI, and the Agence Nationale de la Recherche (ANR) under project ByBloS (ANR-20-CE25-0002). 
Francesco d’Amore is supported by the project Decreto MUR n. 47/2025, CUP: D13C25000750001.
{\small
\bibliographystyle{alpha}
\bibliography{biblio}
}

\appendix

\section*{APPENDIX}
\section{Tools} \label{sec:tools}
\begin{theorem}[Also in \cite{enwiki:1337666560}]
  \label{thm:poisson_approximation}
  Let \(\cpa{X_i}_{i\in[n]}\)  be \(n\) mutually independent \( \text{Poisson}(\lambda_i)\) random variables. Conditioning on the sum \(S = \sum_{i\in[n]} X_i = s\), we have

  \[
    \pa{\pa{X_i}_{i\in[n]} \mid S = s} \sim \text{Multinomial}( p, s)
  \]
  where \(p = \pa{\frac{\lambda_i}{\sum_{j\in[n]} \lambda_j}}_{i\in[n]}\).

\end{theorem}

\begin{definition}[Negative association]
  \label{def:NA}
  Random variables \(X_1,\ldots,X_n\) are \emph{negatively associated} if
  for every two disjoint index sets \(I,J\subseteq [n]\),
  \[
    \bbE\qty[f(X_i,i\in I)g(X_j,j\in J)]\leq \bbE\qty[f(X_i,i\in I)]\bbE\qty[g(X_j,j\in J)]
  \]
  for all functions \(f:\mathbb{R}^I\to \mathbb{R}\) and \(g:\mathbb{R}^J\to \mathbb{R}\)
  that are both non-decreasing.
\end{definition}
\begin{lemma}[Lemma 2 of \cite{Dubhashi1998-cf}]
  \label{lem:product of NA}
  Let \(X_1,\ldots,X_n\) be a sequence of negatively associated random variables.
  Then for any non-decreasing functions \(f_i\), \(i\in [n]\),
  \[
    \bbE\qty[\prod_{i\in [n]}f_i(X_i)]\leq \prod_{i\in [n]}\bbE\qty[f_i(X_i)].
  \]
\end{lemma}
\begin{lemma}[Lemma 8 of \cite{Dubhashi1998-cf}]
  \label{lem:zero one lemma}
  Let \(X_1,\ldots,X_n\) be random variables taking values in \(\cpa{0,1}\) such that \(\sum_{i\in[n]}X_i=1\).
  Then \(X_1,\ldots,X_n\) are negatively associated.
\end{lemma}
\begin{proposition}[Proposition 7 of \cite{Dubhashi1998-cf}]
  \label{lem:concatenation of NA}
  We have the following:
  \begin{enumerate}
    \item Let \(X_1,\ldots, X_n\) and \(Y_1,\ldots, Y_n\) be two sequences of negatively associated random variables that are mutually independent.
          Then \(X_1,\ldots, X_n, Y_1,\ldots, Y_n\) are negatively associated.
    \item Let \(X_1,\ldots, X_n\) be a sequence of negatively associated random variables.
          Let \(I_1,\ldots, I_k\) be disjoint index sets for some \(k\).
          For \(j\in [k]\), let \(h_j:\mathbb{R}^{I_j}\to \mathbb{R}\) be functions that are all non-decreasing or all non-increasing, and define \(Y_j:= h_j(X_i,i\in I_j)\).
          Then, \(Y_1,\ldots,Y_k\) are negatively associated.
          That is, non-decreasing (or non-increasing) functions of disjoint subsets of negatively associated random variables are also negatively associated.
  \end{enumerate}
\end{proposition}

\begin{lemma}[Multiplicative forms of Chernoff bounds \cite{dubhashi2009}]\label{lemma:app:multiplicative-chernoff}
  Let \(X_1, \ldots, X_n\) be independent binary random variables.
  Let \(X = \sum_{i = 1}^n X_i\) and \(\mu = \bbE[X]\).
  Then:
  \begin{enumerate}
    \item For any \(\delta \in (0,1)\) and any \(\mu \le \mu_+ \le n\), it holds that
          \[
            \Pr(X \ge (1+\delta) \mu_+) \le \exp(- \delta^2 \mu_+ / 3).
          \]
    \item For any \(\delta \in (0,1)\) and any \(0 \le \mu_- \le \mu\), it holds that
          \[
            \Pr(X \le (1-\delta) \mu_-) \le \exp(- \delta^2 \mu_- / 2).
          \]
  \end{enumerate}
\end{lemma}

\begin{lemma}[Hoeffding bounds \cite{mitzenmacher2005}]\label{lemma:app:hoeffding}
  Let \(a < b\) be two constants, and \(X_1, \ldots, X_n\) be independent random variables such that \(\Pr(a \le X_i \le b) = 1\) for all \(i \in [n]\).
  Let \(X = \sum_{i = 1}^n X_i\) and \(\mu = \bbE[X]\).
  Then:
  \begin{enumerate}
    \item For any \(t > 0\) and any \(\mu \le \mu_+\), it holds that
          \[
            \Pr(X \ge \mu_+ + t) \le \exp(-\frac{2t^2}{n(b-a)^2}).
          \]
    \item For any \(t > 0\) and any \(0 \le \mu_- \le \mu\), it holds that
          \[
            \Pr(X \le \mu_- - t) \le \exp(-\frac{2t^2}{n(b-a)^2}).
          \]
  \end{enumerate}
\end{lemma}

% \begin{lemma}[Bernstein Inequality \cite{dubhashi2009}]
%   \label{lemma:bernstein_inequality}
%   Let \( X_1, X_2, \dots, X_n \) be i.i.d.\ random variables such that
%   \(\abs{X_i - \mathbb{E}[X_i]} \leq C\) and \(\abs{\bbE[X_i]} \le C\) for some \(C > 0\).
%   Define the sum of centered variables
%   \(
%   S_n = \sum_{i=1}^n (X_i - \mathbb{E}[X_i]).
%   \)
%   Then, for all \( t > 0 \),
%   \[
%     \Pr(\abs{S_n} \geq t) \leq 2 \exp\left( -\frac{t^2}{2 n \Var(X_1) + \frac{2Ct}{3}} \right).
%   \]
% \end{lemma}

\begin{lemma}[Bernstein inequality for independent bounded variables \cite{boucheron2003concentration}]
  \label{lemma:bernstein_non_identical}
  Let \(X_1,\dots,X_n\) be independent random variables such that
  \(
  \mathbb{E}[X_i]=0
  \) and \(
  |X_i|\le b_i
  \)
  Denote
  \[
    S_n = \sum_{i=1}^n X_i,
    \qquad
    V = \sum_{i=1}^n \mathrm{Var}(X_i),
    \qquad
    B = \max_{1\le i\le n} b_i .
  \]
  Then for every \(t>0\),
  \[
    \mathbb{P}(S_n \ge t)
    \le
    \exp\!\left(
    -\frac{t^2}{2V + \frac{2}{3} B t}
    \right).
  \]
  The same bound holds for \(\mathbb{P}(S_n \le -t)\).
\end{lemma}

\begin{theorem}[From \cite{JANSON20181}]
\label{thm:chernoff_geometric}
Let $X = \sum_{i=1}^n X_i$ be a sum of independent geometric random variables $X_i \sim \text{Geom}(p_i)$ with $0 < p_i \le 1$. Let $\mu := \mathbb{E}X = \sum_{i=1}^n 1/p_i$ and let $p_* := \min_i p_i$.
For any $\lambda \ge 1$,
\begin{equation*}
    \pr{X \ge \lambda \mu} \le \exp\pa{-p_* \mu (\lambda - 1 - \ln \lambda)}.
\end{equation*}
\end{theorem}

\begin{lemma}[Newton's inequality for elementary symmetric polynomials \cite{enwiki:1334087247}]
  \label{lem:newton_ineq}
  Let \(x_1, x_2, \dots, x_n\) be non-negative real numbers. For \(k \in \{1, \dots, n-1\}\), let \(e_k(\mathbf{x})\) denote the \(k\)-th elementary symmetric polynomial:
  \[
    e_k(\mathbf{x}) = \sum_{1 \le i_1 < i_2 < \dots < i_k \le n} x_{i_1} x_{i_2} \dots x_{i_k}.
  \]
  The normalized elementary symmetric means \(S_k = e_k / \binom{n}{k}\) satisfy
  \[
    S_k^2 \ge S_{k-1} S_{k+1}.
  \]
\end{lemma}
\begin{corollary}
  The following inequality holds:
  \[
    e_k(\mathbf{x})^2 \ge e_{k-1}(\mathbf{x}) \cdot e_{k+1}(\mathbf{x}).
  \]
\end{corollary}
\begin{proof}
  By \Cref{lem:newton_ineq},
  \[
    \left( \frac{e_k}{\binom{n}{k}} \right)^2
    \ge
    \frac{e_{k-1}}{\binom{n}{k-1}} \cdot \frac{e_{k+1}}{\binom{n}{k+1}}.
  \]
  Rearranging gives
  \[
    e_k^2 \ge
    \frac{\binom{n}{k}^2}{\binom{n}{k-1}\binom{n}{k+1}} \, e_{k-1} e_{k+1}.
  \]
  For every \(1 \le k < n\),
  \[
    \frac{\binom{n}{k}^2}{\binom{n}{k-1}\binom{n}{k+1}}
    =
    \frac{(k+1)(n-k+1)}{k(n-k)}
    > 1.
  \]
  Therefore \(e_k^2 \ge e_{k-1} e_{k+1}\), as claimed.
\end{proof}

\end{document}